\documentclass[a4paper,12pt]{article}
\usepackage{amsmath,amsthm,amsfonts,amssymb,bm,mathrsfs}
\usepackage[protrusion=true,expansion=true]{microtype}
\usepackage{graphicx}
\usepackage{enumerate}
\usepackage{wrapfig}
\usepackage{bibentry,natbib}
\usepackage[T1]{fontenc}
\usepackage[section]{placeins}
\usepackage{footmisc}
\usepackage{sgame}
\usepackage{color}
\usepackage{multicol}
\usepackage{tikz}
\usepackage{subfloat}
\usetikzlibrary{decorations.pathreplacing}
\usepackage{subcaption}
\usepackage{fancyref}
\usepackage[colorlinks=true, linkcolor=blue, urlcolor=blue, citecolor=blue]%
{hyperref}
\usepackage{amsmath}
\usepackage{amsfonts}
\usepackage{amssymb}
\usepackage{lmodern}
\usepackage{enumitem}
\usepackage{setspace}

\usepackage{appendix}
\usepackage{multibib}
\newcites{Appendix}{References for Online Appendix}%

\providecommand{\U}[1]{\protect\rule{.1in}{.1in}}
\allowdisplaybreaks
\usetikzlibrary{matrix,arrows,decorations.pathmorphing}

\newtheorem{theorem}{Theorem}

\newtheorem{definition}{Definition}
\newtheorem{example}{Example}

\newtheorem{lemma}{Lemma}

\newtheorem{proposition}{Proposition}
\newtheorem{corollary}{Corollary}

\newcommand{\cM}{\mathcal{M}}

\newcommand{\RNum}[1]{\uppercase\expandafter{\romannumeral #1\relax}}

\newcommand{\Rnum}[1]{\lowercase\expandafter{\romannumeral #1\relax}}

\DeclareMathOperator*{\argmax}{arg\,max}
\DeclareMathOperator*{\argmin}{arg\,min}

\usepackage{geometry}
\newcommand{\nc}{\newcommand}
\nc{\jl}[1]{\textcolor{blue}{[Jiangtao: #1]}}
\nc{\tr}[1]{\textcolor{blue}{[Rui: #1]}}
\nc{\mz}[1]{\textcolor{red}{[Zhang: #1]}}

\begin{document}
	
\title{\Large \textbf{Associative Networks in Decision Making}\thanks{We are grateful to the co-editors and three anonymous referees, whose comments and suggestions have substantially improved the paper. We also thank David Ahn, Paul Cheung, David Dillenberger, Pawel Dziewulski, Amanda Friedenberg, Faruk Gul, Marina Halac, Junnan He, Navin Kartik, Matthew Kovach, Jay Lu,  Yusufcan Masatlioglu, Paulo Natenzon, Pietro Ortoleva, Erkut Ozbay, Wolfgang Pesendorfer, John K.-H. Quah, Satoru Takahashi, Leeat Yariv, and Junjie Zhou for helpful discussions. Li and Tang acknowledge the financial support from the Lee Heng Fellowship of the Hong Kong University of Science and Technology. Part of this work was completed during Zhang’s visit to the Cowles Foundation and Yale University. Zhang gratefully acknowledges their hospitality. All remaining errors are our own.}}
\author{
Jiangtao Li\thanks{Department of Economics, The Hong Kong University of Science and Technology, jiangtaoli@ust.hk}
\and Rui Tang\thanks{Department of Economics, The Hong Kong University of Science and Technology, ruitang@ust.hk}
\and Mu Zhang\thanks{Department of Economics, University of Michigan, muzhang@umich.edu}
}


\date{\today}

\maketitle

\thispagestyle{empty}

\begin{abstract}

We present a model of associative networks that captures how decision makers expand their consideration set through mental associations between alternatives. Our model provides a tractable approach to study how associations shape choice when some alternatives are available and others are merely observable but unavailable. We characterize the model within a random attention framework and demonstrate unique identification of all  parameters. This framework delivers a unified account of several prominent choice anomalies, including classic menu effects and their ``phantom'' counterparts. We  illustrate how associative links serve as a strategic variable in applications such as branding, imitation, and platform design.

\medskip

\textit{Keywords}: associative network, random attention, consideration set, random choice, availability and observability  

\medskip

\textit{JEL}: D01, D91

\end{abstract}

\newpage

\pagenumbering{arabic} 





\baselineskip=18pt
	
\section{Introduction}

\label{sec:intro}

Memory and attention are fundamental cognitive processes essential for decision making \citep{qje1955behavioral, book1993adaptive, jep1997behavioral}. The impact of these processes on decision making has been extensively studied in recent years \citep{qje2020bordalo, are2022salience}. Among the various memory patterns, mental association plays a key role, linking the recall of one item to another based on an individual's prior experience or learning.  In this paper, we adopt a choice-theoretical approach to explore the impact of mental association on decision making and develop a choice model to capture the effects of this cognitive process.  

	
To fix ideas and highlight some of the motivations behind our analysis, we consider the recent launch of the Xiaomi SU7, which has sparked considerable online debate due to its striking resemblance to the Porsche Taycan.\footnote{Official images are available on the manufacturers' websites for the \href{https://www.mi.com/global/discover/article?id=3095&ref=upstract.com}{Xiaomi SU7} and the \href{https://www.porsche.com/pap/_hong-kong_/accessoriesandservice/exclusive-manufaktur/uniqueness/taycan/}{Porsche Taycan}.}
The Xiaomi SU7 has drawn comparisons to the Porsche Taycan from internet users, with some even dubbing it the ``Mi Porsche.'' This strategic move by Xiaomi to boost its brand recognition through the association with Porsche is evident, and it is easy to understand why this could be valuable for the company. Potential buyers who admire the exterior design of the Porsche Taycan but find it prohibitively expensive might turn to the Xiaomi SU7 as an alternative. It might not be immediately obvious why Porsche would not be concerned about the release of the Xiaomi SU7. Indeed, in an interview during the 4th China International Consumer Goods Fair, Porsche China President and CEO Michael Kirsch addressed the issue for the first time, stating:``As for the similarities between the Xiaomi SU7 and Porsche, I think it's probably that good design always has something in mind.''\footnote{See \href{https://news.futunn.com/en/flash/16955274/porsche-china-president-responds-to-michele-for-the-first-time?level=1&data_ticket=1724900751818322}{link} for a news report on this.}

Our model of associative networks offers one possible explanation for this. For some buyers, factors such as premium branding outweigh cost in their decision-making process. While many models are affordable, these buyers might not initially consider all available options due to limited attention. During the initial launch period of the Xiaomi SU7, the press conference, news coverage, and the controversy surrounding its resemblance to the Porsche Taycan likely drew significant attention to the Xiaomi SU7. This attention could lead consumers to also consider the Porsche Taycan. Even though this association was not intentionally created by Porsche, it could increase the chance that the Porsche Taycan---but perhaps not other competitors like Maserati---ends up in consumers' final consideration set. Ultimately, this could work to Porsche's advantage.\footnote{In Section \ref{subsec:imitation}, we formalize this intuition within a strategic entry game to characterize when imitation is profitable for the entrant, and under what conditions the resulting associative attention spillovers are  substantial enough that the incumbent optimally accommodates the imitation rather than pursuing legal enforcement.}

	
Mental association is a cognitive process that allows the decision maker (DM) to expand her consideration set by linking relevant alternatives that may not have been initially considered. We represent this cognitive process through associative networks, a conceptual model first introduced in cognitive psychology \citep{book1973memory, book1983architecture, PR1981search} and widely applied in the marketing literature \citep{JM1993brand, PM2010brand, IJMR2011brand, JCR2015brand}. In an associative network, objects (nodes) are connected based on their semantic or conceptual relationships. When a specific node or input is activated, the network retrieves related nodes by spreading activation through the links. In our study, we employ the associative network as a descriptive model that captures how the attention to one alternative can trigger the DM to consider another alternative, abstracting away from the underlying conceptual similarities between alternatives that result in the association. More specifically, in our model, a link from alternative $x$ to $y$  indicates that the attention to $x$ can prompt the DM to further consider $y$.


Notably, the process of mental association is not limited to available alternatives; it can also be triggered by \textit{unavailable but observable} alternatives. In the case of Xiaomi SU7 versus Porsche Taycan, even before the Xiaomi SU7 was officially launched and became available for purchase, the attention it received could still prompt consideration of the Porsche Taycan. The importance of including unavailable but observable alternatives into the choice-theoretical framework has been highlighted by recent experimental studies. For instance, \cite{soltani2012range} demonstrate that decoy alternatives can induce choice reversals, such as the attraction effect and the compromise effect, even when these alternatives are not available for selection.\footnote{\cite{Science2015irrationality} report similar choice reversals in the mate choice of T{\'u}ngara frogs.} In an experimental design with differentiated products and objective payoffs, \cite{ee2021relevance} find that the presence of unavailable alternatives can lead to suboptimal decisions and longer decision times, with participants willing to pay significant amounts to avoid being exposed to these alternatives.

Throughout the paper, we study a novel choice domain in which not all observable alternatives are available. Each menu is composed of two distinct sets of alternatives and is represented as a pair $(A, B)$. While all the alternatives in $A$ and $B$ are observable, only the alternatives in $A$ are available. Choice scenarios are abundant where some alternatives are  unavailable but observable. A particularly salient example arises when goods are out of stock or available only in limited quantities. Limited-edition sneakers, luxury handbags, or newly released electronics are often highly publicized and prominently displayed even after they are no longer obtainable through standard channels. Consumers frequently encounter product pages with ``sold out'' labels or displays of waitlisted items in both physical and digital storefronts. These goods remain visible and salient despite being temporarily or effectively inaccessible. Unavailability can also arise for other reasons. Some goods are financially out of reach, making them observable but unattainable for many buyers. Geographic constraints can restrict access to region-specific goods or streaming content, even when such items are widely discussed online. Regulatory and legal barriers limit access to certain products---such as prescription medications---regardless of their visibility. Access may also be restricted based on consumer status, such as subscription-only content, invitation-only events, or loyalty-gated luxury products.

Although the DM cannot choose any alternative from $B$, the presence of these unavailable alternatives can still influence the DM's attention through association, thereby affecting her choices. By incorporating unavailable but observable alternatives, our primitive---a random choice rule---is a function that maps every menu $(A,B)$ to a distribution over $A$ and a default option. This distribution represents the DM's choice probabilities of alternatives in $A$ when presented with the menu $(A, B)$. 

Section \ref{sec:model} introduces our choice model. When presented with a menu $(A,B)$, the DM initially considers a random subset of alternatives from $A \cup B$, which we refer to as her initial consideration set. Following \cite{ecta2014consideration} (henceforth MM14), we assume that each alternative has a fixed probability of being initially considered by the DM, and that the DM directs her attention to each alternative independently. The DM then associates relevant alternatives in $A\cup B$ with those in her initial consideration set, a process captured by the DM's associative network, which is represented by a directed graph over the alternatives. Each link in this associative network is an ordered pair of alternatives $(x,y)$, indicating that the consideration of $x$ prompts the DM to further consider $y$. The association process continues until no further alternatives in $A \cup B$ can be linked to those already considered, resulting in a final consideration set $C$. The DM then selects her most preferred alternative among the available alternatives in $C$---she chooses her most preferred alternative in $A\cap C$ if it is not empty; otherwise, the default option is selected. We refer to the random choice rule induced by this choice procedure as the \textbf{A}ssociation \textbf{B}ased \textbf{C}onsideration rule (ABC). 

Section \ref{sec:characterization} presents the axioms that characterize ABCs. These axioms separately address the underlying attention distribution, the association procedure, the associative network, and the revealed preference relation within our choice model. Specifically: 
\begin{itemize} \vspace{-1ex}
\item Axiom 1 specifies the attention distribution: The DM  directs her initial attention to each alternative independently. \vspace{-1ex}
\item Axiom 2 states that if an unavailable but observable alternative $x$ is revealed to prompt the DM to consider some available alternative in a given menu, then making $x$ available does not affect the choice probability of the default option. This axiom implies that the DM's association process depends solely on the observability of alternatives, not their availability. Consequently, as long as  attention to $x$ leads to the consideration of some available alternative,  the choice of the default option is blocked, no matter whether $x$ is available or not. \vspace{-1ex}
\item Axioms 3 and 4 characterize the associative network. Axiom 3 states that the DM can associate more alternatives with a given alternative when there are more observable alternatives. Axiom 4 states that  if the DM associates $z$ with $x$ when $y$ is observable but fails to do so when $y$ is not observable, then she must associate $z$ with $x$ via the intermediate alternative $y$. \vspace{-1ex}
\item Axiom 5 characterizes the underlying preference relation identified from the DM's choices: For two given alternatives, if the availability of one alternative affects the choice probability of the other in some menu, then the inverse does not occur in any menu where both alternatives are observable. Essentially, the axiom posits that an inferior alternative can only affect the choice probability of a better alternative through its observability but not its availability. \vspace{-1ex}
\end{itemize}
Notably, all the relevant parameters of an ABC can be uniquely identified. 

In Section~\ref{subsec:limit:data:observability}, we characterize the ABC model on a restricted domain where all alternatives are always observable but the set of available alternatives varies. This setting is less demanding on the required choice data, and captures many real-life environments such as online shopping or ordering from a restaurant menu, where out-of-stock items remain visible to the consumer.  Within this restricted domain, we show that the DM’s preference ordering remains uniquely identifiable, and that the associative network is uniquely identified up to the transitive closure. Although the initial attention probabilities are not point-identified---since attention to mutually associated alternatives cannot be disentangled---we show that the aggregate attention probability for such alternatives is uniquely identified.

In Section \ref{subsec:empirical}, we connect our choice model to a broad range of empirical findings, highlighting its advantages over existing approaches. We first demonstrate that the model delivers a unified account of several well-documented menu effects---including the attraction, aspiration, and repulsion effects---as well as their \textit{phantom} counterparts, where merely observable but unavailable alternatives systematically shift choices. To the best of our knowledge, no existing model accommodates this entire suite of effects simultaneously. We then show how our model can capture \textit{diminishing menu effects}, whereby the incremental impact of adding an option weakens as the menu expands. Finally, we discuss how the same associative mechanism is consistent with evidence on \textit{recommendation spillovers}, where drawing attention to some products increases demand for other non-recommended items.

In Section \ref{sec:extension}, we develop three extensions of our baseline framework. First, we incorporate preference heterogeneity by embedding random preferences into the model, so that the choice rule can be interpreted as aggregate choice behavior generated by DMs who differ in their tastes. Crucially, our framework achieves a clean separation between the consideration set formation channel and the preference channel: The identification of the preference distribution proceeds exactly  as it would in the standard full-attention setting. Consequently, random utility  models with unique identification can be embedded in our framework without sacrificing their identification advantages. Second, we allow for a preference-dependent association process where the expansion of the consideration set is directly mediated by the DM's underlying preference. Consequently, two DMs with the same associative network but different preferences may arrive at different final consideration sets even when they start from the same initial consideration set. This extension also yields novel predictions regarding how choices change when the preference   shifts but the associative network remains fixed. Third, we relax the assumption that the set of unavailable but observable alternatives is perfectly recorded in the data. We provide two specific models to accommodate such data limitations, demonstrating in Online Appendix OA-2 that both models preserve robust identification properties.   

Section \ref{sec:app} illustrates how our framework applies to a range of market environments in which associative links are not merely a cognitive primitive but also a strategic variable. We first analyze a multi-product firm’s branding strategy to highlight the tradeoff between expanding demand and facilitating price discrimination. Brand extension exploits attention spillovers but forces the firm to leave information rents to high-valuation consumers. Conversely, by opting for sub-branding, a firm could fracture its associative network to  segment consideration sets, sacrificing some visibility to price discriminate and prevent internal cannibalization. We then analyze an entry game in which a potential entrant can either compete directly or imitate an incumbent’s design, where imitation endogenously creates an associative link and generates attention spillovers. When products remain sufficiently differentiated, this induced association can raise both firms’ profits by expanding attention within each firm’s segment. Finally, we briefly highlight broader design implications of our framework for firms and platforms, exploring how digital platforms optimize associative structures, why luxury brands construct cognitive ``walled gardens,'' and how the strategic deployment of unavailable but observable ``phantom'' products acts as an attention funnel to redirect consumer demand.

Section \ref{sec:conclusion} concludes the paper with a discussion of the underlying assumptions and limitations of our approach, outlining several promising avenues for future research.  The Appendix contains the proofs of our main results. Extended analytical results, proofs, and generalizations are deferred to the Online Appendix.

\subsection{Related Literature}

\label{sec:literature}
	
Our paper belongs to the growing literature on choices with limited attention or limited consideration.\footnote{See, for instance,  \cite{aer2012attention}, MM14, \cite{ecta2016feasibility},   \cite{jet2017attention}, \cite{jet2017more},  \cite{jpe2020random},  \cite{ecta2020inferring},  and \cite{wp2021attention}.} In particular, our approach is closely related to that of MM14, as both models assume that the DM allocates her initial attention randomly and independently. However, our model differs from the model of MM14 in that our DM has a follow-up procedure through which she continues to expand her consideration set via mental association. When the DM does not engage in any mental association, our model reduces to that of MM14. When all alternatives are mutually associated, our model reduces to the rational choice model, where the DM always selects the best available alternative whenever she initially pays attention to at least one observable option.  
	
	
In a concurrent paper, \cite{wp2023memory}  consider  a two-stage stochastic consideration set formation process where the first stage follows MM14. For a given initial consideration set (which they refer to as the awareness set), the DM observes options sequentially and may forget previously observed ones due to limited memory \citep{jet2022memory}. As a result, the final consideration set is a \textit{subset} of the initial one. By contrast, we focus on the mental association process, and the final consideration set is a \textit{superset} of the initial one. 

Our model makes three novel contributions to the literature on limited attention. First, we examine the cognitive process of mental association, which is a fundamental mechanism in forming the DM's consideration set. We provide a concrete procedure (and a random version of it) for how this process operates. Second, our model incorporates bottom-up attention (initial random attention) and top-down attention (mental association), both of which have been shown to be influential factors in decision making  \citep{corbetta2002control, pp2005spatial, gazzaley2012top}.\footnote{Bottom-up attention involves the automatic processing of sensory stimuli in the environment, such as sudden loud noises or bright lights, that capture an individual's attention involuntarily \citep{treisman1980feature, itti2001computational}. By contrast, top-down attention usually refers to the deliberate allocation of attention that is guided by the individual \citep{posner1980orienting, wolfe1989guided}. See also \cite{theeuwes2010top} for a review. In our model, the DM's initial attention is more likely to be bottom-up, as the DM is randomly attracted by the stimuli or salient features of the options. The second-stage mental association is a mixture of bottom-up and top-down attention, as some associated alternatives may come to  mind unintentionally, and individuals may also direct their attention towards options that are relevant to what they have considered in certain dimensions.} Third, we investigate the impact of unavailable but observable alternatives on the DM's  choices.  While those alternatives, which are also called ``phantom'' options \citep{MS1993phantom}, have been studied theoretically in the literature \citep{jet2018aspiration, jpe2019learning},\footnote{See Section \ref{subsec:empirical} for a more detailed discussion of \cite{jet2018aspiration} and \cite{jpe2019learning}.} we focus on investigating how those alternatives affect the DM's attention and obtain a unique identification of our model through those alternatives. In Section \ref{subsec:empirical}, we further show that our framework provides a unified accommodation of several prominent choice anomalies that, to the best of our knowledge, no single existing model can jointly address. 
	
There are a few papers studying the role of networks in individual decisions \citep{te2013search,et2021network, wp2021markov, wp2022subjective}. Among them the most related paper to ours is \cite{et2021network}. In their model, the DM is endowed with an \textit{undirected} associative network. When faced with a menu of available products and an exogenous starting point, she forms her consideration set by including objects that are connected to the starting point through a path (with a potential cap on the length of the path). They also study an extension in which the starting point is unknown to the  analyst.\footnote{In the working paper version of \cite{et2021network}, the authors also consider random networks in which the DM is assumed to consider the alternatives that are directly linked to the starting point.} By comparison, our model studies the combination of both initial random attention and mental association through a directed associative network, and investigates the role played by unavailable but observable alternatives in this process. \cite{te2013search} study a general model of how behavioral search affects the formation of consideration sets. In their model, the connections among the alternatives that determine the search order of the DM can be represented as a network. \cite{wp2021markov} models the exploration of the choice set as a discrete-time Markov chain in which DMs search sequentially by making stochastic pairwise comparisons.  \cite{wp2022subjective} use a directed acyclic network to represent the DM's subjective causal model. 	
 
	
More broadly, our paper contributes to the literature on random choices. Various models have been proposed to rationalize random choice behavior, including the possibility that the DM has random utilities, leading to stochastic choices as a result of utility maximization \citep{ecta1960rum, jmp1978rum, ecta2006randomeu, ecta2014attribute},\footnote{Among the most influential random utility models are the multinomial logit \citep{lu59} and nested logit models \citep{ben1973structure, mcfadden1977modelling}, which are widely used in structural estimations. See also \cite{jpe2022logit} for their behavioral foundations.} and the possibility that the DM randomizes deliberately \citep{aer2019deliberately, aea2022revealed}. While the randomness in our DM's choice behavior is driven by random attention,  our work emphasizes the importance of mental associations in the formation of consideration sets. 
	
Our work also relates to the literature on how choices are influenced by factors beyond the choice menu. These factors can include frames \citep{res2008frame, wp2021recommendation}, the DM's reference points or status quo \citep{jet2005rational, res2014canonical, wp2021reference}, and recommendations from external sources \citep{wp2024recommendation}, among others. While our approach shares some similarities with the work of \cite{wp2021reference} which examines how reference points can shape the DM's attention,  our study focuses on understanding how unavailable but observable alternatives prompt the DM to pay attention to available alternatives through mental association. 	 
	

	
	
\section{Preliminaries}

\label{sec:preliminaries}


There is a nonempty finite set of alternatives $X$, with generic elements denoted by $x, y, z$, etc.  Denote by $\cM$ the collection of all subsets of $X$, with generic elements denoted by $A$, $B$, $C$, etc. Let the default  option be $a^*\not\in X$. We assume that the DM can always choose the default option, and interpret it as choosing nothing or not choosing from a particular menu of alternatives. Examples include walking away from the shop, abstaining from voting, and so on.\footnote{For recent work on allowing a default option  in a random choice setting, see  MM14,  \cite{ecta2016feasibility}, \cite{ecta2020inferring}, \cite{ms2021designing}, and \cite{wp2024recommendation}, among others. It may be hard to observe ``choosing nothing'' outside the laboratory. However, as noted by \cite{ecta2016feasibility}, this concern can be alleviated if the researcher is interested in consumer choice within a class of alternatives. For instance, if the question of interest is the commuting choice of public transportation, the default option could be interpreted as ``driving/walking to work.''}  When there is no confusion, we write $AB$ for $A\cup B$, $Ax$ for $A\cup \{x\}$, and $A \backslash x$ for $A \backslash \{x\}$.
		

A menu consists of two distinct sets of alternatives and is represented by a pair $(A,B) \in \mathcal{M} \times \mathcal{M}$ with $A \cap B = \emptyset$.  While all the alternatives in $AB$ are observable, only those in $A$ are available to the DM. In other words, the DM can pay attention to alternatives in $AB$ but can only choose from $A$ or choose the default option $a^*$. We refer to $A$ as the set of available alternatives, or simply the available set, and $B$ as the set of unavailable but observable alternatives. Let $\mathcal{E}$ denote the collection of all menus.


	

Let $\mathcal{X} =\{(x,x): x \in X\}$. A binary relation on $X$ is a subset $\mathcal{R} \subseteq X\times X$. Fix a binary relation $\mathcal{R}$, and we introduce the following notation. The symmetric part of $\mathcal{R}$ is a binary relation $\mathcal{I}$ such that $x \mathcal{I} y$ if and only if $x\mathcal{R}y$ and $y\mathcal{R}x$. We say that $\mathcal{R}$ is {reflexive} if $\mathcal{X} \subseteq \mathcal{R}$. For all $x,y \in X$, we write $x\mathcal{R}y$ if $(x,y) \in \mathcal{R}$ and use these two notations interchangeably. For all $x \in X$, let $\mathcal{R}(x) = \{y \in X: x\mathcal{R}y\}$, and for all nonempty $A \subseteq X$, let $\mathcal{R}(A) = \cup_{x \in A} \mathcal{R}(x)$. Let $\mathcal{R}^0 = \mathcal{X}$. For all $k \in \mathbb{N}_+$, define $\mathcal{R}^k$ such that $x\mathcal{R}^k y$ if and only if there exists $1 \le t \le k$ and sequence $(x_m)_{m=1}^{t+1}$ in $X$ such that $x_1 = x$, $x_{t + 1} = y$, and for all $m \in \{1, 2, \ldots, t\}$, $x_m \mathcal{R} x_{m + 1}$. Define the transitive closure of $\mathcal{R}$ as $\mathcal{R}^{+} :=\cup_{k=1}^{+\infty} \mathcal{R}^k$.   


A random choice rule is a map $\rho: X\times \mathcal{E} \rightarrow [0,1]$ such that for all $(A,B) \in \mathcal{E}$, (i) $A\neq \emptyset$ implies $\sum_{x \in A} \rho(x, (A, B)) \in (0,1)$, and (ii) $\rho(x, (A,B))>0$ implies $x \in A$. For ease of notation, we write $\rho(x|A, B)$ rather than $\rho(x,(A,B))$. Define $\Phi_{\rho}(A,B) = 1 - \sum_{x \in A} \rho(x|A, B)$. The interpretation is that (1) $\rho(x|A, B)$ denotes the probability that the DM chooses the alternative $x$ in the menu $(A,B)$, and (2) $\Phi_{\rho}(A,B)$ denotes the probability that the DM chooses her default option $a^*$ when facing the menu $(A,B)$.\footnote{Alternatively, we can explicitly incorporate the default option by defining the random choice rule as a map $\rho: (X\cup\{a^*\})\times \mathcal{E} \rightarrow [0,1]$ such that  $\rho(a^*|A, B)\in (0,1)$ whenever $A\neq \emptyset$.}

In the definition of the random choice rule above, condition (i) states that the probabilities of choosing an available alternative and the default option are both positive when $A$ is nonempty; the latter happens if the DM ultimately considers no alternative in $A$. When $A$ is empty, the default option $a^*$ is always chosen and we have $\Phi_{\rho}(A, B)=1$.     

A preference ordering $\succ$ is a strict total order defined on $X$. We use $\max (A; \succ)$ to denote the $\succ$-maximal alternative in $A$ whenever $A$ is not empty. 

\section{Association Based Consideration}

\label{sec:model}


In this section, we formally introduce our choice model in which a DM forms her consideration set through mental association. \bigskip 



\textbf{Initial consideration set.} Following MM14, we assume that each alternative has a fixed probability of being initially considered by the DM, and that the DM attributes her attention to each alternative independently. The attention probability of each alternative is given by the function $\pi: X \rightarrow (0,1)$. For a given $\pi$, define $\mathring{\pi}: X\rightarrow (0,1)$ such that $\mathring{\pi}(x)=1-\pi(x)$ for all $x \in X$. To simplify the notation, we write $\pi_x$ for $\pi(x)$, $\mathring{\pi}_x$ for $\mathring{\pi}(x)$, $\pi_A$ for $\prod_{x \in A} \pi(x)$, and $\mathring{\pi}_{A}$ for $\prod_{x \in A} \mathring{\pi}(x)$. We use the convention that $\pi_A=\mathring{\pi}_A=1$ when $A$ is empty.   


In a given  menu $(A,B)$, the DM  initially considers a subset of $AB$. Since the DM attributes her attention to each alternative independently, the DM initially pays attention to some $C \subseteq AB$ with probability  $\pi_{\scriptscriptstyle\! C} \mathring{\pi}_{\scriptscriptstyle\! (AB)\backslash C}.$ \bigskip  



\textbf{Associative network and the final consideration set.} An \textit{associative network} is a reflexive binary relation on $X$. We use $\mathcal{N}$, $\mathcal{W}$, $\mathcal{U}$, and $\mathcal{V}$ to denote generic associative networks.  The DM expands her initial consideration set through her associative network $\mathcal{N}$. If $(x,y) \in \mathcal{N}$, then $y$ is \textit{directly associated} with $x$, meaning that the attention to $x$ will prompt the DM to further consider $y$.   


In a given menu $(A,B)$, the DM's mental association process only depends on the restricted associative network $\mathcal{N}_{\!\scriptscriptstyle AB}$ of $\mathcal{N}$ on $AB$, where $\mathcal{N}_{\!\scriptscriptstyle AB}=\{(x,y) \in \mathcal{N}: x,y \in AB\}$.  With $\mathcal{N}_{\!\scriptscriptstyle AB}$, the DM's mental association process works as follows. For each alternative $x \in AB$ that she initially considers, she includes every alternative $y$ in $\mathcal{N}_{\!\scriptscriptstyle AB}(x)$ into her consideration set.\footnote{Note that $\mathcal{N}_{\!\scriptscriptstyle AB}$ is a binary relation and $\mathcal{N}_{\!\scriptscriptstyle AB} (x) = \{y \in X: x \mathcal{N}_{\!\scriptscriptstyle AB} y\}$.} For each such alternative $y$, she then expands her consideration set by including each alternative $z$ in $\mathcal{N}_{\!\scriptscriptstyle AB}(y)$. The process terminates when there are no more alternatives in $AB$ that are associated with what the DM already considers.  


Formally, the association procedure described above is modeled as follows.  Consider the transitive closure $\mathcal{N}_{\!\scriptscriptstyle AB}^+$ of $\mathcal{N}_{\!\scriptscriptstyle AB}$.\footnote{Throughout the paper,  $\mathcal{N}^{+}_{\!\scriptscriptstyle AB}$ denotes the transitive closure of $\mathcal{N}_{\!\scriptscriptstyle AB}$ but \textbf{not} the transitive closure of $\mathcal{N}$ restricted on $AB$. Similarly, $\mathcal{N}^{k}_{\!\scriptscriptstyle AB}$ denotes  $(\mathcal{N}_{\!\scriptscriptstyle AB})^k$ but \textbf{not}   $(\mathcal{N}^k)_{\!\scriptscriptstyle AB}$.} If $(x,y)\in \mathcal{N}_{\!\scriptscriptstyle AB}^+$, then there exists a path, i.e., a sequence of alternatives $x_1, x_2, \ldots, x_{n+1} \in AB$ such that $x_1 = x$, $x_{n+1} = y$, and $x_{k+1}$ is directly associated with $x_k$ for all $k \in \{1, 2, \ldots, n\}$. Therefore, the DM can finally consider $y$ as long as she considers $x$. For a given  menu $(A,B)$, we say that $y$ is \textit{associated with} $x$ \textit{in} $(A,B)$ if $(x,y) \in \mathcal{N}_{\!\scriptscriptstyle AB}^+$, and when there is no confusion about the menu $(A,B)$, we simply say that $y$ is \textit{associated with} $x$. Note that if the DM initially considers $x$, the set $\mathcal{N}_{\!\scriptscriptstyle AB}^+(x)$ will be included in her final consideration set. Thus, an initial consideration set $C\subseteq AB$ leads to the final consideration set $\mathcal{N}_{\!\scriptscriptstyle AB}^+(C).$ Figure \ref{fig:association-amplifier} illustrates this process.

	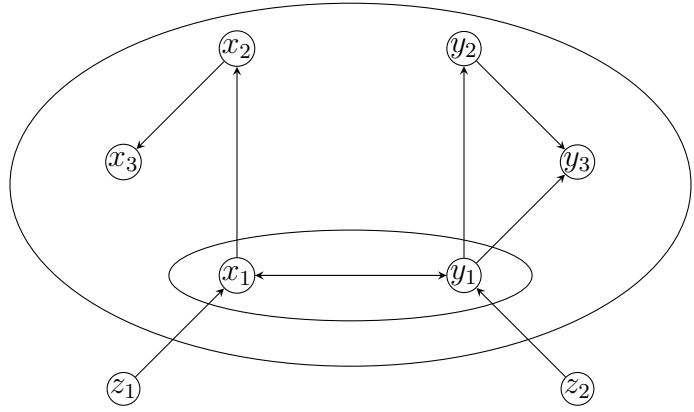
\begin{figure}[hbt!]\centering \begin{tikzpicture}[scale=3]
			\node[circle,draw=black, inner sep=0pt,minimum size=2pt] (c1) at (0,0) {$x_1$};
			
			\node[circle,draw=black, inner sep=0pt,minimum size=2pt] (c2) at (1,0) {$y_1$};
			
			\node[circle,draw=black, inner sep=0pt,minimum size=2pt] (c3) at (0,1) {$x_2$};
			
			\node[circle,draw=black, inner sep=0pt,minimum size=2pt] (c4) at (1,1) {$y_2$};
			
			\node[circle,draw=black, inner sep=0pt,minimum size=2pt] (c5) at (-0.5,0.5) {$x_3$};
			
			\node[circle,draw=black, inner sep=0pt,minimum size=2pt] (c6) at (1.5,0.5) {$y_3$};
			
			\node[circle,draw=black, inner sep=0pt,minimum size=2pt] (c7) at (-0.5,-0.5) {$z_1$};
			
			\node[circle,draw=black, inner sep=0pt,minimum size=2pt] (c8) at (1.5,-0.5) {$z_2$};
			
			\draw [-stealth] (c1) -- (c2);
			\draw [-stealth] (c2) -- (c1);
			\draw [-stealth] (c1) -- (c3);
			\draw [-stealth] (c2) -- (c4);
			\draw [-stealth] (c3) -- (c5);
			\draw [-stealth] (c4) -- (c6);
			\draw [-stealth] (c7) -- (c1);
			\draw [-stealth] (c8) -- (c2);
			\draw [-stealth] (c2) -- (c6);

			\draw (0.5, 0) ellipse (0.8cm and 0.2cm);
			\draw (0.5, 0.4) ellipse (1.5cm and 0.8cm);

		\end{tikzpicture} 
		\caption{The  menu contains 8 alternatives. The associative network is given by the arrows. The initial consideration set of the DM is $\{x_1, y_1\}$, and after association, her final consideration set is $\{x_1, x_2, x_3, y_1, y_2, y_3\}$. Alternatives $z_1$ and $z_2$ are not considered because they are not directly associated with any considered alternative. To illustrate the distinction between $\mathcal{N}$, $\mathcal{N}_{\!\scriptscriptstyle AB}$, and $\mathcal{N}_{\!\scriptscriptstyle AB}^+$, let $A=\{x_1\}$ and $B=\{x_2,x_3\}$. Then $\mathcal{N}$ is the full associative network depicted in the figure; $\mathcal{N}_{\!\scriptscriptstyle AB}$ is the subnetwork consisting of the links $(x_1,x_2)$ and $(x_2,x_3)$; and $\mathcal{N}_{\!\scriptscriptstyle AB}^+$ is the transitive closure of $\mathcal{N}_{\!\scriptscriptstyle AB}$, consisting of $(x_1,x_2)$, $(x_2,x_3)$, and $(x_1,x_3)$.}
		\label{fig:association-amplifier}
	\end{figure}



One important feature of our model is that the DM's association process depends on unavailable yet observable alternatives.\footnote{The importance of understanding how such alternatives affect decision-making has also been highlighted by \cite{ee2021relevance}.}  In our model, such alternatives can help the DM to consider more available alternatives through association. To see this, consider two  menus $(\{z\},\{x\})$ and $(\{z\}, \{x,y\})$, and assume that $y$ is directly associated with $x$ and $z$ is directly associated with $y$, but  $z$ is not directly associated with $x$. If the DM's initial consideration set is $\{x\}$, then she cannot further consider $z$ in menu $(\{z\},\{x\})$. By contrast, in menu $(\{z\},\{x,y\})$, the observable intermediate alternative $y$ allows the DM to reach $z$ through the associative path from $x$ to $y$ and then from $y$ to $z$. 

This feature highlights the importance of allowing the associative network $\mathcal{N}$ to be non-transitive. If $\mathcal{N}$ were transitive, then the links $(x,y)$ and $(y,z)$ would imply a direct link from $x$ to $z$. In that case, considering $x$ would prompt the DM to consider $z$ regardless of whether the intermediate alternative $y$ is observable.\footnote{To further illustrate this point, consider two subnetworks $\mathcal{N}^1$ and $\mathcal{N}^2$ in Figure \ref{fig:association-amplifier}. The subnetwork $\mathcal{N}^1$ is the restriction of $\mathcal{N}$ to alternatives $x_1,x_2$, and $x_3$, while $\mathcal{N}^2$ is the restriction of $\mathcal{N}$ to $y_1,y_2$, and $y_3$. The two subnetworks differ in that $\mathcal{N}^2$ is transitive whereas $\mathcal{N}^1$ is not. Thus, the association from $x_1$ to $x_3$ may arise only when the intermediate alternative $x_2$ is observable, while the association from $y_1$ to $y_3$ persists regardless of the observability of $y_2$.}	

We note that our model can also accommodate the case in which the DM is able to store the associated alternative $y$ in her memory and use it for further association. By doing so, she can directly associate  $z$ with $x$ even in   the menu $(\{z\},\{x\})$. This seems  incompatible with our model because an alternative that is not observable serves as the intermediate alternative in the DM's association process. In fact, our model can accommodate such a case: If the DM can associate $z$ with $x$ through $y$ without $y$ being observable, then it is as if the DM can directly associate $z$ with $x$. The two interpretations lead to the same choice behavior of the DM, and therefore, we do not distinguish them in our choice model. \bigskip 


\textbf{Preference and choice.} In the menu $(A,B)$, if the DM's final consideration set is $C \subseteq AB$, then she chooses $\max(A\cap C; \succ)$ if $A \cap C$ is not empty. Otherwise, she chooses the default option $a^*$.  

\begin{definition} \label{def:ABC}
	
A random choice rule $\rho$ is an  \textbf{A}ssociation \textbf{B}ased \textbf{C}onsideration rule  \emph{(}ABC\emph{)} if there exists a tuple $(\pi, \mathcal{N}, \succ)$, where $\pi$ is an attention probability function, $\mathcal{N}$ is an associative network, and $\succ$ is a preference ordering (i.e., a strict total order over $X$), such that 
\begin{equation}\label{eq:def:ABC}
~\rho(x|A,B) = \sum_{C \subseteq AB:\, x=\max (\mathcal{N}_{\!\scriptscriptstyle AB}^+(C) \cap A;\, \succ) }  \pi_{\scriptscriptstyle\! C} \mathring{\pi}_{\scriptscriptstyle\! (AB)\backslash C}
\end{equation}
for any $(A,B) \in \mathcal{E}$ and $x \in A$. The tuple $(\pi, \mathcal{N}, \succ)$ is said to represent $\rho$ as an ABC.

\end{definition}

With an ABC,  the choice probability of $x$ in the menu $(A,B)$ is the probability with which $x$ is the best alternative in the final consideration set. The default option $a^*$ is chosen only when no available alternative belongs to the final consideration set.\footnote{ Following MM14, we assume a unique preference ordering. An interesting extension is to allow heterogeneous preferences (\citealt{jet2022attention}, \citealt{qe2023random}, \citealt{wp2021attention}, \citealt{wp2024recommendation}). We discuss it in Section \ref{subsec:rp}.}

\section{Axioms and Representation Theorem}

\label{sec:characterization} 

Section \ref{sec:reformulation} reformulates the ABC model to highlight some of its key properties. Section \ref{subsec:axiom} presents the axioms and the representation theorem, discusses the model's comparative statics, and compares it with other random attention models. Section \ref{subsec:limit:data:observability} studies the ABC model in a restricted domain. In Section \ref{subsec:empirical}, we explore the connection of our model to empirical findings, with a particular focus on menu effects.

\subsection{Reformulation}

\label{sec:reformulation}

A feature of our model  is that not every available alternative in a menu is chosen with positive probability. In particular, if the attention to some alternative $x$ always prompts the DM to consider another better alternative, then $x$ is never chosen.   The following proposition characterizes the set of alternatives that are chosen with a positive probability  in a given  menu.  

\begin{proposition} \label{prop:chosen_alternative}
	
Consider an ABC $\rho$ that is represented by $(\pi, \mathcal{N}, \succ)$. For all $(A,B) \in \mathcal{E}$ with $A \neq \emptyset$, $\rho(x|A,B) > 0$ if and only if $x=\max(\mathcal{N}_{\!\scriptscriptstyle AB}^+(x) \cap A; \succ)$.

\end{proposition}

In words,  an alternative $x$ is chosen with a positive probability in the menu $(A,B)$  if and only if there is no association path from $x$ to any available alternative $y$ that is better than it. Otherwise, the attention to $x$ always prompts the attention to $y$ in this menu, which blocks the choice of $x$.   


Next, we investigate the choice probability of each chosen alternative. For a given menu $(A, B)$, define  $H_{\!\scriptscriptstyle \mathcal{N}}(A,B):=\{x \in AB: \mathcal{N}_{\!\scriptscriptstyle AB}^+(x) \cap A \neq \emptyset\}$ as the set of alternatives $x$ in the menu $(A, B)$ such that some alternative in $A$ is associated with $x$. Since $\mathcal{N}$ is reflexive, we have $A \subseteq H_{\mathcal{N}}(A, B)$.  

\begin{proposition} \label{prop:reformulation}

Consider an ABC $\rho$ that is represented by $(\pi, \mathcal{N}, \succ)$. For all $(A, B) \in \mathcal{E}$,  
\begin{equation} \label{eq:prop:reform:2}
\Phi_{\rho}(A,B)=\mathring{\pi}_{\!\scriptscriptstyle H_{\!\scriptscriptstyle \mathcal{N}}(A,B)}.
\end{equation}
Furthermore, if $\{x\in A: \rho(x|A,B)>0 \}=\{x_1, x_2, \ldots, x_n\}$ with $x_1 \succ x_2 \succ \ldots \succ x_n$, then 
\begin{equation} \label{eq:prop:reform:3}
\begin{split}
\rho(x_1|A,B)= 1- \mathring{\pi}_{\!\scriptscriptstyle C_1} ~\text{and} ~	  \rho(x_k|A,B)=\left(1-\mathring{\pi}_{\!\scriptscriptstyle C_k}\right) \prod_{t=1}^{k-1}\mathring{\pi}_{\!\scriptscriptstyle  C_t },  \, \forall k \ge 2,  
\end{split}
\end{equation}
where for each $k \in \{1, 2, \ldots, n\}$,  $C_k = \{y \in AB:  x_k = \max(\mathcal{N}_{\!\scriptscriptstyle AB}^+(y) \cap A; \succ)\}$.
 
\end{proposition}

Equation (\ref{eq:prop:reform:2}) says that the probability for the default option to be chosen is equal to the probability that none of the initially considered alternatives lead to the consideration of any available alternative. Equation (\ref{eq:prop:reform:3}) is a reformulation of our choice rule: The probability for an alternative to be chosen is equal to the probability that it is finally considered through the association process while all better available alternatives are not considered.	
	
\subsection{Characterization}

\label{subsec:axiom}

The first axiom captures the independent attention distribution of the DM.

\bigskip
	
\noindent \textbf{Axiom 1\textemdash Default Independence:} For all $x \in X$ and $A,B\in \mathcal{M}$ with $x \in A\cap B$, $$\frac{\Phi_{\rho}(A,\emptyset)}{\Phi_{\rho}(A\backslash x,\emptyset)} = \frac{\Phi_{\rho}(B,\emptyset)}{\Phi_{\rho}(B\backslash x,\emptyset)}.$$  

\smallskip

Axiom 1 follows from  the I-Independence axiom of MM14.\footnote{The I-Independence axiom is stronger than Axiom 1. It additionally requires that for all $x,y \in X$ and $A,B \in \mathcal{M}$ with $x,y \in A\cap B$ and $x \neq y$, $\frac{\rho(x|A\backslash y,\emptyset)}{\rho(x|A,\emptyset)}=\frac{\rho(x|B\backslash y,\emptyset)}{\rho(x|B,\emptyset)}$.} When all observable alternatives are available, the DM's choice of the default option  depends on whether her initial consideration set is empty or not. The effect of removing an available alternative on the probability of choosing the default option is determined by the extent to which it attracts the DM's initial attention. Axiom 1 posits that the DM has a constant probability of initially considering a given alternative.

\begin{definition} \label{def:network}
	
For all $x \in X$, $A, B \in \mathcal{M}$ with $A \cap B =\emptyset$,  some alternative in $A$ is associated with $x$ through $B$, denoted by $x \xrightarrow[\rho]{B} A$, if either $x \in A$ or $\Phi_{\rho}(A,B) \neq \Phi_{\rho}(A,B\backslash x)$.

\end{definition}

When there is no confusion about $\rho$, we write $x \xrightarrow[]{B} A$ for $x \xrightarrow[\rho]{B} A$. To understand Definition \ref{def:network}, note that if $x$ is in $A$, then some alternative in $A$ (i.e., $x$) is associated with $x$.  If $x$ is not in $A$, then the condition $\Phi_{\rho}(A,B) \neq \Phi_{\rho}(A,B\backslash x)$ indicates $x \in  B$. Furthermore, since removing $x$ from $B$ affects the choice probability of the default option, it follows that the attention of $x$ must lead to the choice of some alternative in $A$, i.e., some alternative in $A$ is associated with $x$, possibly through some intermediate alternatives in $B$.  

\bigskip 

\noindent \textbf{Axiom 2\textemdash Idempotence:} For all $x \in X$ and $(A,B) \in \mathcal{E}$ with $x \in B$,  if $x \xrightarrow[]{B} A$, then $\Phi_{\rho}(Ax, B\backslash x)=\Phi_{\rho}(A, B)$.

\bigskip 

Axiom 2 states that if some available alternative in $A$ is associated with some unavailable but observable alternative $x$,  then the DM's attention to $x$ will result in the selection of an available alternative and block the choice of the default option. Therefore, whether $x$ is available or unavailable has no effect on the choice probability of the default option.

 \bigskip

\noindent\textbf{Axiom 3\textemdash Expansion:} For all $x \in X$ and $(A,B), (C,D) \in \mathcal{E}$, if $C \subseteq A$ and $CD \subseteq AB$,  then $x \xrightarrow[]{D} C$ implies $x \xrightarrow[]{B} A$.

\bigskip

Axiom 3 states that if the attention to $x$ leads to the consideration of some available alternative in a given menu,  then the same would be true in a less restrictive menu with more observable alternatives and more available alternatives. In other words, when there are more intermediate cues in the form of observable alternatives, more alternatives become associated with $x$.

\bigskip 

\noindent\textbf{Axiom 4\textemdash Path Connectedness:} For all $x,y \in X$ and $(A,B) \in \mathcal{E}$ with $x \neq y$, if $x \xrightarrow[]{B} A$ and not $x \xrightarrow[]{B\backslash y} A\backslash y$, then $x \xrightarrow[]{B\backslash y} \{y\}$ and $y \xrightarrow[]{B\backslash x} A$.

\bigskip

Axiom 4 captures the key feature of the associative network: Alternatives are associated with one another via paths in the network. To see this, note that if the DM can associate some alternative $z$ in $A$ with $x$ through $B$ but cannot do so when $y$ is not observable, then either $y = z$ or $y$ is an intermediate alternative for this association process. 

\bigskip 

\noindent\textbf{Axiom 5\textemdash Association Asymmetry:} For all $x,y \in X$ and $(A,B), (C,D) \in \mathcal{E}$ with $x\neq y$ and $x,y \in A\cap C$,  \[\rho(y|A,B) \neq \rho(y|A\backslash x,Bx) \Rightarrow \rho(x|C,D) =\rho(x|C\backslash y,Dy).\]

\smallskip 

Axiom 5 is similar to the I-Asymmetry axiom of MM14.\footnote{The I-Asymmetry axiom states that for all distinct $x,y \in X$ and $A,B\in\mathcal{M}$, if $\rho(y|A,\emptyset) \neq \rho(y|A\backslash x,\emptyset)$, then $\rho(x|B,\emptyset) =\rho(x|B\backslash y, \emptyset)$.}  It states that if the availability of an alternative $x$ affects the choice probability of alternative $y$, then the inverse will not occur. To understand this axiom, note that if the set of observable alternatives is kept the same, so will be the distribution of the final consideration sets. Under such circumstances,  the availability of $x$ affects the choice of $y$ only when $x$ is better than $y$, and thus the inverse never occurs.  

Together, Axioms 1-5 fully characterize our ABC model. 

\begin{theorem}\label{thm:main}
	
A random choice rule $\rho$ is an ABC if and only if  it satisfies Axioms 1-5. The tuple $(\pi, \mathcal{N}, \succ)$ that represents $\rho$ as an ABC is unique.

\end{theorem}

\smallskip 

\noindent \textbf{Identification of the parameters.} The identification of the attention probability function $\pi$ is the same as that in MM14. For every $x \in X$,   $\pi_x=\rho(x|\{x\},\emptyset)$.\footnote{Here we adopt the standard practice of identification in decision theory by exploiting variations in the set of available (observable) alternatives. We discuss the case with limited data in Section \ref{subsec:limit:data:observability} and Online Appendix OA-1. To identify attention probabilities in MM14, the empirical literature makes weaker assumptions by imposing exclusion restrictions \citep{ms2021designing} or utilizing asymmetric demand responses \citep{qje2021consumers}. It would be interesting to study whether similar techniques can be applied to identify parameters in our more general model, and we leave it for future research.}    

The preference ordering $\succ$ can be identified through the DM's choice probabilities in binary menus. Consider two alternatives $x$ and $y$ and assume $x \succ y$. Following the interpretation of Axiom 5, the choice probabilities of $x$ are the same in the two menus $(\{x,y\}, \emptyset)$ and $(\{x\}, \{y\})$. However, the choice probabilities of $y$ must differ in the two menus $(\{x,y\}, \emptyset)$ and $(\{y\}, \{x\})$: In the menu $(\{x,y\}, \emptyset)$, the choice probability of $y$ is less than or equal to the probability of the initial consideration set being $\{y\}$; in the menu $(\{y\}, \{x\})$, the choice probability of $y$ is greater than or equal to the probability of the initial consideration set being either $\{y\}$ or $\{x,y\}$. Since the probability of the initial consideration set being $\{x,y\}$ is positive, $y$ is chosen more frequently in the menu $(\{y\}, \{x\})$ than in the menu $(\{x,y\}, \emptyset)$. Therefore, $x \succ y$ if and only if $\rho(y|\{x,y\}, \emptyset) \neq \rho(y|\{y\}, \{x\})$. Note that the identification of the preference only relies on the assumption that for any given menu, every binary subset of it  has a positive probability of being the initial consideration set and not on the independence of the initial attention.

For the associative network $\mathcal{N}$, note that $x$ is associated with $y$ if and only if $\rho(x|\{x\}, \{y\}) \neq \rho(x|\{x\},\emptyset)$, i.e., removing the unavailable but observable alternative $y$ affects the choice probability of $x$. Again, the identification of the associative network relies on a weaker assumption than the independence of the initial attention; as we formally establish in Online Appendix OA-5, our general identification strategy for $\mathcal{N}$ remains valid under menu-dependent attention. To summarize, the tuple $(\pi, \mathcal{N}, \succ)$ can be uniquely identified based on choices made in menus $(A,B)$ with $|AB|\leq 2$.
\bigskip

\noindent \textbf{More on the default option.} In Online Appendix OA-6, we discuss a natural companion of our ABC model that does not postulate a default option. As in MM14 and \cite{ecta2016feasibility}, we no longer have unique identification once the default option is removed. However, we show that the associative network and the preference can still be uniquely identified under a mild condition on the choice rule. 

\bigskip 

\noindent \textbf{Proof sketch of Theorem \ref{thm:main}.}  We focus on the sufficiency part. With the identified parameters $(\pi, \mathcal{N}, \succ)$, we briefly demonstrate how our axioms lead to the desired representation.  In Step 1, we show that for each menu $(A, B)$ and alternative $x \in B$,  $x \xrightarrow[]{B} A$ implies that some alternative in $A$ is associated with $x$, i.e., $\mathcal{N}^+_{\!\scriptscriptstyle AB}(x) \cap A \neq \emptyset$. We further show that $\mathcal{N}(x) \cap A \neq \emptyset$ implies $x \xrightarrow[]{B} A$. We then focus on the subset $C \subseteq B$ of alternatives with which some alternative in $A$ is associated, i.e., $C=\{y \in B:\mathcal{N}^+_{\!\scriptscriptstyle AB}(y) \cap A \neq \emptyset\}$. By Axiom 2 and the definition of $\xrightarrow[]{B}$, we can inductively show that  $\Phi_{\rho}(A,B)=\Phi_{\rho}(AC, B\backslash C)=\Phi_{\rho}(AC,\emptyset)$. Note that the set $AC$ is equal to $H_{
\scriptscriptstyle\!\mathcal{N}}(A,B)$, and we can show that $\Phi_{\rho}(A,B)=\mathring{\pi}_{\scriptscriptstyle\! H_{\scriptscriptstyle\! \mathcal{N}}(A,B)}$.

In Step 2, we show that for any distinct $x,y \in X$, either $x \succ y$ or $y \succ x$, and that $\succ$ is asymmetric and transitive, i.e., it is indeed a preference ordering.\footnote{The binary relation $\succ$ is asymmetric if for all $x,y \in X$, $x \succ y$ implies not $y \succ x$, and is transitive if for all $x,y,z \in X$, $x \succ y$ and $y \succ z$ imply $x \succ z$.}  

The final step is to show that for any given menu, each alternative is chosen with the  probability specified by the model. To illustrate, we provide a simple example. Consider a menu $(\{x,y,z,w\}, \{r\})$ such that
\begin{equation*}
\begin{split}
& x \succ y \succ z \succ w \succ r ~~~\hbox{and} ~~~\mathcal{N}_{\scriptscriptstyle\!\{x,y,z,w,r\}} =\mathcal{X}_{\scriptscriptstyle\!\{x,y,z,w,r\}} \cup \{(w,y), (r,z)\}.
\end{split}
\end{equation*} Since $r\mathcal{N}z$, by Step 1, we have \begin{equation}\label{eq:thm1:eq1}
\begin{split}
\sum_{\hat{x} \in \{x,y,z,w\}}\rho(\hat{x}|\{x,y,z,w\}, \{r\}) & = 1-\Phi_{\rho}(\{x,y,z,w\}, \{r\}) \\
& = 1-\Phi_{\rho}(\{x,y,z,w, r\}, \emptyset) = 1-\mathring{\pi}_{\scriptscriptstyle\!\{x,y,z,w, r\}}. 
\end{split}
\end{equation}
Since for all $\hat{x} \in \{x,y,z\}$, $\hat{x} \succ w$, by the Association Asymmetry axiom (Axiom 5), making $w$  unavailable does not affect the choice probabilities of $x$, $y$ and $z$. Since $r\mathcal{N}z$ and $w\mathcal{N}y$, again by Step 1, we have
\begin{equation}\label{eq:thm1:eq2}
\begin{split}
\sum_{\hat{x} \in \{x,y,z\}}\rho(\hat{x}|\{x,y,z,w\}, \{r\}) & = \sum_{\hat{x} \in \{x,y,z\}}\rho(\hat{x}|\{x,y,z\}, \{w, r\})  \\ 
& =  1-\Phi_{\rho}(\{x,y,z\}, \{w, r\}) \\
& = 1-\Phi_{\rho}(\{x,y,z,w, r\}, \emptyset) =1- \mathring{\pi}_{\scriptscriptstyle\!\{x,y,z,w, r\}}. 
\end{split}
\end{equation}
Consecutively, we can move $z$ and $y$ to the unavailable but observable set and obtain 
\begin{equation}\label{eq:thm1:eq3}
\begin{split}
\sum_{\hat{x} \in \{x,y\}}\rho(\hat{x}|\{x,y,z,w\}, \{r\}) & = \sum_{\hat{x} \in \{x,y\}}\rho(\hat{x}|\{x,y\}, \{z, w, r\})  \\ 
& =  1-\Phi_{\rho}(\{x,y\}, \{z, w, r\}) \\
& = 1-\Phi_{\rho}(\{x,y, w\}, \emptyset) = 1-\mathring{\pi}_{\scriptscriptstyle\! \{x,y,w\}}. 
\end{split}
\end{equation} \begin{equation}\label{eq:thm1:eq4}
\begin{split}
\rho(x|\{x,y,z,w\}, \{r\}) & =  \rho(x|\{x\}, \{y, z, w, r\})  \\ 
& =  1-\Phi_{\rho}(\{x\}, \{y, z, w, r\}) \\
& = 1-\Phi_{\rho}(\{x\}, \emptyset) =1- \mathring{\pi}_{x}. 
\end{split}
\end{equation}
Note that equations (\ref{eq:thm1:eq1})-(\ref{eq:thm1:eq4}) pin down the choice probabilities of $x,y,z$ and $w$ in the menu $(\{x,y,z,w\}, \{r\})$, which are consistent with equation (\ref{eq:prop:reform:3}) in Proposition \ref{prop:reformulation}. 

\bigskip 

\noindent \textbf{Comparative Statics.} Our framework allows us to directly compare two DMs' associative networks without imposing any restriction on the alignment of their preferences or attention probabilities. The following proposition directly follows from the construction of $\mathcal{N}$, and its proof is omitted. 

\begin{proposition} \label{prop:comparative}

For any two ABCs $\rho_1$ and $\rho_2$ that are represented by $(\pi, \mathcal{N}, \succ)$ and $(\pi', \mathcal{W}, \succ')$ respectively,  the following  statements are equivalent:  \medskip 

\emph{(i)} For all $x \in X$ and $(A,B) \in \mathcal{E}$ with $x \notin AB$,  $$\Phi_{\scriptscriptstyle\!\rho_1}(A,B)\neq \Phi_{\scriptscriptstyle\!\rho_1}(A,Bx) \Rightarrow \Phi_{\scriptscriptstyle\!\rho_2}(A,B)\neq \Phi_{\scriptscriptstyle\!\rho_2}(A,Bx).$$  


\emph{(ii)} The associative network $\mathcal{N}$ is a subset of $\mathcal{W}$. 
\end{proposition} 


\bigskip

\noindent \textbf{Connection to other random attention models.} When all observable alternatives are available, our ABC model induces an attention rule that determines the probabilities of final consideration sets in each menu $(A,\emptyset)$. This makes it possible to compare our model with general random attention models in the literature. In Online Appendix OA-3, we characterize  this attention rule with two properties and show that  it satisfies monotone attention in \cite{jpe2020random} and violates attention overload in \cite{wp2021attention} generically. Moreover, thanks to the additional structure, our model primitives  are uniquely identified, while those in \cite{jpe2020random} and \cite{wp2021attention} remain only partially identified in general.\footnote{ To see this, note that the unique identification of parameters in our ABC model is shown in  Theorem \ref{thm:main}. In the random attention models of \cite{jpe2020random} and \cite{wp2021attention}, the identification of the preference is essentially based on regularity violations, and the revealed preference relation is not unique in general.  As discussed in \cite{jpe2020random}, examples with three alternatives can be constructed in which data cannot reveal the entire preference. For the same reason, the distribution over consideration sets is also only partially identified. We note that this comparison is not surprising because we impose parametric assumptions on the  formation of consideration sets, while \cite{jpe2020random} and \cite{wp2021attention} take a nonparametric approach.}


\bigskip

\noindent \textbf{Discussion of the data requirement.} We acknowledge that  axiomatization and identification of the ABC model demand a rich dataset involving unavailable but observable alternatives. Two cases of limited data are investigated: one in which the DM observes all alternatives with some of them potentially unavailable (Section \ref{subsec:limit:data:observability}), and one in which every observable alternative is available (Online Appendix OA-1). Given the difficulty in perfectly documenting all unavailable but observable alternatives, we examine the extension in which the DM may privately observe additional alternatives that are not recorded in the data in Section \ref{subsec:unobserve} and Online Appendix OA-2. 

There are recent developments in designing controlled choice experiments with unavailable but observable alternatives (``phantom options''). Human subjects are often presented with all observable alternatives, with certain options labeled as unavailable \citep{PW2007phantom, ee2021relevance} or   removed later in the experiment \citep{soltani2012range}. Designing analogous experiments for non-human subjects requires additional nuance; for example,  \cite{Science2015irrationality} study mate choice in T{\'u}ngara frogs using phantom options created by ceiling-placed speakers, and \cite{ac2011phantom} investigates food choice in cats using bowls sealed with transparent plastic lids. Such paradigms are promising for future empirical investigations of the ABC model. 

We also note that certain real-world marketplaces may provide suitable data to test our model. For instance, 
out-of-stock products, due to limited stock or flash sales,  remain physically observable in vending machines and online platforms  \citep{ms2019effectiveness, ms2021designing,conlon2013demand}. 
In online job boards, ``ghost'' jobs---job listings posted by employers with no immediate intent to hire---are prevalent.\footnote{For instance, a recent \href{https://www.bbc.com/news/articles/clyzvpp8g3vo}{report} from BBC cites a study showing that up to $22\%$ of jobs advertised online across the US, UK, and Germany  in 2024 were positions listed with no intent to hire.} 
In college admission, students select from the schools that have extended admission offers (available options), which is a subset of the schools to which they have applied (observable options).
Even though the system is decentralized, data may be available given the growing popularity of online platforms.\footnote{For example, a \href{https://www.commonapp.org/about/reports-and-insights/end-season-report-2024-2025-first-year-application-trends}{report} from Common App---a unified online platform covering over $1100$ universities and colleges---shows that around $1.5$ million students used it in the application cycle 2024-2025.}


\subsection{Characterization with Limited Data}

\label{subsec:limit:data:observability}

In many applications, we cannot observe the DM's choices for every possible menu. Instead, we only observe the random choice rule $\rho$ on a restricted collection of menus $\mathcal{E}' \subseteq \mathcal{E}$. In this section, we focus on the restriction
$\mathcal{E}^O := \{(A,B)\in\mathcal{E} : A\cup B = X\}$. That is, the DM always observes all alternatives, while only a subset of them is available. This setting captures situations such as choosing from a restaurant menu, choosing from a vending machine, or shopping online, where all products are displayed but some may be out of stock. 

We first introduce some notation and terminology. For any subset $A \subseteq X$, denote by $A^c$ the complement of $A$ in $X$, i.e., $A^c = X\backslash A$. For any nonempty $\hat{\mathcal{E}} \subseteq \mathcal{E}$, a $\hat{\mathcal{E}}$-choice rule $\rho$ is defined as a map $\rho: X \times \hat{\mathcal{E}} \rightarrow [0,1]$ such that for all $(A,B) \in \hat{\mathcal{E}}$, (i) $A \neq \emptyset$ implies $\sum_{x \in A} \rho(x|A,B) \in (0,1)$, and (ii) $\rho(x|A,B) > 0$ implies $x \in A$. The definition of $\Phi_{\rho}(A,B)$ remains unchanged. We have the following definition.

\begin{definition} \label{def:fullob}
A $\hat{\mathcal{E}}$-choice rule $\rho$ is a $\hat{\mathcal{E}}$-ABC if there exists a tuple $(\pi, \mathcal{N}, \succ)$, where $\pi$ is an attention probability function, $\mathcal{N}$ is an associative network, and $\succ$ is a preference ordering, such that for all $(A,B) \in \hat{\mathcal{E}}$ and $x \in A$, equation (\ref{eq:def:ABC}) holds. The tuple $(\pi, \mathcal{N}, \succ)$ is said to represent $\rho$ as a $\hat{\mathcal{E}}$-ABC. Furthermore, if $\mathcal{N}$ is transitive, then $\rho$ is said to be a transitive $\hat{\mathcal{E}}$-ABC, and the tuple $(\pi, \mathcal{N}, \succ)$ is said to represent $\rho$ as a transitive $\hat{\mathcal{E}}$-ABC. 
\end{definition}

Note that when the restricted domain is $\mathcal{E}^O$, the DM always observes all alternatives in $X$. Consequently, for any $x,y,z \in X$, if $z$ is associated with $y$ and $y$ is associated with $x$, then initial consideration of $x$ leads to the consideration of $z$. Hence, on the restricted domain $\mathcal{E}^O$, it is behaviorally impossible to distinguish whether $z$ is directly associated with $x$ or only indirectly through $y$, and what matters for the DM's choices is the transitive closure $\mathcal{N}^+$ of the associative network $\mathcal{N}$. In other words, the ABCs represented by $(\pi,\mathcal{N},\succ)$ and $(\pi,\mathcal{N}^+,\succ)$ coincide on $\mathcal{E}^O$. For this reason, we introduce transitive $\mathcal{E}^O$-ABCs in Definition \ref{def:fullob} and focus on them when discussing the uniqueness properties of the model. We begin by characterizing $\mathcal{E}^O$-ABCs below.

\begin{definition}\label{def:asso:proof}
A set $A \in \mathcal{M}$ is association-proof if for every $x \in A^c$, $\Phi_{\rho}(A,A^c)\neq \Phi_{\rho}(Ax,A^c\backslash x)$.    
\end{definition}

If $A$ is association-proof, then no alternative $x$ outside $A$ can direct the DM's attention to $A$. Otherwise, merely observing such an alternative $x$ would already block the choice of the default option, and making it further available would not change the default's choice probability.

\bigskip

\noindent \textbf{Axiom L1\textemdash Association-proof Independence:}  For any association-proof subsets $A,\hat{A}, B, \hat{B} \in \mathcal{M}$, if $\hat{A} \subseteq A$, $\hat{B} \subseteq B$, and $A\backslash \hat{A}= B\backslash \hat{B}$, then $$\frac{\Phi_{\rho}(A,A^c)}{\Phi_{\rho}(\hat{A},\hat{A}^c)} = \frac{\Phi_{\rho}(B,B^c)}{\Phi_{\rho}(\hat{B},\hat{B}^c)}.$$

When $A$ is association-proof, whether the DM chooses the default option from the menu $(A,A^c)$ depends entirely on whether she initially pays attention to some alternative in $A$, and not on her initial attention to alternatives in $A^c$. Hence, Axiom L1 admits an interpretation analogous to that of Axiom 1: The DM has a constant probability of initially considering a given alternative.
\bigskip

\noindent \textbf{Axiom L2\textemdash Monotonicity:} For all $A,B \in \mathcal{M}$ with $A \subseteq B$, $\Phi_{\rho}(A,A^c) \ge \Phi_{\rho}(B,B^c)$. \bigskip

Axiom L2 states that when the set of available alternatives expands, the choice probability of the default option weakly declines. With more alternatives available, the DM is more likely to end up paying attention to at least one available alternative.
\bigskip

\noindent \textbf{Axiom L3\textemdash Decomposability:} For all $A,B,D \in \mathcal{M}$ and $x \in A^c$ with $A=BD$, $$\Phi_{\rho}(A,A^c) \neq \Phi_{\rho}(Ax,A^c\backslash x) \Longleftrightarrow   \Phi_{\rho}(B,B^c) \neq \Phi_{\rho}(Bx,B^c\backslash x) \text{~and~} \Phi_{\rho}(D,D^c) \neq \Phi_{\rho}(Dx,D^c\backslash x).$$

Axiom L3 can be interpreted as capturing the combined logic of Axioms 3 and 4 on the restricted domain $\mathcal{E}^O$. To see this, first note that if $\Phi_{\rho}(A,A^c) \neq \Phi_{\rho}(Ax,A^c \backslash x)$, then initial attention to $x$ cannot lead the DM to pay attention to alternatives in $A$. Consequently, initial attention to $x$ cannot lead the DM to pay attention to alternatives in the smaller sets $B$ and $D$. This corresponds to the logic of Axiom 3. 

Conversely, if $\Phi_{\rho}(A,A^c) = \Phi_{\rho}(Ax,A^c \backslash x)$, then initial attention to $x$ can lead the DM to pay attention to some alternative in $A$---say $y$---via an association path in $A^c$. The same path is also feasible in the larger sets $B^c$ and $D^c$, and $y$ must belong to either $B$ or $D$ since $A=BD$. Therefore, attention to $x$ leads the DM to further consider an available alternative in at least one of the two menus $(B,B^c)$ or $(D,D^c)$, implying that either $\Phi_{\rho}(B,B^c) = \Phi_{\rho}(Bx,B^c \backslash x)$ or $\Phi_{\rho}(D,D^c) = \Phi_{\rho}(Dx,D^c \backslash x)$. This corresponds to the logic of Axiom 4.
\medskip

The next axiom directly weakens Axiom 5 by restricting its statement to $\mathcal{E}^O$. 
\bigskip

\noindent \textbf{Axiom L4\textemdash Weak Association Asymmetry:} For all $x,y \in X$ and $A,B \in \mathcal{M}$ with $x \neq y$ and $x,y \in A\cap B$, $$\rho(y|A,A^c) \neq \rho(y|A\backslash x, (A\backslash x)^c) \Rightarrow \rho(x|B,B^c) = \rho(x|B\backslash y, (B\backslash y)^c).$$ 
 
Axioms L1-L4 fully characterize $\mathcal{E}^O$-ABCs. We state this result, together with the corresponding uniqueness result, in the following theorem.

\begin{theorem} \label{thm:main:limdata}
A $\mathcal{E}^O$-choice rule is a $\mathcal{E}^O$-ABC if and only if it satisfies Axioms L1-L4. Furthermore, every $\mathcal{E}^O$-ABC is a transitive $\mathcal{E}^O$-ABC. Both tuples $(\pi, \mathcal{N}, \succ)$ and $(\pi', \mathcal{N}', \succ')$ represent $\rho$ as a transitive $\mathcal{E}^O$-ABC  if and only if $\mathcal{N}'=\mathcal{N}$, $\succ' = \succ$, and for all $x \in X$, $$\prod\limits_{y: \,y\mathcal{I}x} (1-\pi_y)= \prod\limits_{y: \,y\mathcal{I}x} (1-\pi'_y),$$ where $\mathcal{I}$ is the symmetric part of $\mathcal{N}$.\footnote{Recall that 
the symmetric part $\mathcal{I}$ of $\mathcal{N}$ is defined such that $(x,y)\in\mathcal{I}$ if and only if $(x,y),(y,x)\in \mathcal{N}$.}  \end{theorem}


By the uniqueness result in Theorem \ref{thm:main:limdata}, we can uniquely identify the transitive associative network $\mathcal{N}$ and the DM's preference $\succ$ on the restricted domain $\mathcal{E}^O$. To see this, note that the preference ranking between two alternatives $x$ and $y$ can be recovered as in the full-domain case. Specifically, for any distinct $x,y \in X$, one can show that 
$\rho(y|\{y\},\{y\}^c) \neq \rho(y|\{x,y\},\{x,y\}^c)$ if and only if $x \succ y$. This is because when $y$ is the only available alternative, any initial consideration set that leads to the final consideration of $y$ results in choosing $y$, whereas when $x$ is also available, we must exclude at least the event that the initial consideration set includes both $x$ and $y$, which occurs with positive probability. For the transitive associative network $\mathcal{N}$, we have $x\mathcal{N}y$ if and only if the choice probabilities of the default option are identical in the two menus $(\{x,y\},\{x,y\}^c)$ and $(\{y\},\{y\}^c)$.

The attention probability function $\pi$ is identified only up to the symmetric part of $\mathcal{N}$. Intuitively, if both $x\mathcal{N}y$ and $y\mathcal{N}x$ hold, then attending to either $x$ or $y$ leads the DM to consider both options. Because all alternatives are always observable, one cannot distinguish whether the final consideration of $x$ originates from initial attention to $x$ or to $y$. Hence, initial attention probabilities for $x$ and $y$ can only be jointly identified.

\subsection{Connection to Empirical Findings}

\label{subsec:empirical}

This section explores how our choice model connects to empirical findings and discusses its comparative advantages relative to existing models. 

First, we demonstrate that our choice model can provide a unified accommodation of  well-documented choice anomalies. Suppose $X\subseteq \mathbb{R}^2_+$, where each alternative $x=(x_1,x_2)$ is identified with two attributes.\footnote{The analysis extends straightforwardly to settings with more than two attributes.} We say that $x$ \textit{dominates} $y$ if $x \geq y$ and $x \neq y$. When $x$ dominates $y$, we assume $x\succ y$, i.e., the DM strictly prefers a dominating alternative to a dominated one.  We begin with three choice anomalies concerning the effects that arise when an option $z$ is added to a binary menu $\{x,y\}$ with $z$ being dominated by (or dominating) $x$ but not $y$:

\begin{enumerate}
\item[(1)] \textit{the attraction effect}, where the availability of the dominated option $z$ increases the choice probability of the dominating option $x$ \citep{huber1982adding, highhouse1996phantom, PW2000phantom}, 
\item[(2)] \textit{the phantom attraction effect}, which is similar to the attraction effect except that $z$ is unavailable but observable \citep{PW2007phantom, soltani2012range, ms2016choosing}, and
\item[(3)] \textit{the aspiration effect}, where the presence of an unavailable but observable option $z$ increases the choice probability of option $x$ that is  similar to and dominated  by option $z$ \citep{BASS1992phantom,MS1993phantom,highhouse1996phantom,PW2007phantom, trueblood2017phantom}. 
\end{enumerate}

To see how our model accounts for the attraction effect and the phantom attraction effect, suppose that $x\succ y$ and $x$ dominates $z$.\footnote{The same argument applies when $y\succ x$ and $x$ dominates $z$. In fact, by choosing $\pi$ appropriately, our model can accommodate choice reversals in which $y$ is chosen more frequently than $x$ from $(\{x,y\}, \emptyset)$, but less frequently than $x$ once $z$ becomes observable, regardless of whether $z$ is available or not.} We assume that the dominance relation is salient such that the only mental association among the three alternatives is that attention to either \( x \) or \( z \) prompts the DM to consider both of them. In our model, the probability of choosing $x$ in the menu \( (\{x,y\}, \emptyset) \) is $\pi_x$, whereas in both menus \( (\{x,y,z\}, \emptyset) \) and \( (\{x,y\}, \{z\}) \), the probability increases to \( 1 - \mathring{\pi}_{\{x,z\}} \). Thus, adding \( z \) to the menu increases the choice probability of \( x \). Moreover, when $z$ is unavailable but observable (i.e., the menu is  \( (\{x,y\}, \{z\}) \)), the choice probability of $x$ increases regardless of whether $z$ dominates $x$ or is dominated by $x$, due to the mental association. The former case accommodates the aspiration effect. Notably, \cite{highhouse1996phantom} finds no statistically significant difference in the magnitudes of the attraction and aspiration effects, which our model can potentially explain by the \textit{bidirectional} nature of the associative link between the dominated and the dominating alternatives.

To the best of our knowledge, existing models face significant challenges in accommodating the three aforementioned anomalies simultaneously. For instance, the Random Utility Model (RUM) satisfies regularity, which precludes the possibility that the choice probability of alternative 
$x$
 increases when the menu is expanded.\footnote{A choice rule is called a RUM if there exists a distribution $\mu$ over preference orderings such that, for every menu and every alternative therein, the alternative's choice probability equals the probability (under $\mu$) that the realized preference ranks it best. A choice rule $\rho$ satisfies regularity if, for any pair of menus $(A,\emptyset)$ and $(B,\emptyset)$ with $A \subseteq B$, we have $\rho(x|A,\emptyset) \ge \rho(x|B,\emptyset)$ for all $x \in A$.}
As discussed in Section \ref{sec:literature}, a substantial body of work on context-dependent choice models is able to account for phenomena such as the attraction effect; however, these models generally do not explicitly address the effect of unavailable but observable alternatives. Two recent exceptions are presented in \cite{jet2018aspiration} and \cite{jpe2019learning}. In the model proposed by \cite{jet2018aspiration}, the DM identifies the most preferred observable option as an aspiration. If the aspiration is available, the DM selects it; if not, the DM chooses the most similar available alternative. While this model can explain the aspiration effect, it does not account for the (phantom) attraction effect, as it predicts that the introduction of a dominated option has no impact on choice behavior.

By comparison, \cite{jpe2019learning} introduces a Bayesian probit model in which the DM perceives each choice problem as a collection of options whose values are independently drawn ex ante from a common Gaussian distribution. The DM can glean additional information from all observable alternatives, regardless of their availability.  When $x$ dominates $z$, the addition of $z$ to the menu provides information that favors the dominating option $x$, thereby increasing its choice probability---a result consistent with the (phantom) attraction effect. When $z$ dominates $x$, the same mechanism implies that including $z$ dissuades choice of the dominated option $x$, thus decreasing its choice probability. This prediction, however, runs counter to the aspiration effect.\footnote{To better illustrate this, consider the leading example in \cite{jpe2019learning}  where the DM's prior belief assigns probability $1/6$ to each strict ranking involving $x,y$, and $z$. Suppose that she learns that the event $z\succ x$ occurred and nothing else. Updating the prior with this information will result in assigning a probability of $2/3$ to the event that $y\succ x$, the opposite of what the aspiration effect would predict.}

When the formation of associative links is not solely determined by the dominance relation, ABCs can also accommodate other well-documented menu effects. Recall the previous example in which $x\succ y$ and $x$ dominates $z$. We have assumed an associative link from $z$ to $x$, but it is also behaviorally plausible for the link to be from $z$ to $y$ instead when $z$ and $x$ are similar, because the presence of $z$ may taint the nearby region of the attribute space \citep{simonson2014vices} and prompt the DM's attention to the more distant option $y$. The corresponding ABC model predicts that the choice probability of $y$ increases after the addition of $z$.\footnote{The probability of choosing $y$ in the menu \( (\{x,y\}, \emptyset) \) is $\mathring{\pi}_{x}\pi_y$, whereas in the menu \( (\{x,y,z\}, \emptyset) \), this probability increases to $\mathring{\pi}_{x}(1 - \mathring{\pi}_{\{y,z\}})$.} This prediction accounts for the recent robust evidence of the \textit{repulsion effect} \citep{aaker1991negative,frederick2014limits,spektor2018good}, as well as its phantom version.\footnote{The repulsion effect refers to situations in which adding a dominated option to the menu increases the choice probability of the non-dominating option  rather than that of the dominating one.} Our choice model can also generate behavior consistent with the \textit{(phantom) compromise effect} \citep{simonson1989choice,Science2015irrationality} when the DM associates the compromise option with the extreme options but not vice versa.\footnote{The compromise effect refers to situations in which an option's choice probability increases when the choice set expands so that the option becomes an intermediate choice between more extreme options.}


Second, our model can capture \textit{diminishing menu effects}: If an alternative $y$ boosts the choice probability of $x$ in a smaller menu, then this effect weakens when more alternatives are included. Formally, for all menus $(A,\emptyset)$ and $(B,\emptyset)$ with $A \subseteq B$, and all $x \in A$ and $y \in X\backslash B$, the following condition holds when the associative network $\mathcal{N}$ is transitive: 
\begin{equation*}
\rho(x|Ay, \emptyset) > \rho(x|A, \emptyset) \Rightarrow \rho(x|By, \emptyset) - \rho(x|B, \emptyset) \le \rho(x|Ay, \emptyset) - \rho(x|A, \emptyset). 
\end{equation*}

To see this, consider two distinct alternatives \( x \) and \( y \), and two menus \( (A, \emptyset) \) and \( (B, \emptyset) \) such that \( x \in A \subseteq B \) and \( y \notin B \). Suppose adding \( y \) to \( (A, \emptyset) \) increases the choice probability of \( x \), i.e.,
$$ \rho(x|A, \emptyset) < \rho(x|Ay, \emptyset).$$
Under our model, this implies that attention to \( y \) prompts attention to \( x \) in \( (Ay, \emptyset) \), and the same applies in the larger menu \( (By, \emptyset) \). The increase in choice probability, \( \rho(x|Ay, \emptyset) - \rho(x|A, \emptyset) \), equals $\pi_y$---the probability that \( y \) is initially considered---multiplied by the probability that neither \( x \) nor any better alternative is finally considered in \( (A, \emptyset) \). In the larger menu, \( \rho(x|By, \emptyset) - \rho(x|B, \emptyset) \) is \textit{at most equal to} $\pi_y$ multiplied by the probability that neither \( x \) nor any better alternative is finally considered in \( (B, \emptyset) \). Since the probability that neither \( x \) nor any better alternatives are considered weakly decreases with menu size, we have:
\[ \rho(x|Ay, \emptyset) - \rho(x|A, \emptyset) \geq \rho(x|By, \emptyset) - \rho(x|B, \emptyset). \]
Thus, our model aligns with empirical and experimental evidence showing that menu effects diminish in larger or more complex menus \citep{lee2017influence, park2022more, stanley2024impact}. 

Lastly, the way we incorporate association into consideration is consistent with the \textit{spillover effect} found in \cite{ms2019effectiveness, ms2021designing}. Using a field experiment with vending machines, they demonstrate that recommendations have strong positive spillover effects on the sales of non-recommended products in the menu. One potential mechanism is that recommendations could enhance consumers' attention to the recommended products, thereby increasing their attention to other associated products in the menu. This is empirically validated by the structural estimation of \cite{ms2021designing}. Although we do not model recommendations explicitly, our associative network effectively captures such spillover of attention.




\section{Extensions}

\label{sec:extension}

In this section, we study three natural extensions of the ABC model:
(1) preferences may be random; (2) the association process may depend on the DM’s preference; and (3) the DM’s set of observable alternatives may be imperfectly recorded in the data.\footnote{We discuss additional extensions---such as random networks and general models of initial attention---in the Online Appendix.} For each extension, our primary focus is the identification property of the model. 

\subsection{Random Preference}

\label{subsec:rp}

In this section, we extend the ABC framework to accommodate random preferences. This generalization is particularly useful for modeling preference uncertainty at the individual level or heterogeneous demand in a market environment. In the extended model,  both the attention probability function $\pi$ and the association network $\mathcal{N}$ remain uniquely identified, and our focus is to show that the distribution of preferences can be separated from the attention and association channels. In other words, our framework neither complicates nor simplifies the identification of the DM's random preferences, and RUMs with strong identification properties can be freely embedded into our framework without sacrificing  their identification advantages.
 

Formally, let $\mathscr{P}$ be the set of all possible preference orderings over $X$. A \textit{random preference distribution} is a probability distribution $\tau$ over  $\mathscr{P}$. 
We say that a random choice rule $\rho$ is a \textbf{R}andom \textbf{U}tility \textbf{M}aximizer with \textbf{A}ssociation \textbf{B}ased  \textbf{C}onsideration rule (RUMABC) if there exists a tuple $(\pi, \mathcal{N}, \tau)$, where $\pi$ is an attention probability function, $\mathcal{N}$ is an associative network, and $\tau$ is a random preference distribution, such that for all  $(A,B) \in \mathcal{E}$ and $x \in A$, \begin{equation}\label{eq:def:RUMABC}
\rho(x|A,B) = \sum_{\succ \in \mathscr{P}} \tau(\succ) \left( \sum_{C \subseteq AB:\, x=\max (\mathcal{N}_{\!\scriptscriptstyle AB}^+(C) \cap A; \,\succ) }  \pi_{C} \mathring{\pi}_{(AB)\backslash C} \right). 
\end{equation} The tuple $(\pi, \mathcal{N}, \tau)$ is said to represent $\rho$ as a RUMABC.

For a given RUMABC, the unique identification of the attention probability function $\pi$ and the associative network $\mathcal{N}$ proceeds as in the ABC model. To identify the random preference distribution $\tau$, note that our framework allows us to recover, for each alternative, its probability of being the best within any given set of available alternatives. Specifically, for each nonempty $A\in\mathcal{M}$ and each $x\in A$, we can pin down the probability that $x$ is the best alternative in $A$ under $\tau$, i.e., we can identify
$$Z(x,A):=\sum_{\succ\in\mathscr{P}:\, x=\max(A;\succ)} \tau(\succ).$$  

To see how $Z(\cdot,\cdot)$ can be obtained, note that once $\pi$ and $\mathcal{N}$ are identified, we can recover the distribution of the DM's final consideration sets for each menu. Specifically, for each $A\in\mathcal{M}$ and each $B\subseteq A$, the probability that $B$ is the final consideration set in the menu $(A,\emptyset)$ is $$ \mathcal{F}(B|A):= \sum_{C \subseteq A: \mathcal{N}_{\!\scriptscriptstyle A}^+(C)= B} \pi_{C}\mathring{\pi}_{A\backslash C}.$$ Given $\mathcal{F}(\cdot|\cdot)$, we can derive $Z(\cdot,\cdot)$ by induction. First, for each $A\in\mathcal{M}$ with $|A|=1$, we have $Z(x,A)=1$ for all $x\in A$. Next, fix $n\in\mathbb{N}_+$ and suppose that $Z(x,A)$ has been obtained for all $A\in\mathcal{M}$ and $x\in A$ with $|A|\le n$. Finally, take any $A\in\mathcal{M}$ and $x\in A$ with $|A|=n+1$. Since $$\rho(x|A,\emptyset) = \sum_{B \subseteq A:\, x \in B} \mathcal{F}(B|A) Z(x,B),$$ it follows that
$$Z(x,A)=\frac{\rho(x|A,\emptyset) - \left(\sum\limits_{B \subsetneq A:\, x \in B} \mathcal{F}(B|A) Z(x, B)\right)}{\mathcal{F}(A|A)}.$$ Therefore, $Z(\cdot,\cdot)$ is uniquely pinned down. 

It is well known that, in general, the random preference distribution that generates a given $Z(\cdot,\cdot)$ is only \textit{partially} identified via the Block-Marschak polynomials \citep{ecta1960rum}. In particular, we can find $Z(\cdot,\cdot)$ that is generated by two distinct random preference distributions.\footnote{More specifically, given $Z(\cdot,\cdot)$, if it is induced by some random preference distribution $\tau$, then for every $x\in X$ and every partition $\{D,D'\}$ of $X\backslash\{x\}$, one can uniquely identify the probability that $x$ is ranked above all alternatives in $D$ and below all alternatives in $D'$. Any other random preference distribution $\tau'$ that induces the same collection of such probabilities also generates $Z(\cdot,\cdot)$.}  Nevertheless, because our framework uniquely identifies $Z(\cdot,\cdot)$, it does not weaken the identification of the random preference distribution $\tau$: The extent of identification for $\tau$ is exactly the same as in the corresponding RUM, where the DM always pays attention to all available alternatives.\footnote{A similar observation is made in \cite{wp2024recommendation}.}

Notably, for structured RUMs with stronger identification properties, those advantages are preserved in our framework. Prominent examples include Luce's model \citep{lu59} and the nested logit model \citep{ben1973structure,mcfadden1977modelling}, both of which can be embedded into our framework without sacrificing their unique identification.  Specifically, in Luce's model, the parameter is a utility function $u:X\to\mathbb{R}_{++}$ such that for any set $A$ of available alternatives, the choice probability of $x \in A$ is $$Z(x,A)=\frac{u(x)}{\sum_{y\in A}u(y)}.$$ Such $Z(\cdot,\cdot)$ can be generated by a random preference distribution, and given $Z(\cdot,\cdot)$, the parameter $u$ is uniquely identified up to rescaling. In the nested logit model, the parameters include a partition $\{X_k\}_{k=1}^n$ of $X$, a utility function $u:X\to\mathbb{R}_{++}$, and $\{\beta_k\}_{k=1}^n\subseteq\mathbb{R}_{++}$ such that when the available set is $A$, the probability of choosing $x\in A\cap X_t$ is
$$
Z(x,A)=\frac{u(x)}{\sum_{y\in A\cap X_t}u(y)}
\cdot
\frac{\left(\sum_{z\in A\cap X_t}u(z)\right)^{\beta_t}}{\sum_{k=1}^n\left(\sum_{w\in A\cap X_k}u(w)\right)^{\beta_k}}.
$$
The literature typically assumes $\beta_k\le 1$ for each $k$; this condition ensures that $Z(\cdot,\cdot)$ can be generated by some random preference distribution. The identification of the nested logit model---and its generalization---is established in \cite{jpe2022logit}.\footnote{Unique identification also holds for the single-crossing RUMs \citep{ecta2017singlecrossing, aer2021discrete} and other subclasses of random preferences \citep{jet2024branching}.}   

\subsection{Preference-Dependent Association}

\label{subsec:preference:dependent}

A key feature of our baseline model is that the association process operates independently of the preference. In this section, we introduce an extension that incorporates both associative networks and preferences into the formation of the final consideration set. 

We begin with some notation. For any $k \in \mathbb{N}_{+}$, $A \in \mathcal{M}$, and preference ordering $\succ$ on $X$, let $\max[k](A;\succ)$ denote the set of the top $k$ alternatives in $A$ according to $\succ$. Specifically, if $|A|\le k$, then $\max[k](A;\succ)=A$, and if $|A|>k$, then $\max[k](A;\succ)$ is the unique subset of $A$ of cardinality $k$ such that $x\succ y$ for  all $x\in \max[k](A;\succ)$ and $y\in A\backslash \max[k](A;\succ)$. 

We consider a DM who expands her consideration set by forming associations only from the current top $k$ alternatives. This captures limited cognitive effort: The DM restricts mental association to the most promising options rather than processing the entire set. Formally, let $\mathcal{N}$ and $\succ$ denote the DM's associative network and preference ordering, respectively, and fix a  menu $(A,B)$. Starting from an initial consideration set $C_0$, the DM iteratively expands her consideration set as follows. In each round $t\ge 0$, given $C_t$, the next-round consideration set is
$C_{t+1} = C_t \cup \mathcal{N}_{AB}(\max[k](C_t; \succ)).$ This procedure terminates in finitely many rounds and results in the  final consideration set denoted by $\mathcal{N}^+_{AB}(C_0|k,\succ)$. The DM then chooses the $\succ$-maximal alternative from $\mathcal{N}^+_{AB}(C_0|k,\succ) \cap A$, provided that the set is not empty.

\begin{definition}
A choice rule $\rho$ is a top-$k$ \textbf{a}ssociation-\textbf{b}ased \textbf{c}onsideration rule \emph{(}k-ABC\emph{)} if there is a tuple $(\pi, \mathcal{N}, \succ)$, where  $\pi$ is an attention probability function, $\mathcal{N}$ is an associative network, and $\succ$ is a preference ordering, such that for all $(A,B) \in \mathcal{E}$ and all $x \in A$, $$\rho(x|A,B)= \sum_{C \subseteq AB: \, x=\max(\mathcal{N}^{+}_{AB}(C|k, \succ)\cap A; \succ )}  \pi_{C} \mathring{\pi}_{(A B)\backslash C} .$$ The rule $\rho$ is said to be represented by $(\pi, \mathcal{N}, \succ)$ as a $k$-ABC. 
\end{definition}

The class of ABCs can be viewed as a special case of $k$-ABCs when $k$ is sufficiently large---specifically, when $k \ge |X|$---so that $\max[k](A; \succ) = A$ for all $A \in \mathcal{M}$. At the other extreme, when 
$k=1$, the DM expands her consideration set by including only those alternatives associated with the best option currently under consideration. This procedure  corresponds to the Markovian choice by iterative search (CIS) model studied by \cite{te2013search} in the deterministic setting.


This generalized model accommodates greater heterogeneity in the formation of consideration sets across DMs. To illustrate, let $X=\{x,y,z\}$ and consider two DMs who share the same associative network
$\mathcal{N}=\mathcal{X}\cup\{(x,y),(y,x),(y,z),(z,y)\}.$ DM1's preference is $z \succ_1 y \succ_1 x$, while DM2's preference is $z \succ_2 x \succ_2 y$. Suppose their choice rules $\rho_1$ and $\rho_2$ are represented, respectively, by $(\pi,\succ_1,\mathcal{N})$ and $(\pi,\succ_2,\mathcal{N})$ as $1$-ABCs, where $\pi$ assigns initial attention probability $1/2$ to each alternative. When the menu is $(\{z\},\{x,y\})$, DM1 ends up attending to $z$ whenever her initial consideration set is nonempty; by contrast, DM2 ends up attending to $z$ only when her initial consideration set contains $z$ or equals $\{y\}$. Consequently, DM1 chooses $z$ more frequently than DM2. 

We also note that preference-dependent association can rationalize choice behavior that no preference-independent attention procedure can capture. Returning to the previous example, interpret $x$, $y$, and $z$ as a mechanical keyboard of brand 1, a mechanical keyboard of brand 2, and a magnetic keyboard of brand 2, respectively. Consider a DM with attention probability function $\pi$, associative network $\mathcal{N}$, and preference $\succ_1$ as above. If her choice rule is described by a $1$-ABC, then in the menu $(\{z\},\{x,y\})$, she chooses $z$ with probability $7/8$. If the price of $y$ increases so that her preference changes to $\succ_2$, then she chooses $z$ with probability $5/8$. This decline in the choice probability of $z$ cannot arise under any preference-independent attention procedure, because such procedures require the distribution of final consideration sets to be invariant to the DM's preference and hence imply a constant probability of choosing $z$.

We next highlight a novel menu effect generated by the $k$-ABC model: Adding an unavailable but observable alternative $x$ to a menu $(A,B)$ can \textit{reduce} the choice probability of \textit{every} available alternative in $A$. To see this, let $A=\{y_1,y_2\}$ and $B=\{y_3,y_4\}$. Consider a DM who follows a $1$-ABC with attention distribution $\pi$, preference $x \succ y_1 \succ y_2 \succ y_3 \succ y_4$, and an associative network $\mathcal{N}$ satisfying $\mathcal{N}(x)=\{x\}$, $\mathcal{N}(y_1)=\{y_1\}$, $\mathcal{N}(y_2)=\{y_2\}$, $\mathcal{N}(y_3)=\{y_1,y_3\}$, and $\mathcal{N}(y_4)=\{y_2,y_4\}$. The new  observable but unavailable alternative $x$ disrupts the association process whenever it is considered together with $y_3$ or $y_4$ (or both). If the DM initially considers $\{y_3\}$ (respectively, $\{y_4\}$), she expands her consideration set to include $y_1$ (respectively, $y_2$). However, if $x$ is also in the initial consideration set, it becomes the best option among those considered. The DM then focuses attention on $x$, and the association process terminates. Intuitively, an unavailable but observable alternative that is desirable yet isolated can impede mental association: It attracts attention because it is highly valued, but it fails to lead the DM toward feasible substitutes.

While the $k$-ABC model can accommodate a richer set of choice patterns, its association process and hence the choice probabilities coincide with those of the ABC model in menus with at most two observable alternatives. Hence,  $\pi$,  $\mathcal{N}$, and $\succ$ can all be uniquely identified in the same way as in the ABC model. 


\subsection{Imperfectly Recorded Set of Observable Alternatives}

\label{subsec:unobserve}

An essential assumption in our baseline analysis is that every alternative observable to the DM is perfectly recorded in the data. In practice, however, the DM may privately observe additional alternatives that are not recorded. In this section, we maintain the assumption that the data record the set of \emph{available} alternatives precisely, but we allow the set of \emph{unavailable but observable} alternatives to be imperfectly recorded.

To accommodate the possibility that some observable alternatives are missing from the data, consider a mapping $\Lambda$ that assigns to each recorded menu $(A,B)$ a probability distribution $\Lambda(\cdot|A,B)$ over $\mathcal{E}$ such that, for all $(C,D)\in\mathcal{E}$, $\Lambda(C,D|A,B)>0$ implies $C=A$ and $B\subseteq D$. The distribution $\Lambda(\cdot|A,B)$ is interpreted as the (true) distribution of menus the DM faces when the recorded menu is $(A,B)$: Availability is correctly recorded (so $C=A$), but the true set of unavailable but observable alternatives may strictly contain the recorded set (so $B$ can be a proper subset of $D$).

Accordingly, if the DM's underlying choice rule is an ABC $\rho$, then her observed choice rule $\hat{\rho}$ satisfies, for every $(A,B)\in\mathcal{E}$ and every $x\in A$,
\[
\hat{\rho}(x|A,B)
\;=\;
\sum_{(A,D)\in\mathcal{E}} \Lambda(A,D|A,B)\,\rho(x|A,D).
\]

When the recorded menu lies in $\mathcal{E}^O$, i.e., it takes the form $(A,A^c)$, the menu actually faced by the DM must also be $(A,A^c)$. Hence, by the analysis in Section \ref{subsec:limit:data:observability}, the DM's preference $\succ$ remains uniquely identified. However, without additional restrictions on $\Lambda$, the observed choice rule $\hat{\rho}$ typically lacks necessary information for identifying the associative network $\mathcal{N}$ and the attention probability function $\pi$. In the worst case, identification collapses to what could be achieved if we only observed choices under menus in $\mathcal{E}^O$. To see this, suppose $\rho$ satisfies that for all $(A,B)\in\mathcal{E}$ and all $x\in A$, $\rho(x|A,B)=\rho(x|A,A^c).$ Then one can always rationalize the data by assuming that the DM in fact observes all unavailable alternatives (i.e., faces $(A,A^c)$) even when the recorded menu is much smaller. In this case, the DM's observed choices contain no additional information beyond that characterized by Theorem \ref{thm:main:limdata} in Section \ref{subsec:limit:data:observability}.


We discuss two specific assumptions on $\Lambda$ in Online Appendix OA-2 that yield stronger identification results. First, we consider the case in which extra observable alternatives arise from exogenous sources (e.g., the consumer may encounter a product elsewhere, which is not available at the point of purchase). We capture this via an observability probability function $\eta:X\to(0,1)$. When the documented menu is $(A,B)$, the actual menu faced by the DM is $(A,D)$ with probability
$\left(\prod_{x\in X\backslash (AD)} \bigl(1-\eta(x)\bigr)\right)
\left(\prod_{y\in D\backslash B} \eta(y)\right)$ for each $D\supseteq B$. That is, each alternative $y\in X\backslash (AB)$ can be observable independently with probability $\eta(y)$. We show that the associative network can be uniquely identified. Second, we study the case in which extra observable alternatives come from the DM's past experiences or memory. We model this by introducing an additional set $P\subseteq X$, interpreted as the DM's set of privately observable alternatives. When the documented menu is $(A,B)$, the actual menu faced by the DM is $(A, (BP)\backslash A)$. We axiomatize this extension and provide a general identification strategy for the model parameters, including the case in which the choice rule is  only observed on menus of the form $(A,\emptyset)$.\footnote{A natural further extension is to consider a distribution $\gamma$ over subsets of $X$, where $\gamma(P)$ represents the probability that the DM's additional observable set of alternatives is $P$. An interesting question is the extent to which $\gamma$ can be identified, together with the behavioral implications of the resulting model. We leave this extension for future research.}

\section{Applications}

\label{sec:app}

This section illustrates the applied value of our framework across a range of market environments. We begin by analyzing a multi-brand firm that leverages sub-branding to intentionally sever associative links; this strategic decoupling allows the firm to segment consumers' consideration sets and sustain price discrimination across quality tiers. Next, we examine an entry game in which a new competitor challenges an incumbent premium brand. We demonstrate how an entrant's decision to imitate, rather than compete conventionally, endogenously forges an associative link that drives cross-product attention spillovers. Finally, we explore additional settings, such as association design by platforms, the strategic creation of cognitive ``walled gardens'' by luxury brands, and the deployment of ``phantom'' goods to funnel consumer attention. Together, these applications highlight that associative links are not merely a cognitive primitive of consumer behavior, but also a critical strategic variable that fundamentally alters competitive dynamics, market structure, and optimal pricing.

\subsection{Brand Extension versus Sub-Branding}

\label{subsec:brand}


The strategic choice between brand extension and sub-branding plays a pivotal role in how firms expand their product portfolios and capture new market segments \citep[see, for example,][]{sood2012effects,peng2023meta}. Brand extension allows a company to introduce new products under the well-known parent brand, thereby benefiting from associations with existing products and the resulting high demand. Examples include Pepsi's launch of Pepsi Max and Coca-Cola's introduction of Coke Zero. By contrast, sub-branding creates a distinct brand identity to partially decouple a new (and often high-end) product from the parent brand in consumers' minds. This cognitive separation expands the scope for price discrimination, enabling the firm to extract more revenue from premium segments. For example, Toyota created the Lexus sub-brand for luxury vehicles, and Nestl\'{e} established Nespresso as an upscale alternative. In what follows, we formalize this tradeoff by studying a firm that sells a low-end product in a competitive market and a high-end product in a monopolistic market. Although the setting is stylized, our analysis can readily extend to richer environments.

Consider a firm that produces a low-end product, $L$, and is deciding whether to introduce a new high-end product, $H$, to target a different market segment. We take the price of the low-end product $p_L$ as exogenously given, reflecting mainstream markets with many close substitutes where individual firms have little pricing power. In contrast, we model the high-end segment as monopolistic, allowing the firm to choose $p_H$; this is natural when the premium offering is highly differentiated or shielded by brand loyalty, technological advantages, and other kinds of entry barriers. This framework captures the firm's dual role as a price taker in the mainstream segment and a price setter in the premium segment, allowing for a transparent analysis of entry and pricing decisions in the high-end market. We assume that the marginal cost of production for each product is $0$ and that the introduction of the new product requires a fixed cost $C\ge 0$.

To capture heterogeneous demand, we model the market using a RUMABC $\rho$,  represented by $(\pi, \mathcal{N}, \tau)$. Consumers share an associative network $\mathcal{N}$, which the firm strategically dictates via its branding choice, but exhibit heterogeneous preferences  captured by $\tau$. Let $\pi_{H},\pi_{L}\in(0,1)$ denote consumers' initial attention probabilities for products $H$ and $L$, respectively. There are two types of consumers. A fraction $\alpha \in (0,1)$ of consumers (type 1) value product $H$ at $v_H$ and $L$ at $v_L$, while the remaining fraction $1-\alpha$ of consumers (type 2) value both products at $b$. We assume that $v_H > v_L > b > 0$. This implies that type-1 consumers have a higher valuation for both products, and their utility premium for the high-end product is strictly positive ($v_H - v_L > 0$), whereas the utility premium is zero for type-2 consumers.\footnote{The assumption that high-type (type-1) consumers possess a higher marginal utility for quality is standard in the screening literature. We assume a zero utility premium for low-type (type-2) consumers purely for analytical simplicity.}  For simplicity, we analyze the case in which $p_L = b$. Our analysis can be readily extended to other values of $p_L$. A consumer who purchases a product of value $v$ at price $p$ obtains  utility $v - p$. As is standard, we assume consumers break ties in favor of the firm.
Each consumer has an outside option with utility $0$, and paying attention to any product prompts her to consider the outside option.\footnote{Note that the outside option differs from the default option $a^*$ in Section \ref{sec:preliminaries}, which is not chosen even if all available products in the final consideration set deliver negative  utility.} Finally,  we assume that $\alpha v_H > b$ to ensure that the premium segment is sufficiently lucrative, meaning the firm strictly prefers to maintain a premium price for the high-end product rather than dropping its price to serve the mass market.

We consider the firm’s decision of whether to introduce the high-end product and, conditional on entry, how to determine the price $p_H$ and the branding strategy.  Under brand extension, the firm engineers a bidirectional associative link between the products. Attention to either the existing low-end product $L$ or the new high-end product $H$ naturally triggers consideration of the other; once a consumer pays initial attention to one product, her final consideration set includes both. This gives brand extension a demand advantage through attention spillovers.
 

By contrast, sub-branding intentionally severs the associative links between the products, ensuring that each enters consideration independently. To see why cognitive isolation can be profitable, consider a type-1 consumer whose initial consideration set contains only the high-end product $H$. Under brand extension, associative recall brings $L$ into consideration, so the firm would have to leave the consumer with information rent $v_L-p_L$ to prevent her from choosing the low-end product. Under sub-branding, however, $L$ does not enter the consumer's    consideration set, allowing the firm to extract the full surplus by setting $p_H = v_H$. Thus, while brand extension raises demand by exploiting associative links, sub-branding facilitates price discrimination and surplus extraction by segmenting consideration sets. Proposition \ref{prop_branding} below characterizes when the firm introduces $H$ and identifies the optimal branding strategy.

\begin{proposition}\label{prop_branding}
The firm introduces the high-end product $H$ if and only if \begin{equation}\label{eq:app:subbrand1}
C\leq \max\left\{ \alpha \pi_H(1-\pi_L)v_H,\, \alpha \left(\pi_H+\pi_L-\pi_H\pi_L\right)(v_H-v_L) + \pi_H(1-\pi_L)p_L \right\}.
\end{equation}
Given that the firm introduces $H$, it chooses   sub-branding if and only if 
\begin{equation}\label{eq:app:subbrand2}
v_L-\frac{p_L}{\alpha} \ge \frac{\pi_L}{(1-\pi_L) \pi_H }(v_H-v_L).
\end{equation}
\end{proposition}

Since the price of product $L$ is exogenously given, it is straightforward to show that the optimal pricing strategy under brand extension is to sell product $H$ to type-1 consumers at the price $v_H - v_L + p_L$. Since  $\alpha v_H > b$, it can be shown that the optimal strategy under sub-branding is to extract full surplus from type-1 consumers who only pay attention to $H$. This corresponds to selling product $H$ at the price $v_H$. The tradeoff between the two branding strategies can be illustrated through the following inequality:
\begin{equation}\label{eq:app:subbrand3}
\alpha \pi_H(1-\pi_L) [v_H - (v_H - v_L + p_L) ] \geq \alpha \pi_L [(v_H - v_L + p_L) - p_L] + (1-\alpha) \pi_H(1-\pi_L)  p_L,
\end{equation} 
and (\ref{eq:app:subbrand2}) follows by rearranging (\ref{eq:app:subbrand3}). The left-hand side of inequality (\ref{eq:app:subbrand3}) captures the advantage of sub-branding over brand extension arising from improved price discrimination. Type-1 consumers who pay attention only to the high-end product $H$ will purchase $H$ at the price of $v_H$ under sub-branding compared to $v_H - v_L + p_L$ under brand extension. The opportunity cost of sub-branding---reduced demand---is reflected in the two terms on the right-hand side of (\ref{eq:app:subbrand3}). Specifically, the first term represents the revenue loss from type-1 consumers who initially pay attention to the low-end product, regardless of whether they initially pay attention to the high-end product. Under brand extension, these consumers also consider $H$ due to mental association and eventually purchase $H$ at the price $v_H - v_L + p_L$. Under sub-branding, they purchase the low-end product $L$ at the price $p_L$. The second term captures the revenue loss from type-2 consumers who pay attention only to the high-end product: Under brand extension mental association leads them to also consider the low-end product $L$ and purchase it at the price of $p_L$, whereas under sub-branding they do not buy from the firm. Finally, inequality (\ref{eq:app:subbrand1}) follows from comparing the firm's optimal profit under the two branding strategies (net of the fixed cost $C$) with its profit when it does not introduce the high-end product.


Inequality (\ref{eq:app:subbrand2}) delivers clear comparative statics. Sub-branding is more attractive when the low-end product is a strong fallback (high $v_L$ and relatively low utility premium $v_H - v_L$), because associative recall of $L$ under brand extension would otherwise force the firm to reduce the high-end price to deter down-trading. Sub-branding is also more likely when the low-end product is less likely to be the consumer's entry point (low $\pi_L$) while the high-end product can attract sufficient attention on its own (high $\pi_H$), and when the premium segment is relatively large (high $\alpha$).

The comparative statics above are consistent with observed branding choices in practice. When Toyota introduced its luxury line Lexus, it adopted a sub-branding strategy. This aligns with the model’s prediction because Toyota Camry offers a relatively high standalone value, while the utility premium of Lexus over  Camry is not excessively large.\footnote{Similarly, in the hospitality industry, strong mid-tier offerings provide a credible fallback option, which motivates separation at the top; for example, Ritz-Carlton (part of Marriott) and Waldorf Astoria (part of Hilton) maintain distinct luxury identities rather than extending their parent brands upward.} In contrast, when Toyota introduced the ultra-luxury Toyota Century,  it opted for brand extension rather than creating a separate sub-brand. This choice is also consistent with the theory, as the utility premium of  Century relative to Toyota’s other models is exceptionally large. Such a substantial quality gap justifies positioning the product as the pinnacle of the existing brand, rather than launching a new luxury marque.\footnote{Existing research posits that brand extension is advantageous when the new product aligns well with the parent brand and does not lead to brand dilution \citep{aaker1990consumer, loken1993diluting, john1998negative}. However, these traditional channels alone do not fully explain Toyota's divergent branding decisions. Both Lexus and Century are premium offerings, making it difficult to attribute the sub-branding of Lexus to concerns about diluting mainstream models such as Camry. Furthermore, Lexus appears at least as ``fit-consistent'' with Toyota's core passenger-car competence as Century.}

Turning to attention patterns, sub-branding is more likely to be optimal when the high-end product commands high baseline attention on its own (high $\pi_H$) and does not require the mass-market line to serve as an associative gateway (low $\pi_L$). In the spirits industry, for instance,  Diageo's single malts, such as Lagavulin and Talisker, are often sought directly by consumers actively searching for high-end options. As a result, their demand is less reliant on being discovered through a flagship brand with broad mass-market visibility, such as Johnnie Walker. Conversely, when a new variant suffers from low standalone visibility (low $\pi_H$) and must rely on the parent brand to capture consumer attention (high $\pi_L$), brand extension is more common. This reliance drives the ubiquitous use of line extensions in consumer packaged goods. For instance, Coca-Cola launches Diet Coke and Coca-Cola Zero under the corporate umbrella, and P\&G introduces Tide variants (e.g., Pods, Free and Gentle) under the primary Tide name, precisely because these new products benefit from the parent brand's shelf presence and consumer awareness.

\subsection{Imitation versus Normal Competition}

\label{subsec:imitation}

In this section, we analyze an entry game in which a potential entrant can opt to imitate the incumbent firm's product. This strategy enhances the entrant’s perceived (or actual) quality and, crucially for our purposes, endogenously establishes a bidirectional associative link between the two products: Once a consumer pays attention to one, she is prompted to consider both. Imitation thus entails a sharp strategic tradeoff---it not only boosts attention and perceived quality but also heightens competition and risks provoking enforcement actions (such as litigation) by the incumbent. Our aim is to formalize the conditions under which imitation is profitable for the entrant, and to identify when the resulting attention spillovers are substantial enough that the incumbent optimally accommodates the imitation rather than initiating enforcement actions.

We consider a market with two firms.  Firm 1 is an established incumbent and firm 2 is a potential entrant. We normalize the quality of firm 1's product to $q_1=1$ and assume that firm 2's product possesses a baseline quality $q_2 \in(0,1)$. Both firms are assumed to have zero cost of production. To model heterogeneous demand, we assume a continuum of consumers whose type $\theta$ is uniformly distributed on $[0,1]$. A consumer of type $\theta$ who purchases a product of quality $q$ at price $p$ obtains  utility   $\theta q - p$. As in Section \ref{subsec:brand}, each consumer possesses an outside option with utility $0$, which automatically enters her consideration set whenever she pays attention to any product. Each consumer's initial attention probabilities for products 1 and 2 are given by $\pi_1$ and $\pi_2$, respectively.

Because firm 2 has zero cost, it always enters the market. The strategic focus is entirely on whether it adopts an imitation strategy. Specifically, if firm 2 chooses to imitate, the quality of product 2 improves to $\bar{q}_2 \in [q_2,1)$, and the imitation successfully forges a bidirectional associative link between the two products. Consequently, a fraction $\pi_1+\pi_2-\pi_1\pi_2$ of consumers form a consideration set containing both products simultaneously (while the remaining consumers consider neither). The two firms then set prices $p_1$ and $p_2$, competing directly for this shared pool of consumers. 

Conversely, if firm 2 opts for normal competition (i.e., it does not imitate), no associative link is formed. The firms operate as a standard duopoly, subject to random attention. A fraction $\pi_1(1-\pi_2)$ of consumers pay attention only to product 1, and a fraction $\pi_2(1-\pi_1)$ of consumers pay attention only to product 2. Because a fraction $\pi_1\pi_2$ of consumers pay attention to both products, the firms face a pricing tradeoff between extracting surplus from their captive audiences and competing for the overlapping segment. 

We compare the potential entrant’s equilibrium payoffs under these two entry strategies, focusing on the regime where its baseline quality $q_2$ is low. This case is of primary economic interest, because a low baseline quality is precisely what makes imitation an attractive strategy for an entrant. Denote $$J(\pi_1, \pi_2):= \frac{(2-\pi_1)^2-1}{\pi_1} - \left( \frac{1}{\pi_2} -1\right).$$
The following proposition characterizes the equilibrium under imitation and identifies the conditions under which the incumbent actually benefits from being imitated.

\begin{proposition}

\label{prop:app:imitation}

Under imitation, there exists a unique pair $(p_1, p_2)$ that constitutes a pure-strategy equilibrium. In this equilibrium, firm 1's  profit is weakly higher than its monopoly profit without firm 2's entry if and only if \begin{equation}\label{eq:app:xiaomi1}
1 + \frac{1-\pi_1}{\pi_1}\pi_2 \geq \frac{(4-\bar q_2)^2}{16(1-\bar q_2)}.
\end{equation} 

\noindent Furthermore, for each given $\pi_1$ and $\pi_2$, there exists $q^* \in (0,1)$ such that the following statements hold for every $q_2 < q^*$:

\noindent \emph{(1)} Under normal competition, there exists a unique pair $(\hat p_1, \hat p_2)$ that constitutes a pure-strategy equilibrium. 

\noindent \emph{(2)} If $J(\pi_1, \pi_2) > 0$, then 
there exist   $\bar{q}^*_1<\bar{q}^*_2\in(q_2,1)$ such that $\bar q_2\in[\bar{q}^*_1,\bar{q}^*_2]$ if and only if firm 2's equilibrium profit under imitation is higher  than that under normal competition.

\noindent \emph{(3)} If $J(\pi_1, \pi_2) < 0$, then there exists $\bar{q}^*_3 \in(q_2,1)$ such that $\bar q_2\in [q_2,\bar{q}^*_3]$ if and only if firm 2's equilibrium profit under imitation is higher  than that under normal competition.

\end{proposition}

We focus on pure-strategy equilibria. The proof of Proposition \ref{prop:app:imitation} is provided in Online Appendix OA-8, where we also derive the necessary and sufficient  condition for the existence of a pure-strategy equilibrium under  normal competition  and show that it must be  unique if it exists. Notably, for given $\pi_1$ and $\pi_2$, such an equilibrium always exists when $q_2$ is sufficiently small. This ensures tractability and allows us to characterize firm 2's entry decision when its product has low baseline quality.

Condition (\ref{eq:app:xiaomi1}) characterizes when the incumbent benefits from the entrant’s imitation relative to the monopoly benchmark in which firm 2 does not enter. It also serves as a sufficient condition for non-enforcement: When condition (\ref{eq:app:xiaomi1}) holds, the incumbent optimally accommodates the imitator. Conversely, if the condition is violated, then the incumbent may  have incentives to enforce its rights against imitation---for example, by suing firm 2 for trade dress infringement to deter imitation or induce exit---in order to restore monopoly profits, provided that the litigation costs are low and the probability of winning the lawsuit is high.

Intuitively, condition (\ref{eq:app:xiaomi1}) says that accommodation is more likely when (i) the incumbent's baseline attention is low (low $\pi_1$) while the entrant's is high (high $\pi_2$), allowing the incumbent to harvest substantial attention spillovers via the associative link, and (ii) the imitated quality $\bar{q}_2$ is not too high, which prevents intense competition (noting that the right-hand side of condition (\ref{eq:app:xiaomi1}) is strictly increasing in $\bar{q}_2$).

While we do not claim that condition (\ref{eq:app:xiaomi1}) is the sole determinant of whether an incumbent sues an imitative entrant, it implies comparative statics that strongly align with observed market behavior. Returning to our motivating example, Xiaomi has a large consumer base and therefore attracts consumer attention easily; moreover, its SU7 targets a different market segment than Porsche. Consequently, even if Porsche views Xiaomi's strategy as imitation, the resulting attention spillovers generate sufficient  demand advantages to offset competitive losses, weakening Porsche's incentive to litigate. By contrast, Apple’s landmark lawsuit against Samsung centered on an entrant whose flagship smartphones closely matched the highly salient iPhone in both quality and user experience. Because this dispute featured a high $\bar{q}_2$ and a highly visible incumbent, condition (\ref{eq:app:xiaomi1}) is less likely to hold, leading to stronger incentives for litigation. 

Turning to the entrant’s strategic decision, the first message of Proposition \ref{prop:app:imitation} is that firm 2 does not choose to imitate when the quality after imitation $\bar{q}_2$ is too high, even if its baseline quality $q_2$ is low. As $\bar{q}_2$ approaches $1$, the products become nearly homogeneous, triggering fierce price competition that drives profits to zero for both firms. 

When $q_2$ is low and the quality improvement is marginal ($\bar{q}_2$ is close to $q_2$), the entrant’s imitation decision hinges on the sign of $J(\pi_1,\pi_2)$. Note that $J(\pi_1,\pi_2)$ is increasing in $\pi_2$ and decreasing in $\pi_1$. Intuitively, when the incumbent is highly salient (high $\pi_1$) and the entrant is obscure (low $\pi_2$), imitation allows the entrant to free-ride on the incumbent's attention capital. In this case, even when $\bar{q}_2=q_2$---so that imitation does not directly improve firm 2’s product quality---firm 2 can still have a strict incentive to imitate solely to forge the associative link. Moreover, for small $\bar q_2$, condition (\ref{eq:app:xiaomi1}) also holds. In this case, not only does firm 2 prefer imitation even without a quality improvement, but firm 1 also prefers not to adopt enforcement actions, since the competitive pressure from a low-quality entrant is negligible relative to the additional attention it brings.

An important channel through which imitation can benefit both firms---beyond any direct improvement in product 2’s quality---is that the association created by imitation helps the two firms better complement each other in serving different consumer segments. In particular, firm 1 primarily serves  consumers with high valuations, whereas firm 2 primarily serves  consumers with low valuations. Hence, the incentive for both firms to accommodate imitation can arise beyond the case in which $q_2$ is sufficiently low; the same logic can still apply for moderate $q_2$.\footnote{However, when $q_2$ is high, imitation is less likely to be jointly favored. Since a high $q_2$ intensifies competition, ensuring firm 1's incentive to accommodate imitation requires low $\pi_1$ and high $\pi_2$ so that product 1 gains sufficient additional attention; but then product 2’s attention gain is relatively small, weakening firm 2’s incentive to imitate. Moreover, when $q_2$ is already high, the scope for quality improvement via imitation (i.e., $\bar{q}_2-q_2$) is limited, further undermining firm 2’s incentive to imitate.} For instance, when $q_2=\bar{q}_2=0.29$, $\pi_1=0.48$, and $\pi_2=0.2$, both imitation and normal competition admit a unique pure-strategy equilibrium. Firm 1’s equilibrium profit under imitation exceeds its monopoly profit, and firm 2’s equilibrium profit under imitation exceeds that under normal competition.\footnote{The detailed analysis is included in Online Appendix OA-8.}

The main predictions in Proposition \ref{prop:app:imitation} are consistent with the recent empirical evidence in the study of imitation strategies \citep{yilmaz2023does,wang2023performance}. For instance, \cite{yilmaz2023does} demonstrate that imitation operates via both discovery and substitution effects, where the discovery effect---akin to increased mental associations and attention---may increase the demand of the incumbent. Similarly, \cite{wang2023performance} provide experimental and field evidence that imitation can enhance consumer perceptions of the incumbent, provided that the imitator is sufficiently vertically differentiated from the incumbent. These findings provide empirical support for our central insight that imitation can function as a cognitive catalyst, generating  associations for both products when quality differentiation reduces direct competition and initial brand salience is low, thereby creating mutual benefits rather than zero-sum rivalry.

Although our analysis focuses on the scenario in which the incumbent possesses a higher quality than the entrant, it readily extends to the reverse case. Mirroring the logic of Proposition \ref{prop:app:imitation}, we can characterize the conditions under which a high-quality entrant optimally imitates a mainstream incumbent's design, generating attention spillovers that the incumbent strictly prefers to accommodate. This extension explains why Dr. Martens actively protects the intellectual property of its iconic boot design against fast-fashion competitors of comparable quality, but not against high-fashion brands such as Prada, where extreme vertical differentiation prevents direct competition  and instead yields positive associative attention.

\subsection{Discussion on Additional Applications}

\label{subsec:moreapp}

We have presented two stylized applications demonstrating how firms strategically add or remove associative links. More broadly, manipulating these links is not merely a branding exercise; it is a powerful design lever that shapes which products enter consumers' consideration sets, how attention is allocated across options, and ultimately, how competitive dynamics unfold. Moreover, forming or severing links within an associative network often co-moves with other product attributes---such as perceived quality, positioning, and salience---so the overall effect is typically multifaceted rather than purely attentional. This perspective naturally opens the door to a broader set of environments in which associative networks are likely to play a pivotal role. Below, we briefly outline several possibilities.

\paragraph{Platform Network Design.} 
A natural application of our framework concerns digital platforms that can engineer associative links---e.g., through ``customers also viewed'' modules, sponsored placements, bundles, and cross-page hyperlinks. Such design choices redirect attention across products and, when some items are out of stock, shape which available alternatives enter consumers' final consideration sets. This is particularly relevant for hybrid platforms such as Amazon that both sell their own products and match buyers with third-party sellers. A key complication here is that platform payoffs are product-specific: Margins, fees, and strategic benefits differ significantly depending on whether an item is supplied by a third-party seller or by the platform itself. In  Online Appendix OA-7, we study a special case with our RUMABC framework, analyzing how adding a single associative link  affects the sales volume of a target product. We characterize the optimal link as a function of the existing associative network and the distribution of menus.

\paragraph{The Walled Garden: Dyson’s Strategic Decoupling.}
Our framework clearly illuminates the strategic imperative for a premium brand like Dyson to sever associative links with lower-end competitors. The core issue arises from the potential association flowing from Dyson to budget alternatives: A consumer initially noticing Dyson (perhaps via an advertisement) may, through the associative network, recall a cheaper and functional substitute, thereby expanding her consideration set to include both Dyson and the budget option. Such simultaneous presence invites a direct cost comparison, prompting the consumer to rationalize the ``good enough'' cheaper option as the more sensible choice, which may result in a lost sale for Dyson.  Consequently, Dyson’s strategic goal is to establish  a mental ``walled garden.'' By severing the associative link to the low-end segment---through distinctive design, luxury branding, and marketing that avoids direct functional comparisons---Dyson ensures that consumers who initially consider Dyson do not have their attention diverted to cheaper substitutes. In this context, the association becomes a liability because it empowers inferior alternatives to ``poach'' potential sales.

\paragraph{The Attention Funnel: Popmart’s Strategic Scarcity.}
Our model also illustrates how brands can strategically manipulate associative networks using unavailable yet observable alternatives---a strategy exemplified by Popmart’s management of its flagship character, Labubu. By releasing Labubu in limited editions or consistently maintaining ``sold out'' statuses, Popmart ensures that specific figures remain highly visible yet unattainable, thereby transforming Labubu into a marketing asset fueled by scarcity. 
Crucially,  our model posits that it is the observability of an alternative, rather than its availability, that drives the association process. 
Popmart leverages this insight by featuring these ``phantom'' Labubu figures prominently across social media, physical displays, and marketing campaigns. 
In this framework, Labubu functions as a central hub: When a consumer thinks about collectible toys, their initial consideration triggers an associative link to the highly salient Labubu, effectively capturing attention from the broader product category. 
The strong mental connection with the unavailable Labubu then prompts consideration of other available Popmart products. In effect, the unavailable Labubu acts as an ``attention funnel,'' drawing consumers away from competitors and channeling them toward Popmart's purchasable inventory.\footnote{Introducing unavailable but observable alternatives is not without risk. For instance, 
the pre-announcement of a ``phantom'' product can backfire. While intended to build anticipation, such announcements can inadvertently activate specific features in the consumer's mind; if competitors already offer products with these features, the announcement serves as a ``cognitive bridge'' to immediately available rival options. This creates a competitive diversion that intensifies the traditional Osborne effect:  Not only does the firm lose sales to its future self, but it also inadvertently subsidizes the sales of its current rivals.}




\section{Conclusion} \label{sec:conclusion} 

In this paper, we develop and characterize a random choice model  in which a DM expands her consideration set through mental associations between alternatives. The framework provides a tractable way to study how association shapes choice, particularly when some alternatives are merely observable but unavailable. We demonstrate that our model provides a unified explanation for several prominent choice anomalies and illustrate the role of associative links as a critical strategic variable in applied settings such as branding, imitation, and platform design. We conclude the paper with a discussion of several assumptions and limitations of our approach, as well as promising avenues for future research.

First, the data requirement for operationalizing our model can be demanding. Researchers therefore need to assess the empirical context carefully, as the practicality of this requirement differs significantly across empirical settings. In laboratory environments, the researcher has considerable control over the data-generating process, as illustrated by the experimental designs of \cite{soltani2012range}, \cite{ee2021relevance}, and \cite{Science2015irrationality}.  In field contexts, we point to the vending-machine studies by \cite{conlon2013demand} and \cite{ms2019effectiveness, ms2021designing}. Although these papers have a different primary focus and do not explicitly study mental association, they demonstrate that the relevant type of data collection is feasible: They track stock-outs and exploit variation in product availability as part of their empirical strategy.  This observation extends to broader environments in which consumers frequently encounter product pages with ``sold out'' labels or displays of waitlisted items, both in physical and digital storefronts. We caution, however, that some settings require greater care. For example, when the ``observable but unavailable'' status arises because an item is unaffordable---as in some car-market applications---the analyst must understand what is affordable or unaffordable for the individual decision maker. By contrast, measurement is often more straightforward when unavailability is structurally defined. Geographic constraints, for example, may restrict access to region-specific goods or streaming content, even when those items remain visible or widely discussed online. Thus, while the data requirement is substantial, it is reasonable in appropriately targeted settings, provided that the researcher carefully accounts for the source and observability of unavailability.

Second, our analysis remains silent regarding the underlying mechanisms by which associative links are established; we either adopt a revealed preference approach or rely on intuitive connections between firms' strategic actions (such as branding and imitation strategies) and consumers' mental associations. It would be both theoretically interesting and practically relevant to develop microfoundations for the formation of associative networks---identifying the specific factors and behavioral processes that drive the emergence and structure of associations among alternatives. Such foundations would enhance the explanatory power of the model and anchor it more deeply in economic and psychological theory.


Third, we view our model as complementing and extending recent advances in the literature on consideration sets and associative processes. For instance, the network structure in \cite{et2021network} is assumed to be undirected (i.e., associative links are bidirectional). Our framework naturally extends their analysis to the more general scenario of directed networks. In many economic contexts, associations are inherently asymmetric: Attention flows from one option to another, but not necessarily in the reverse direction. This generalization provides a more accurate depiction of behavioral phenomena   and clarifies the strategic implications of network design. A highly promising direction for future research is to extend our framework to further integrate random networks with random attention. While Online Appendix OA-4 outlines one specific approach and our initial explorations provide some preliminary insights, a comprehensive analysis spanning multiple modeling frameworks would significantly advance the literature. Fully characterizing this interplay remains an important and open direction for future research, as also highlighted by \cite{et2021network}.

Fourth, in the applications presented in Section \ref{sec:app} and Online Appendix OA-7, we assume that decisions regarding associative links are costless. When an association forms between the products of two competing firms, a particularly promising direction for future research is to examine costly link formation, where both the existence and the strength of connections depend endogenously on the joint strategic actions of both firms. This extension would bridge our framework with the literature on weighted networks, shifting the analytical focus from the mere existence of a link to the intensity and cost of competitive associations.

Finally, while we follow the convention in decision theory by focusing on choice data, non-choice data---such as eye-tracking, mouse-tracking, or search process logs---would be highly complementary to our current approach. These data have the potential to reveal how individuals cognitively connect alternatives by tracing visual fixations and search patterns that correspond to associative links in real time. For example, systematically focusing first on $x$ and then on $y$ across different menus would provide empirical support for an associative link from $x$ to $y$. Incorporating such non-choice data can help to validate---and potentially recover---the dynamic nature of mental associations.




\appendix

\section{Appendix}

\begin{proof}[Proof of Proposition \ref{prop:chosen_alternative}]
	
If $\rho(x|A,B) > 0$, then there exists $C \subseteq AB$ such that $x=\max(\mathcal{N}_{\!\scriptscriptstyle AB}^+(C) \cap A;\succ)$. Since $x \in \mathcal{N}_{\!\scriptscriptstyle AB}^+(C)$, we have $\mathcal{N}_{\!\scriptscriptstyle AB}^+(x) \subseteq \mathcal{N}_{\!\scriptscriptstyle AB}^+(C)$, and thus $x=\max(\mathcal{N}_{\!\scriptscriptstyle AB}^+(x) \cap A; \succ)$.  Conversely, if $x=\max(\mathcal{N}_{\!\scriptscriptstyle AB}^+(x) \cap A; \succ)$, then $x$ is chosen when the initial consideration set is $\{x\}$. Thus, $\rho(x|A,B) \ge \pi_{\scriptscriptstyle\!x} \mathring{\pi}_{\scriptscriptstyle\! (AB)\backslash x} > 0$.   
\end{proof}

\begin{proof}[Proof of Proposition \ref{prop:reformulation}]
For equation (\ref{eq:prop:reform:2}), note that for any initial consideration set $C \subseteq AB$, we have  $\mathcal{N}_{\!\scriptscriptstyle AB}^+(C) \cap A \neq \emptyset$ if and only if $\mathcal{N}_{\!\scriptscriptstyle AB}^+(x) \cap A \neq \emptyset$ for some $x \in C$. Thus, the default option is chosen if and only if any alternative $x$ that satisfies $\mathcal{N}_{\!\scriptscriptstyle AB}^+(x) \cap A \neq \emptyset$ is not initially considered. This leads to equation (\ref{eq:prop:reform:2}).\medskip

For equation (\ref{eq:prop:reform:3}),  we just need to show that it holds for each $k \in \{2,...,n\}$. When the initial consideration set is $D$, $x_k$ is chosen if and only if $\mathcal{N}_{\!\scriptscriptstyle AB}^+(D) \cap \{x_1,...,x_k\}=\{x_k\}$. That is, $D \cap \{y \in AB:  x_k = \max(\mathcal{N}_{\!\scriptscriptstyle AB}^+(y) \cap A; \succ)\} \neq \emptyset$ and for all $m\le k-1$, $D \cap \{y \in AB:  x_m = \max(\mathcal{N}_{\!\scriptscriptstyle AB}^+(y) \cap A; \succ)\} = \emptyset$. This leads to  equation (\ref{eq:prop:reform:3}).
\end{proof}

\begin{proof}[Proof of Theorem \ref{thm:main}] We first prove the following lemma, which will also be useful for proving Theorem \ref{thm:main:limdata} and  Theorem \ref{thm:oa:pabc:axiom}.

\begin{lemma}\label{lm:general:preference}
Consider a random choice rule $\rho$ and  a nonempty set $D \subseteq X$. Define a binary relation $\succ_D$ on $D$ such that for all $x, y \in D$, $x \succ_D y$ if $x \neq y$ and for some $A \subseteq D$,  $$\rho(y|A,D\backslash A) \neq \rho(y|A\backslash x, D\backslash (A\backslash x)).$$ Let $\sigma$ be a distribution over subsets of $D$ which satisfies that for all $A \subseteq D$ with $|A| \le 3$, there is $B \subseteq D$ such that $A \subseteq B$ and $\sigma(B) > 0$. If $\succ_D$ is asymmetric, and for all $A \subseteq D$,  $$\Phi_{\rho}(A,D\backslash A)=\sum_{B \subseteq D: B\cap A = \emptyset} \sigma(B),$$  then $\succ_D$ is transitive, and for all distinct $x,y \in D$, either $x \succ_D y$ or $y \succ_D x$ holds. 
\end{lemma}

\begin{proof}
First, the assumption on $\sigma$ implies that $\Phi_{\rho}(\emptyset,D) =1 =\sum_{B \subseteq D} \sigma(B)$. Consider distinct $x,y \in D$. If $x \not\succ_D y$ and $y \not\succ_D x$, then \begin{equation}\label{eq:general:preference1}
\begin{split}
\Phi_{\rho}(\{x,y\}, D\backslash \{x,y\}) & = 1- \rho(x|\{x,y\}, D\backslash \{x,y\})-\rho(y|\{x,y\}, D\backslash \{x,y\}) \\
& = 1- \rho(x|\{x\}, D\backslash x)-\rho(y|\{y\}, D\backslash y) \\ 
& = 1-(1-\Phi_{\rho}(\{x\}, D\backslash x))-(1-\Phi_{\rho}(\{y\}, D\backslash y)) \\
& = \Phi_{\rho}(\{x\}, D\backslash x) + \Phi_{\rho}(\{y\}, D\backslash y) - 1; \text{~and}  
\end{split}
\end{equation}   \begin{equation}\label{eq:general:preference2}
    \begin{split}
    ~~\, \Phi_{\rho}(\{x,y\}, D\backslash \{x,y\}) & = \sum_{A\subseteq D: \{x,y\} \cap A = \emptyset} \sigma(A) = 1- \sum_{A\subseteq D: \{x,y\} \cap A \neq \emptyset} \sigma(A) \\ 
    & = 1- \sum_{A\subseteq D: x \in A} \sigma(A) -  \sum_{B\subseteq D: y \in B} \sigma(B) + \sum_{C\subseteq D: x,y \in C} \sigma(C) \\
    & = \Phi_{\rho}(\{x\}, D\backslash x) + \Phi_{\rho}(\{y\}, D\backslash y) + \sum_{C\subseteq D: x,y \in C} \sigma(C) - 1. 
    \end{split}
\end{equation}
By our assumption on $\sigma$,  we have  $\sum\limits_{C\subseteq D:\, x,y \in C} \sigma(C) > 0$. It then follows that equations (\ref{eq:general:preference1}) and (\ref{eq:general:preference2}) are contradictory to each other. Thus, we have either $x \succ_D y$ or $y \succ_D x$. \medskip

Next, consider distinct $x,y,z \in D$ such that $x \succ_D y$ and $y \succ_D z$. We show $x \succ_D z$. Suppose to the contrary that $z \succ_D x$. Let $A=\{x,y,z\}$.  Similar to equation (\ref{eq:general:preference1}), by the conditions that $y \not\succ_D x$, $z \not\succ_D y$, and $x \not\succ_D z$, we have \begin{equation}\label{eq:general:preference3}
    \Phi_{\rho}(A, D\backslash A) = 1 - \sum_{a \in A} \Phi_{\rho}(\{a\}, D\backslash a) + \sum_{B \subseteq A: |B|=2 }\Phi_{\rho}(B, D\backslash B).
\end{equation} However, similar to equation (\ref{eq:general:preference2}),  we have \begin{equation}\label{eq:general:preference4}
\begin{split}
 & ~~~\, \Phi_{\rho}(A, D \backslash A)  = 
 \sum_{B \subseteq D: B\cap A = \emptyset} \sigma(B) = 1-\sum_{B \subseteq D: B\cap A \neq \emptyset} \sigma(B) \\ 
  & = 1- \left( \sum_{a \in A} \left( \sum_{B\subseteq D: a \in B} \sigma(B) \right) -  \sum_{\{b,c\} \subseteq A} \left( \sum_{C\subseteq D: \{b,c\} \subseteq C} \sigma(C) \right) + \sum_{E\subseteq D: A \subseteq E} \sigma(E) \right)\\
  & = 1 - \sum_{a \in A} \Phi_{\rho}(\{a\}, D\backslash a) + \sum_{B \subseteq A: |B|=2 }\Phi_{\rho}(B, D\backslash B) -  \sum_{E\subseteq D: A \subseteq E} \sigma(E), 
\end{split}    
\end{equation} where the last equality holds since for all $a \in A$ and $\{b,c\} \subseteq A$, we have $\sum_{B\subseteq D: a \in B} \sigma(B) = 1- \Phi_{\rho}(\{a\},D\backslash a)$ and $\sum_{C\subseteq D: \{b,c\} \subseteq C} \sigma(C) =  1-\Phi_{\rho}(\{b\}, D\backslash b) + 1-\Phi_{\rho}(\{c\}, D\backslash c)- (1-\Phi_{\rho}(\{b,c\}, D\backslash \{b,c\})).$ 
By our assumption on $\sigma$, $\sum\limits_{E\subseteq D: A \subseteq E} \sigma(E)$ is strictly positive, and thus equations (\ref{eq:general:preference3}) and (\ref{eq:general:preference4}) are contradictory to each other. It follows that $\succ_D$ is transitive.  \end{proof}

\noindent \textbf{(Necessity)} Consider an ABC $\rho$ that is represented by $(\pi, \mathcal{N}, \succ)$. We show that it satisfies Axioms 1-5. For Axiom 1, note that for all $A \in \mathcal{M}$, $\Phi_{\rho}(A,\emptyset)=\mathring{\pi}_{\!\scriptscriptstyle A}$. Thus, for any $x \in A$, the ratio $\frac{\Phi_{\rho}(A,\emptyset)}{\Phi_{\rho}(A\backslash  x,\emptyset)}=\mathring{\pi}_{\!\scriptscriptstyle x}$ is independent of $A$. Next, we prove the following lemma.

\begin{lemma}\label{lm:necessary:arrow} Consider an ABC $\rho$ represented by $(\pi, \mathcal{N}, \succ)$. For all $(A,B) \in \mathcal{E}$ and $x \in X$, $x \xrightarrow[]{B} A$ if and only if $\mathcal{N}^+_{\scriptscriptstyle\! AB}(x) \cap A \neq \emptyset$.  \end{lemma}
			
\begin{proof}[Proof of Lemma \ref{lm:necessary:arrow}] The case where $x \in A$ is trivial. For the nontrivial case,  assume $x \in B$. If $\mathcal{N}^+_{\scriptscriptstyle\! AB}(x) \cap A = \emptyset$, then for all $y \in AB$ such that $\mathcal{N}^+_{\scriptscriptstyle\! AB}(y) \cap A \neq \emptyset$, we can find a sequence $(y_k)_{k=0}^n$ in $AB$ such that $y_0=y$, $y_n \in A$, and for all $k=\{0,...,n-1\}$, $y_k \mathcal{N} y_{k+1}$. Note that any $y_k$ cannot be $x$, since otherwise $\mathcal{N}^+_{\scriptscriptstyle\! AB}(x) \cap A \neq \emptyset$. Thus, we have $\mathcal{N}^+_{\scriptscriptstyle\! (AB)\backslash x}(y) \cap A \neq \emptyset$. It then follows that $H_{\mathcal{N}}(A,B) \subseteq H_{\mathcal{N}}(A,B\backslash x)$. Since $H_{\mathcal{N}}(A,B\backslash x) \subseteq H_{\mathcal{N}}(A,B)$ trivially holds, we have $H_{\mathcal{N}}(A,B)  = H_{\mathcal{N}}(A,B\backslash x)$. By Proposition \ref{prop:reformulation}, we have $\Phi_{\rho}(A,B)=\Phi_{\rho}(A,B\backslash x)$, i.e., not $x \xrightarrow[]{B} A$.  Conversely, if $\mathcal{N}^+_{\scriptscriptstyle\! AB}(x) \cap A \neq \emptyset$, then $x \in H_{\mathcal{N}}(A,B)$ and $x \notin H_{\mathcal{N}}(A,B\backslash x)$. It follows that $H_{\mathcal{N}}(A,B\backslash x) \subsetneq H_{\mathcal{N}}(A,B)$. By Proposition \ref{prop:reformulation}, $\Phi_{\rho}(A,B)\neq \Phi_{\rho}(A,B\backslash x)$, and we have $x \xrightarrow[]{B} A$. 
\end{proof}

For Axiom 2, for any $(A,B) \in \mathcal{E}$, consider some $x \in B$ such that $x \xrightarrow[]{B} A$.  To show $\Phi_{\rho}(A,B)=\Phi_{\rho}(Ax,B\backslash x)$, by Proposition \ref{prop:reformulation}, it suffices to show $H_{\mathcal{N}}(A,B) = H_{\mathcal{N}}(Ax,B\backslash x)$. Since $H_{\mathcal{N}}(A,B) \subseteq H_{\mathcal{N}}(Ax,B\backslash x)$, it suffices to show $H_{\mathcal{N}}(Ax,B\backslash x) \subseteq H_{\mathcal{N}}(A,B)$. For any $y \in H_{\mathcal{N}}(Ax,B\backslash x)$, there exists a sequence $(y_k)_{k=0}^n$ in $AB$ such that $y_0=y$, $y_n \in Ax$, and for all $k \in \{0,...,n-1\}$, $y_{k}\mathcal{N}y_{k+1}$. If the sequence does not contain $x$, then $y\in H_{\mathcal{N}}(A,B)$. If $y_t=x$ for some $t \in \{0,...,n\}$, then $x \xrightarrow[]{B} A$ and Lemma \ref{lm:necessary:arrow} imply that that there exists a path in $AB$ from $x$ to some $z \in A$. Concatenating this path with the initial segment $(y_k)_{k=0}^t$ yields a path in $AB$ from $y$ to $z$. Hence, $y \in H_{\mathcal{N}}(A,B)$.

By Lemma \ref{lm:necessary:arrow}, Axiom 3 trivially holds. For Axiom 4, consider $x,y \in X$ and $(A,B) \in \mathcal{E}$ such that $x \neq y$, $x \xrightarrow[]{B} A$, and not $x\xrightarrow[]{B\backslash y} A\backslash y$. By Lemma \ref{lm:necessary:arrow}, we have $x \in B$, $\mathcal{N}^+_{\!\scriptscriptstyle AB}(x) \cap A \neq \emptyset$, and $\mathcal{N}^+_{\!\scriptscriptstyle (AB)\backslash y}(x) \cap (A \backslash y) = \emptyset$. Since $\mathcal{N}^+_{\!\scriptscriptstyle AB}(x) \cap A \neq \emptyset$, there is a sequence of mutually distinct alternatives $(x_k)_{k=0}^n$   in $AB$ such that $x_0 = x$, $\{x_0,...,x_n\} \cap A = \{x_n\}$, and for all $k \in \{0,...,n-1\}$, $x_k\mathcal{N}x_{k+1}$. However, since $\mathcal{N}^+_{\!\scriptscriptstyle (AB)\backslash y}(x) \cap (A\backslash y) = \emptyset$, for any such sequence, there exists $k \in \{1,...,n\}$ such that $x_k=y$. Therefore, by Lemma \ref{lm:necessary:arrow}, we have $x \xrightarrow[]{B\backslash y} \{y\}$ and  $y \xrightarrow[]{B\backslash x} A$.

For Axiom 5, it suffices to show that for all $(A,B) \in \mathcal{E}$ and $x,y \in A$, $x \succ y$ implies $\rho(x|A,B)=\rho(x|A\backslash y,By)$. Since $AB=(A\backslash y)\cup (By)$,  for all $C \subseteq AB$, $x \succ y$ implies that $x  = \max(\mathcal{N}^+_{\!\scriptscriptstyle AB}(C) \cap A; \succ)$ if and only if $x  = \max(\mathcal{N}^+_{\!\scriptscriptstyle (A\backslash y)\cup (By)}(C) \cap (A\backslash y); \succ)$. It then follows  from the definition of ABC that $\rho(x|A,B)=\rho(x|A\backslash y, By)$.

\bigskip

\noindent \textbf{(Sufficiency)} 
Consider a choice rule $\rho$. Throughout the proof of sufficiency, assume that $\rho$ satisfies Axioms 1-5. We first define the tuple $(\pi, \mathcal{N}, \succ)$ as follows. For each $x \in X$, let $\pi_{\scriptscriptstyle\! x}=\rho(x|\{x\},\emptyset) \in (0,1)$.  Define $\mathcal{N}$ such that $$\mathcal{N}= \mathcal{X} \cup  \{(x,y) \in X^2: x \neq y \text{~and~} \rho(y|\{y\},\emptyset) \neq \rho(y|\{y\},\{x\})\}.$$ Define $\succ$ such that for any distinct $x,y \in X$, $x \succ y$ if  $\rho(y|\{x,y\},\emptyset) \neq \rho(y|\{y\},\{x\})$. By Axiom 5, $x \succ y$ implies that for all menu $(A,B)$ with $x,y \in A$, $\rho(x|A,B)=\rho(x|A\backslash y,By).$

\begin{lemma}\label{lm:sufficient:side1} For all  $(A,B) \in \mathcal{E}$ and $x \in B$,   $x \xrightarrow[]{B} A$ implies $\mathcal{N}_{\!\scriptscriptstyle AB}^+(x) \cap A \neq \emptyset$. \end{lemma}

\begin{proof}[Proof of Lemma \ref{lm:sufficient:side1}]
Consider  $(A,B) \in \mathcal{E}$  and $x \in B$ such that $x  \xrightarrow[]{B} A$. It follows that $A \neq \emptyset$, since otherwise we have $\Phi_{\rho}(A,B) = \Phi_{\rho}(A,B\backslash x) = 1$. By repeated applications of Axiom 4, there is $y \in A$ such that $x \xrightarrow[]{B} \{y\}$. It then suffices to show $y \in \mathcal{N}_{\!\scriptscriptstyle By}^+(x)$.

First, if $|B|=1$, then $B=\{x\}$. In this case, $x  \xrightarrow[]{B} \{y\}$ means $x  \xrightarrow[]{\{x\}} \{y\}$, which 
implies $\Phi_{\rho}(\{y\},\{x\}) \neq \Phi_{\rho}(\{y\},\emptyset)$. It then follows that $\rho(y|\{y\},\{x\}) \neq \rho(y|\{y\}, \emptyset)$. By the definition of $\mathcal{N}$, we have $x\mathcal{N}y$. Thus, we have $y \in  \mathcal{N}_{\!\scriptscriptstyle \{x,y\}}^+(x) =\mathcal{N}_{\!\scriptscriptstyle By}^+(x)$.

				Next, assume by induction that when $|B| \le n$, $x  \xrightarrow[]{B} \{y\}$ implies $y \in \mathcal{N}_{\!\scriptscriptstyle By}^+(x)$. We show that it remains true when $|B| = n+1$. To see this, note that if there is $z \in B\backslash x$ such that $x  \xrightarrow[]{B\backslash z} \{y\}$, then by the induction hypothesis, $y \in \mathcal{N}_{\!\scriptscriptstyle (By)\backslash z}^+(x) \subseteq \mathcal{N}_{\!\scriptscriptstyle By}^+(x)$. Otherwise, for every $z \in B\backslash x$, not $x  \xrightarrow[]{B\backslash z} \{y\}$. By Axiom 4, for every $z \in B\backslash x$, $x \xrightarrow[]{B\backslash z} \{z\}$ and $z  \xrightarrow[]{B\backslash x} \{y\}$. By the induction hypothesis, for every $z \in B\backslash x$, we have $z \in \mathcal{N}_{\!\scriptscriptstyle B}^+(x) \subseteq \mathcal{N}_{\!\scriptscriptstyle By}^+(x)$ and $y \in  \mathcal{N}_{\!\scriptscriptstyle (By)\backslash x}^+(z) \subseteq \mathcal{N}_{\!\scriptscriptstyle By}^+(z)$. By the transitivity of $\mathcal{N}_{\!\scriptscriptstyle By}^+$, we have  $y \in \mathcal{N}_{\!\scriptscriptstyle By}^+(x)$. \end{proof}

\begin{lemma}\label{lm:sufficient:side2} For all $(A,B) \in \mathcal{E}$ and $x \in B$,   $\mathcal{N}(x) \cap A \neq \emptyset$ implies $x \xrightarrow[]{B} A$. 
\end{lemma}

\begin{proof}[Proof of Lemma \ref{lm:sufficient:side2}]
Since $\mathcal{N}(x) \cap A \neq \emptyset$, there exists $y \in A$ such that $x\mathcal{N}y$. It then follows by the definition of $\mathcal{N}$ that $x \xrightarrow[]{\{x\}} \{y\}$. By Axiom 3, we have $x \xrightarrow[]{B} A$. \end{proof}

\begin{lemma}\label{lm:sufficient:default:prob}
For all $(A,B) \in \mathcal{E}$, $\Phi_{\rho}(A,B) = \mathring{\pi}_{H_{\mathcal{N}}(A,B)}$.
\end{lemma}
\begin{proof}[Proof of Lemma \ref{lm:sufficient:default:prob}]
Consider some $(A,B) \in \mathcal{E}$. Note that $H_{\mathcal{N}}(A,B)$ admits a partition $\{A_k\}_{k=0}^n$ such that $A_0=A$ and for all $t \ge 1$, $$A_t=\{x \in B\backslash (\cup_{k=0}^{t-1} A_k): \mathcal{N}(x) \cap (\cup_{k=0}^{t-1} A_k) \neq \emptyset\}.$$
We now sequentially move alternatives in $B\cap H_{\mathcal{N}}(A,B)$ from the unavailable but observable set to the available set.  Specifically, we move all elements of $A_1$ first, then all elements of $A_2$, and so on.  The construction of the partition ensures that at each step, letting the current menu be $(\hat{A},\hat{B})$ and the alternative to be shifted from $\hat{B}$ to $\hat{A}$ be $\hat{x}$, we have $\mathcal{N}(\hat{x})\cap \hat{A}\neq \emptyset$.  By Lemma \ref{lm:sufficient:side2}, $\hat{x}\xrightarrow[]{\hat{B}}\hat{A}$. 
Axiom 2 then implies that $\Phi_{\rho}(\hat{A},\hat{B})=\Phi_{\rho}(\hat{A}\hat{x},\hat{B}\backslash \hat{x})$. 
Therefore, $\Phi_{\rho}(A,B)=\Phi_{\rho}\!\left(H_{\mathcal{N}}(A,B),\, B\backslash H_{\mathcal{N}}(A,B)\right).$  Since for all $x \in B\backslash H_{\mathcal{N}}(A,B)$, $\mathcal{N}^+_{AB}(x) \cap H_{\mathcal{N}}(A,B) = \emptyset$, Lemma \ref{lm:sufficient:side1} implies that for all $x \in B\backslash H_{\mathcal{N}}(A,B)$, not $x \xrightarrow[]{B\backslash H_{\mathcal{N}}(A,B)} H_{\mathcal{N}}(A,B)$. By Axiom 3, for any $C \subseteq B\backslash H_{\mathcal{N}}(A,B)$ with $x \in C$, not $x \xrightarrow[]{C} H_{\mathcal{N}}(A,B)$. Inductively, we can show $\Phi_{\rho}(H_{\mathcal{N}}(A,B), B\backslash H_{\mathcal{N}}(A,B)) = \Phi_{\rho}(H_{\mathcal{N}}(A,B), \emptyset)$. By Axiom 1 and the definition of $\pi$, we have $\Phi_{\rho}(A,B) =\Phi_{\rho}(H_{\mathcal{N}}(A,B), \emptyset) = \mathring{\pi}_{H_{\mathcal{N}}(A,B)}$. 
\end{proof}

\begin{lemma}\label{lm:sufficient:preference}
The revealed preference relation $\succ$ is transitive and asymmetric, and satisfies that for all distinct $x,y \in X$, either $x \succ y$ or $y \succ x$. 
\end{lemma}

\begin{proof}[Proof of Lemma \ref{lm:sufficient:preference}]
The asymmetry of $\succ$ is implied by Axiom 5. Consider distinct $x,y \in X$. Let $D=\{x,y\}$, and $\sigma$ be a distribution over subsets of $D$ such that for all $A \subseteq D$, $\sigma(A)=\sum\limits_{B\subseteq A: \mathcal{N}^+_{\scriptscriptstyle\! D}(B)=A} \pi_B \mathring{\pi}_{D\backslash B}$. By Lemma \ref{lm:sufficient:default:prob}, $D$ and $\sigma$ satisfy the primitive conditions in Lemma \ref{lm:general:preference}. Let $\succ_D$ be defined as in Lemma \ref{lm:general:preference}. Since $\succ$ coincides with $\succ_D$ on $D$ and is asymmetric, by Lemma \ref{lm:general:preference}, either $x \succ y$ or $y \succ x$. Next, let $\hat{D}=X$ and $\hat{\sigma}$ be the distribution over subsets of $\hat{D}$ such that for all $A \subseteq \hat{D}$, $\hat{\sigma}(A)=\sum\limits_{B\subseteq A: \mathcal{N}^+_{\scriptscriptstyle\! \hat{D}}(B)=A} \pi_B \mathring{\pi}_{\hat{D} \backslash B}$.  Lemma \ref{lm:sufficient:default:prob} again ensures that $\hat{D}$ and $\hat{\sigma}$ satisfy the primitive conditions in Lemma \ref{lm:general:preference}. Define the binary relation $\succ_{\hat{D}}$ as in Lemma \ref{lm:general:preference}. By Axiom 5, $x\succ y$ implies not $y \succ_{\hat{D}} x$. Thus, $\succ_{\hat{D}}$ is asymmetric. By Lemma \ref{lm:general:preference}, $\succ_{\hat{D}}$ is transitive and for any distinct $x,y \in X$, either $x \succ_{\hat{D}} y$ or $y \succ_{\hat{D}} x$. It follows that $\succ$ coincides with $\succ_{\hat{D}}$ on $\hat D= X$, and thus $\succ$ is also transitive. 
\end{proof}

To finish the proof of the sufficiency part, consider a menu $(A,B)$. By Lemma \ref{lm:sufficient:preference}, $\succ$ is a preference ordering. We can enumerate alternatives in $A$ such that $A=\{x_i\}_{i=1}^n$ and $x_1 \succ ... \succ x_n$. By the definition of $\succ$ and Axiom 5, it holds for all $k \in \{1,...,n\}$ that
\begin{equation*}
\begin{split}
S_k:=\sum_{t=1}^k \rho(x_t|A,B) & =\sum_{t=1}^k \rho(x_t|A\backslash \{x_{k+1},...,x_n\},B\cup \{x_{k+1},...,x_n\}) \\
& = 1- \Phi_{\rho}(A\backslash \{x_{k+1},...,x_n\},B\cup \{x_{k+1},...,x_n\}). 
\end{split}
\end{equation*}
Denote by $A_k=A\backslash \{x_{k+1},...,x_n\}$ and $B_k=B\cup \{x_{k+1},...,x_n\}$.  By Lemma \ref{lm:sufficient:default:prob}, we have for each $k \in \{1,...,n\}$,
\begin{equation}\label{eq:sufficient:identification} 
S_k  =   1- \mathring{\pi}_{H_{\mathcal{N}}(A_k, B_k)}.
\end{equation} Note that the ABC $\hat{\rho}$ that is represented by $(\pi, \mathcal{N}, \succ)$ must satisfy the system of equations (\ref{eq:sufficient:identification}), and that this system uniquely pins down a random choice rule (i.e., for each $k \in \{1,...,n\}$, $\rho(x_k|A,B)=S_k-S_{k-1}$, where $S_0:=0$). Therefore, $\rho$ coincides with $\hat{\rho}$. That is, $\rho$ is the ABC represented by $(\pi, \mathcal{N}, \succ)$. The uniqueness of the parameters $\pi$, $\mathcal{N}$, and $\succ$ has already been shown in Section \ref{subsec:axiom} and is not repeated here.   \end{proof}

\bigskip

\begin{proof}[Proof of Theorem \ref{thm:main:limdata}]
The necessity of the axioms can be shown by analogy to Theorem \ref{thm:main}.  For sufficiency, fix a $\mathcal{E}^O$-choice rule $\rho$. Throughout the proof, assume that Axioms L1-L4 hold for $\rho$. We also write $\Phi_{\rho}(A)$ for $\Phi_{\rho}(A,A^c)$ throughout this proof for simplicity. Define the associative network $\mathcal{N}$ such that $x\mathcal{N}y$ if either $x=y$, or $x\neq y$ and $\Phi_{\rho}(\{x,y\})=\Phi_{\rho}(\{y\}).$ We first prove the following two lemmas.

\begin{lemma}\label{lm:limitdata:transitive}
The associative network $\mathcal{N}$ is transitive and satisfies that for all $A \in \mathcal{M}$ and $x \in A^c$, $\Phi_{\rho}(A)=\Phi_{\rho}(Ax)$ if and only if  $x\mathcal{N}y$ for some $y \in A$.
\end{lemma}

\begin{proof}[Proof of Lemma \ref{lm:limitdata:transitive}]
Consider mutually distinct $x,y,z \in X$ such that $x\mathcal{N}y$ and $y\mathcal{N}z$. By the definition of $\mathcal{N}$, we have $\Phi_{\rho}(\{x,y\}) = \Phi_{\rho}(\{y\})$ and $\Phi_{\rho}(\{y,z\}) = \Phi_{\rho}(\{z\})$. Applying Axiom L3 yields $\Phi_{\rho}(\{x,y,z\}) = \Phi_{\rho}(\{y,z\}) = \Phi_{\rho}(\{z\})$. By Axiom L2, $\Phi_{\rho}(\{z\}) \ge \Phi_{\rho}(\{x,z\}) \ge \Phi_{\rho}(\{x,y,z\})$ and thus $\Phi_{\rho}(\{z\}) = \Phi_{\rho}(\{x,z\})$. Hence $x\mathcal{N}z$, establishing the transitivity of $\mathcal{N}$. For the second part of the statement, note that Axiom L3 implies that for all $A \in \mathcal{M}$ and $x \in A^c$, $\Phi_{\rho}(A)=\Phi_{\rho}(Ax)$ if and only if for some $y \in A$, $\Phi_{\rho}(\{y\})=\Phi_{\rho}(\{x,y\})$, i.e., $x \mathcal{N} y$.
\end{proof}

\begin{lemma}\label{lm:limitdata:asso}
For all $A \in \mathcal{M}$, let $B=\{x \in A^c: \mathcal{N}(x) \cap A \neq \emptyset\}$. Then $\Phi_{\rho}(A)=\Phi_{\rho}(AB)$ and for all $x \in (AB)^c$, $\Phi_{\rho}(A) \neq \Phi_{\rho}(ABx)$. 
\end{lemma}

\begin{proof}[Proof of Lemma \ref{lm:limitdata:asso}]
Consider a menu $(A,A^c)$ and divide $A^c$ to a binary partition $\{B,C\}$ such that for all $x \in B$, $\mathcal{N}(x) \cap A \neq\emptyset$ and for all $y \in C$, $\mathcal{N}(y) \cap A = \emptyset$. 
Consider any $D \in \mathcal{M}$ such that $A \subseteq D \subseteq AB$. For all $x \in B$, we have $\mathcal{N}(x) \cap D \neq\emptyset$, and for all $y \in C$, by the transitivity of $\mathcal{N}$ (Lemma \ref{lm:limitdata:transitive}), we have $\mathcal{N}(y) \cap D = \emptyset$. By Lemma \ref{lm:limitdata:transitive}, starting from menu $(A,A^c)$, we can consecutively move alternatives from $B$ to the available set without affecting the choice probability of the default option. This leads to $\Phi_{\rho}(A)=\Phi_{\rho}(AB)$. For any $x \in C = (AB)^c$, by  Lemma \ref{lm:limitdata:transitive} and the fact that $\mathcal{N}(x) \cap AB = \emptyset$, we have  $\Phi_{\rho}(A) = \Phi_{\rho}(AB) \neq \Phi_{\rho}(ABx)$.
\end{proof}

To proceed, let $\mathcal{I}$ be the symmetric part of $\mathcal{N}$. Since $\mathcal{N}$ is transitive and reflexive, $\mathcal{I}$ is also transitive and reflexive. For each $x \in X$, define a subset $[x] \subseteq X$ such that $y \in [x]$  if $x \mathcal{I} y$. Note that for all $x \in X$, $x \in [x]$. By the symmetry of $\mathcal{I}$, we have $y \in [x]$ if and only if $[x]=[y]$, and whenever $[x] \neq [z]$, $[x] \cap [z] = \emptyset$. The network $\mathcal{N}$ induces a binary relation $\rhd$ on $\{[x]\}_{x \in X}$ such that $[x] \rhd [y]$ if and only if $[x] \neq [y]$ and $\hat{x} \mathcal{N} \hat{y}$ for some $\hat{x} \in [x]$ and $\hat{y} \in [y]$. Note that $[x] \rhd [y]$ also implies $\hat{x} \mathcal{N} \hat{y}$ for all $\hat{x} \in [x]$ and $\hat{y} \in [y]$. By definition, $\rhd$ is asymmetric. 
The transitivity of $\mathcal{N}$ implies that $\rhd$ is transitive.

Now we define $\pi$. For each $A \in \mathcal{M}$, let $\overleftarrow{\mathcal{N}}(A)= \{x \in X: \mathcal{N}(x) \cap A \neq \emptyset\}$. Note that $A \subseteq \overleftarrow{\mathcal{N}}(A)$. For each $x \in X$, define $$\pi_x=1-\left(\frac{\Phi_{\rho}(\overleftarrow{\mathcal{N}}([x]))}{\Phi_{\rho}(\overleftarrow{\mathcal{N}}([x])\backslash [x])} \right)^{\frac{1}{|[x]|}}.$$  By the definition, we have  \begin{equation}\label{eq:sufficiency:limitdata:pi}
\mathring{\pi}_{[x]}=\frac{\Phi_{\rho}(\overleftarrow{\mathcal{N}}([x]))}{\Phi_{\rho}(\overleftarrow{\mathcal{N}}([x])\backslash [x])}.   
\end{equation} 
We show that $\pi_x \in (0,1)$ for every $x \in X$. By the asymmetry of $\rhd$, for all $y \in [x]$, we have $\mathcal{N}(y) \cap \left( \overleftarrow{\mathcal{N}}([x])\backslash [x] \right) = \emptyset$. Hence, by Lemma \ref{lm:limitdata:transitive}, for any $y \in [x]$, $\Phi_{\rho}(\overleftarrow{\mathcal{N}}([x])\backslash [x]) \neq \Phi_{\rho}((\overleftarrow{\mathcal{N}}([x])\backslash [x]) \cup \{y\})$. Axiom L2 then implies that $\Phi_{\rho}(\overleftarrow{\mathcal{N}}([x])\backslash [x]) > \Phi_{\rho}( (\overleftarrow{\mathcal{N}}([x])\backslash [x] ) \cup \{y\})$, which further implies $\Phi_{\rho}(\overleftarrow{\mathcal{N}}([x])\backslash [x]) > \Phi_{\rho}(\overleftarrow{\mathcal{N}}([x]))>0$.  This ensures $\pi_x \in (0,1)$.

While we have defined a specific $\pi$, the remaining proof works for any $\pi: X \rightarrow (0,1)$  that satisfies (\ref{eq:sufficiency:limitdata:pi}) for all $x \in X$. For any $A \in \mathcal{M}$, by Lemma \ref{lm:limitdata:transitive},  $\overleftarrow{\mathcal{N}}(A)$ is association-proof. Since $\rhd$ is asymmetric and transitive, there exists $x \in \overleftarrow{\mathcal{N}}(A)$ such that for all $y \in \overleftarrow{\mathcal{N}}(A)$, not $[x] \rhd [y]$. The set $\overleftarrow{\mathcal{N}}(A)\backslash [x]$ is also association-proof. By Axiom L1, we have \begin{equation*}
\begin{split}
\Phi_{\rho}(\overleftarrow{\mathcal{N}}(A))    & = \frac{\Phi_{\rho}(\overleftarrow{\mathcal{N}}(A))}{\Phi_{\rho}(\overleftarrow{\mathcal{N}}(A)\backslash [x])} \Phi_{\rho}(\overleftarrow{\mathcal{N}}(A)\backslash [x]) \\ & =  \frac{\Phi_{\rho}(\overleftarrow{\mathcal{N}}(x))}{\Phi_{\rho}(\overleftarrow{\mathcal{N}}(x)\backslash [x])} \Phi_{\rho}(\overleftarrow{\mathcal{N}}(A)\backslash [x])  = \mathring{\pi}_{[x]}   \Phi_{\rho}(\overleftarrow{\mathcal{N}}(A)\backslash [x]). 
\end{split}
\end{equation*} By a simple induction, we have $\Phi_{\rho}(\overleftarrow{\mathcal{N}}(A)) = \mathring{\pi}_{\overleftarrow{\mathcal{N}}(A)}$.

Next, we define the preference relation $\succ$ such that for distinct $x,y \in X$, $x \succ y$ if for some $A \in \mathcal{M}$, $\rho(y|A,A^c) \neq \rho(y|A\backslash x, (A\backslash x)^c)$. By Axiom L4, $\succ$ is asymmetric. Similar to the proof of Lemma \ref{lm:sufficient:preference},  we can apply Lemma \ref{lm:general:preference} to show that $\succ$ is a preference ordering.  Again by Axiom L4, for all distinct $x,y \in X$, $x\succ y$ implies $\rho(x|A,A^c)=\rho(x|A\backslash y, (A\backslash y)^c)$ for all $A \in \mathcal{M}$.

Finally, we show that $\rho$ can be represented by $(\pi, \mathcal{N}, \succ)$ as a transitive $\mathcal{E}^O$-ABC. For any nonempty $A \in \mathcal{M}$, let $A=\{x_1,...,x_n\}$ such that $x_1 \succ x_2\succ...\succ x_n$. Similar to the proof of the sufficiency part of Theorem \ref{thm:main}, it suffices to show that for all $k \in \{1,...,n\}$, $\sum_{t=1}^k \rho(x_t|A,A^c)= 1-\Phi_{\rho}(\overleftarrow{\mathcal{N}}(\{x_1,...,x_k\}))$. This can be shown by   Lemma \ref{lm:limitdata:asso} and the fact that $\sum_{t=1}^k \rho(x_t|A,A^c) = \sum_{t=1}^k \rho(x_t|\{x_1,...,x_k\},\{x_1,...,x_k\}^c)$.   \medskip

\noindent \textbf{Uniqueness.} The uniqueness of $\succ$ and $\mathcal{N}$ has already been shown in Section \ref{subsec:limit:data:observability}. 
It remains to show the uniqueness property of $\pi$. As shown by the proof above, given $\mathcal{N}$ and $\succ$, any $\pi: X \rightarrow (0,1)$ that satisfies equation (\ref{eq:sufficiency:limitdata:pi}) for all $x \in X$ works for the representation. Each such $\pi$ yields the same value for $\mathring{\pi}_{[x]}$. Conversely, the value of $\mathring{\pi}_{[x]}$ is uniquely pinned by equation (\ref{eq:sufficiency:limitdata:pi}). This finishes the proof of Theorem \ref{thm:main:limdata}. \end{proof}

\bibliographystyle{ecta}
\setstretch{0.8}
\bibliography{association}
	
\newpage 
\baselineskip=18pt
\appendix

\section*{Online Appendix (for online publication only)}

\setcounter{page}{1} 

This online appendix to ``\textit{Associative Networks in Decision Making}'' is organized as follows. In Section OA-1, we characterize ABCs restricted on the domain $\mathcal{E}^F=\{(A,B) \in \mathcal{E}: B = \emptyset\}$ where every
observable alternative is available. Section OA-2 
revisits the issue that unavailable but observable alternatives may be imperfectly recorded in the data and provides formal analysis of two  assumptions discussed in Section \ref{subsec:unobserve}. In Section OA-3, we characterize properties of the random attention rule implied by our ABC model and compare it with other random attention models in the literature. Section OA-4 studies the extension  in which the DM's associative network can be random. 
In Section OA-5, we discuss more general models of initial attention distributions. In Section OA-6, we discuss the identification of our model without the default option. Section OA-7 studies the application of platform network design. The omitted proofs in Section \ref{sec:app} are contained in Section OA-8.

\subsection*{OA-1. Restricted Domain}
\textbf{Identification of $\mathcal{E}^F$-ABCs.} Following the terminology introduced in Section \ref{subsec:limit:data:observability}, we consider a $\mathcal{E}^F$-choice rule $\rho$ where $\mathcal{E}^F := \{(A,B)\in\mathcal{E} : B=\emptyset\}$. Suppose that $\rho$ is a $\mathcal{E}^F$-ABC represented by $(\pi, \mathcal{N}, \succ)$. We first show that $\pi$ and $\succ$ can be uniquely identified.  For all $x \in X$, we have $\pi(x)=\rho(x|\{x\}, \emptyset)$. For all $x,y \in X$,  $x\succ y$  implies $\rho(x|\{x,y\},\emptyset) \ge \pi_x = \rho(x|\{x\},\emptyset)$ and 
$\rho(y|\{x,y\},\emptyset) \le \mathring{\pi}_x\pi_y < \rho(y|\{y\},\emptyset)$. Therefore,  $x \succ y$ if and only if  \begin{equation}\label{eq:identify_preference_F}
\rho(y|\{x,y\},\emptyset) < \rho(y|\{y\},\emptyset).
\end{equation}

However, the associative network $\mathcal{N}$ may not be uniquely identified. In what follows, we provide a partial identification of $\mathcal{N}$ by showing that we can identify the minimum associative network which is valid for representing the random choice rule on $\mathcal{E}^F$. We illustrate the idea of identification by the following examples.

\begin{example} \label{eg:two-alternatives}
\emph{Let $X=\{x,y\}$. If $\rho(y|\{x,y\},\emptyset)> 0 = \rho(x|\{x,y\},\emptyset)$, then any  $(\pi, \mathcal{N}, \succ)$ that represents $\rho$ as a $\mathcal{E}^F$-ABC must satisfy $x\mathcal{N}y$: Since $x$ is never chosen in  menu $(\{x,y\}, \emptyset)$, the attention to $x$ must prompt the DM to consider some better alternative, which has to be $y$.  \qed}
\end{example}

Note that Example \ref{eg:two-alternatives} also demonstrates that the associative network cannot be fully identified: Whether $y\mathcal{N}x$ is unclear as it does not affect the DM's choice probabilities in any menu. The next example generalizes the identification strategy in Example \ref{eg:two-alternatives}.

\begin{example} \label{eg:four-alternatives}
\emph{Consider mutually distinct $x,y,z,w \in X$ and a $\mathcal{E}^F$-ABC $\rho$ such that:} 
\smallskip 

\emph{\noindent (i) In the menu $(\{x,y,z,w\}, \emptyset)$, only $x$ is chosen with positive probability;}

\emph{\noindent (ii) In the menu $(\{x,y,w\}, \emptyset)$, only $x$ and $w$ are chosen with positive probability;}

\emph{\noindent (iii) In the menu $(\{x,y\}, \emptyset)$, only $x$ is chosen with positive probability. \smallskip}

\noindent \emph{Condition (i) implies that $x$ is better than $y,z$ and $w$. By conditions (ii) and (iii),  one can infer that $x$ is directly associated with $y$, and neither $y$ nor $x$ is associated with $w$, since otherwise the attention to $w$ would prompt the consideration of $x$ and thus blocks the choice of $w$. Now, by adding $z$ to menu $(\{x,y,w\}, \emptyset)$, $w$ becomes unchosen. Hence, $z$ must be directly associated with $w$. That is, for any tuple $(\pi, \mathcal{N}, \succ)$ that represents $\rho$ as a  $\mathcal{E}^F$-ABC, we must have $w\mathcal{N}z$.  \qed}
\end{example}

The two examples above suggest the following identification of the associative network. For a given random choice rule $\rho$, define 
\begin{equation}\label{eq:def:Nrho}
\begin{split}
&	\mathcal{N}[\rho]:= \mathcal{X} \cup \bigl\{(x,y) \in X^2: x \neq y \text{ and }  \{y\}=\{w \in X: \rho(w|\{x,y\},\emptyset) > 0\} \bigl\}~ \cup  \\ 
&	\Bigl\{(x,y) \in X^2: \exists A \subseteq X \text{ and } z \in X\backslash x \text{~such that~} \{x,z\}= \{w \in X: \rho(w|Ax,\emptyset) > 0\} \\ 
&\text{and~}  \{z\} = \{w \in X: \rho(w|A\cup \{x,y\},\emptyset)> 0\} =\{w \in X: \rho(w|A,\emptyset)> 0\} \Bigl\}.
\end{split}
\end{equation}  Consider any $A\in\mathcal{M}$ and $z\in X$ that satisfy the condition in (\ref{eq:def:Nrho}). Then $z$ is revealed to be the best alternative in $A\cup\{x,y\}$, and, moreover, attention to any alternative induces eventual consideration of $z$ when the menu is either $(A\cup\{x,y\},\emptyset)$ or $(A,\emptyset)$.  Since $x$ is chosen with positive probability in the menu $(A\cup\{x\},\emptyset)$, attention to $x$ cannot trigger consideration of $z$ in that menu. It follows that no alternative in $A$ is directly associated with $x$. Therefore, to ensure that attention to $x$ can nevertheless lead to consideration of $z$ in $(A\cup\{x,y\},\emptyset)$, it must be that $y$ is directly associated with $x$. Hence, the DM's associative network must contain $\mathcal{N}[\rho]$.\footnote{An equivalent way of identifying the set $\mathcal{N}[\rho]$ is to observe changes in the choice frequencies of the alternatives. For two distinct alternatives $x$ and $y$, parallel to the definition of $\mathcal{N}[\rho]$, we can include $(x,y)$ in $\mathcal{N}[\rho]$ if one of the following two situations occurs. The first situation is that $\rho(y|\{y\})<\rho(y|\{x, y\})$.  The second situation is that there exists a set $A \subseteq X\backslash \{x,y\}$ and some $z \in A$ such that (1) $\rho(z|A)>\rho(z|A\backslash \{w\})$ for each $w \in A\backslash \{z\}$, (2) $\rho(z|A\cup\{y\})>\rho(z|A)$, (3) $\rho(z|A\cup\{x\}) = \rho(z|A)$, and (4) $\rho(z|A\cup\{x,y\}) > \rho(z|A\cup \{y\})$. The latter situation can be interpreted in exactly the same way as the second part of the definition of $\mathcal{N}[\rho]$. This method complements the original definition of $\mathcal{N}[\rho]$ for empirically identifying the network when the data sample is small.}

Indeed, $\mathcal{N}[\rho]$ is the \emph{minimal} associative network such that $(\pi,\mathcal{N}[\rho],\succ)$ represents $\rho$ as a $\mathcal{E}^F$-ABC. We state this as a proposition below.

\begin{proposition}\label{prop:oa:EF:id}
For any $\mathcal{E}^F$-choice rule $\rho$ that can be represented by $(\pi,\mathcal{N},\succ)$ as a $\mathcal{E}^F$-ABC, $\pi$ and $\succ$ can be uniquely identified. Furthermore, $\mathcal{N}[\rho] \subseteq \mathcal{N}$, and $(\pi,\mathcal{N}[\rho],\succ)$ also represents $\rho$ as a $\mathcal{E}^F$-ABC.    
\end{proposition}

\begin{proof}[Proof of Proposition \ref{prop:oa:EF:id}.] 
Consider a $\mathcal{E}^F$-ABC $\rho$  represented by $(\pi, \mathcal{N}, \succ)$. By the arguments before Proposition \ref{prop:oa:EF:id}, it remains to show that $\rho$ can also be represented by $(\pi, \mathcal{N}[\rho], \succ)$. To simplify the notation, let $\mathcal{W}=\mathcal{N}[\rho]$.  Since $\mathcal{W}\subseteq \mathcal{N}$, it suffices to show that for all  $(A,\emptyset) \in \mathcal{E}^F$, $x \in A$ and $C \subseteq A$, if  $x= \max(\mathcal{N}^+_{\scriptscriptstyle\! A}(C) \cap A; \succ)$, then $x \in \mathcal{W}^+_{\scriptscriptstyle\! A}(C)$. The case $x \in C$ is trivial. Suppose that $x \notin C$. Since $x = \max(\mathcal{N}^+_{\scriptscriptstyle\! AB}(C) \cap A; \succ)$, there exists a sequence of mutually distinct alternatives $(x_k)_{k=0}^{n}$, where $n \in \mathbb{N}_+$, such that $x_0 \in C$, $x_{n} = x = \max(\{x_k\}_{k=0}^{n}; \succ)$, and for all $t \le n-1$, $\mathcal{N}(x_t) \cap \{x_k\}_{k=t}^{n} = \{x_t, x_{t+1}\}.$ It then suffices to show that for all $k \le n-1$, $x_k \mathcal{W} x_{k+1}.$ We prove this by induction. If $n=1$, then $x_1 \succ x_0$ and $x_0 \mathcal{N} x_1$ imply $\rho(x_0|\{x_0, x_1\}, \emptyset) = 0$, which further implies $x_0 \mathcal{W} x_1$ by the definition of $\mathcal{W}$. Suppose that our hypothesis is true for all $n \le m$. Consider the case where $n = m +1$. If the second best alternative in $\{x_k\}_{k=0}^{n}$ is not $x_0$, then applying our induction hypothesis twice delivers the result. If the second best alternative in $\{x_k\}_{k=0}^{n}$ is $x_0$, then by our assumption on $\mathcal{N}$, $x_{n}$ is the only alternative chosen with positive probabilities in menus $(\{x_k\}_{k=0}^{n}, \emptyset)$ and $(\{x_k\}_{k=2}^{n}, \emptyset)$, and $x_0$ and $x_{n}$ are the only two alternatives chosen with positive probabilities in the menu $(\{x_0\} \cup \{x_t\}_{t=2}^{n}, \emptyset)$. Therefore, by the definition of $\mathcal{W}$, we have $x_0 \mathcal{W} x_1$. By applying the induction hypothesis to $(x_k)_{k=1}^{n}$,  we are done. 
\end{proof}  \medskip

\noindent \textbf{Axiomatization.} Since the domain $\mathcal{E}^F$ is the standard one studied in the literature, we axiomatize $\mathcal{E}^F$-ABCs in the remainder of this section for completeness. The first axiom is Default Independence, introduced in Section \ref{subsec:axiom}.

\bigskip
\noindent\textbf{Axiom A1\textemdash Default Independence:} For all $x \in X$ and $A,B\in \mathcal{M}$ with $x \in A\cap B$:  $$\frac{\Phi_{\rho}(A,\emptyset)}{\Phi_{\rho}(A\backslash x,\emptyset)} = \frac{\Phi_{\rho}(B,\emptyset)}{\Phi_{\rho}(B\backslash x,\emptyset)}.$$  
\smallskip
		
		For a given choice rule $\rho$ and menu $(A,\emptyset)$, let $c_{\scriptscriptstyle\!\rho}(A,\emptyset):=\{x\in A: \rho(x|A,\emptyset) > 0\}$ be the {set of chosen alternatives} in $A$, i.e., those in $A$ that are chosen with positive probabilities. We impose the next two axioms on the set of chosen alternatives.


		\bigskip	
		\noindent\textbf{Axiom A2\textemdash Sen's $\alpha$:} For all $A,B\in \mathcal{M}$,   $B\subseteq A$ implies $c_{\scriptscriptstyle\!\rho}(A,\emptyset) \cap B \subseteq c_{\scriptscriptstyle\! \rho}(B,\emptyset)$.

		\noindent\textbf{Axiom A3\textemdash Reducibility:} For all $A \in \mathcal{M}$, if for every $x \in A$, $c_{\scriptscriptstyle\!\rho}(A, \emptyset)\neq c_{\scriptscriptstyle\! \rho}(A\backslash x, \emptyset)$, then $c_{\scriptscriptstyle\! \rho}(A, \emptyset) = A$. \bigskip


		Axiom A2 states that if an alternative is selected from a larger menu, it must also be selected from any smaller menu that contains it. To interpret, if a particular alternative  $x$  is not chosen in a smaller menu, then given the presence of more competitive alternatives in a larger menu, it should also remain unselected. In our context, if an alternative is not chosen, its consideration must lead to the consideration of a better alternative. Consequently, in a larger menu, the superior alternative remains to be associated with $x$ and thus blocks the choice of $x$.   \medskip

	The contrapositive of Axiom A3 states that if not all alternatives are selected, then there exists an unselected alternative whose removal does not alter the set of chosen alternatives.  To illustrate this axiom, consider a menu $(\{x,y,z\}, \emptyset)$ where only $x$ is chosen. As both $y$ and $z$ are unselected, their consideration must lead to the consideration of a superior alternative in this  menu, which has to be $x$. If the removal of $y$ results in a change in the set of chosen alternatives such that $z$ becomes chosen, then $x$ must be associated with $z$ through $y$, and $x$ must be directly associated with $y$. In this scenario, the deletion of $z$ does not alter the association relation between $x$ and $y$, and thus does not affect the set of chosen alternatives. In summary, Axiom A3 establishes the existence of an unselected alternative (if not all alternatives are chosen) whose removal does not impact the association relation among the remaining alternatives, thereby preserving the set of chosen alternatives.  \medskip

		For any  $x \in X$ and $A \in \mathcal{M}$, we say that $x$ is \textit{associatively independent of $A$}, denoted by $x \vdash A$, if $x \notin A$ and for all $y \in A$, $\rho(y|A,\emptyset)=\rho(y|Ax, \emptyset)$. Note that according to this definition, for all $x \in X$, we have $x \vdash \emptyset$. \bigskip

		\noindent\textbf{Axiom A4\textemdash Weak I-Independence:} For all $x \in X$ and $A, B \in \mathcal{M}$,  if $x \vdash A$ and $x \vdash B$, then $x \vdash A\cup B$. \bigskip

		Axiom A4 posits that if $x$ is associatively independent of both $A$ and $B$, then it is also associatively independent of their union. Notably, this axiom can be implied by the I-Independence axiom of \citeAppendix{OA-ecta2014consideration}.  According to the I-Independence axiom, if $x$ does not affect the probability of selecting alternative $y$ in a particular menu,  then it should not impact the probability of choosing $y$ in \textit{every} menu.   \medskip

For any two alternatives $x$ and $y$, we say that $x$ \textit{weakly dominates} $y$, denoted by $x \trianglerighteq y$, if there is a menu $A$ such that $y \in A$ and $c_{\scriptscriptstyle\! \rho}(A,\emptyset)=\{x\}$.  Note that every alternative weakly dominates itself. \bigskip


		\noindent\textbf{Axiom A5\textemdash Dominance Asymmetry:} For all  $x, y, z \in X$ and $A,B \in \mathcal{M}$ such that $x \neq z$,  $y \in c_{\scriptscriptstyle\!\rho}(A,\emptyset)$ and $z \in c_{\scriptscriptstyle\!\rho}(B,\emptyset)$, if $x \trianglerighteq y$, then  $$\rho(z|A,\emptyset) \neq \rho(z|A\backslash y,\emptyset) \Rightarrow \rho(x|B,\emptyset)  = \rho(x|B\backslash z,\emptyset).$$

To understand Axiom A5, observe that since $y$ is chosen in $(A, \emptyset)$, the attention to it cannot prompt the DM to consider any better alternative in $A$. Hence, the change in the choice probability of $z$ by the deletion of $y$ can only be attributed to the fact that  the presence of $y$ hinders the choice of $z$, i.e., $y$ is better than $z$.  It follows that the alternative $x$ that weakly dominates $y$ must also be better than $z$. Therefore, the inverse cannot occur, that is, deleting $z$ from any menu in which $z$ is chosen will not affect the choice probability of $x$. \bigskip

\begin{theorem}\label{thm:oa:EF:axiom}
A $\mathcal{E}^F$-choice rule $\rho$ is a  $\mathcal{E}^F$-ABC if and only if it satisfies Axioms A1-A5.   
\end{theorem}

\begin{proof}[Proof of Theorem \ref{thm:oa:EF:axiom}]
\textbf{Necessity.} Consider a $\mathcal{E}^F$-ABC $\rho$ represented by $(\pi, \mathcal{N}, \succ)$. Axiom A1 holds trivially. For Axiom A2, consider $A,B \in \mathcal{M}$ with $B \subseteq A$. If for some $x \in B$,  $x \notin c_{\scriptscriptstyle\! \rho}(B,\emptyset)$, then we have $x \neq \max(\mathcal{N}^+_{\scriptscriptstyle\! B}(x); \succ)$. Since $\mathcal{N}^+_{\scriptscriptstyle\! B}(x) \subseteq \mathcal{N}^+_{\scriptscriptstyle\! A}(x)$, we have $x \neq \max(\mathcal{N}^+_{\scriptscriptstyle\! A}(x); \succ)$. Thus, $x \notin c_{\scriptscriptstyle\! \rho}(A,\emptyset)$.\medskip

			For Axiom A3, consider $A \in \mathcal{M}$ such that $c_{\scriptscriptstyle\!\rho}(A,\emptyset) \neq A$. Let $c_{\scriptscriptstyle\!\rho}(A,\emptyset)=\{x_1,...,x_n\}$ such that for all $k \in \{1,...,n-1\}$, $x_k \succ x_{k+1}$. 
			Consider a partition $\{B_k\}_{k=1}^n$ of $A$ such that  for every $k$, $B_k=\{y\in A:  \mathcal{N}^+_{\scriptscriptstyle\! A}(y) \cap \{x_1,...,x_k\}=\{x_k\}\}.$ Note that each $B_k$ contains $x_k$, and for all $y \in B_k \backslash x_k$, we have $y \notin c_{\scriptscriptstyle\! \rho}(A,\emptyset)$, $x_k \succ y$ and $x_k \in \mathcal{N}^+_{\scriptscriptstyle\! B_k}(y)$. Consider some $k$ such that $B_k \backslash x_k \neq \emptyset$. It is easy to show that there is an alternative $y \in B_k\backslash x_k$ such that for all $z \in B_k \backslash y$, $x_k \in \mathcal{N}^+_{\scriptscriptstyle\! B_k\backslash y}(z)$. Deleting $y$ from menu $(A,\emptyset)$ will not affect the choices. \medskip


			For Axiom A4, note that $x\vdash A$ if and only if  for all $y \in c_{\scriptscriptstyle\! \rho}(A,\emptyset)$, $y \succ x$, and for all $z \in A$, $(x,z) \notin \mathcal{N}$. Since $c_{\rho}(AB) \subseteq c_{\rho}(A) \cup c_{\rho}(B)$, it follows that $x \vdash A$ and $x \vdash B$ imply $x \vdash AB$.  \medskip

			For Axiom A5, it suffices to show that for two distinct alternatives $x$ and $y$, if $x \in c_{\scriptscriptstyle\! \rho}(A,\emptyset)$ and $\rho(y|A,\emptyset) \neq \rho(y|A\backslash x,\emptyset)$, then $x \succ y$. Let $c_{\scriptscriptstyle\!\rho}(A,\emptyset)=\{x_1,...,x_n\}$ such that for all $k \in \{1,...,n-1\}$, $x_k \succ x_{k+1}$.  Consider the partition $\{B_k\}_{k=1}^n$ constructed in the proof for Axiom A3. We have $x=x_k$ for some $k$. If $y \notin c_{\scriptscriptstyle\! \rho}(A,\emptyset)$, then $\rho(y|A,\emptyset) \neq \rho(y|A\backslash x,\emptyset)$ implies $y \in B_k$, and thus $x \succ y$. If $y \in c_{\scriptscriptstyle\! \rho}(A,\emptyset)$, then $\rho(y|A,\emptyset) \neq \rho(y|A\backslash x,\emptyset)$ implies $y=x_{t}$ for some $t > k$, and thus $x \succ y$. 
			\bigskip

			\noindent \textbf{Sufficiency.} Through out the proof of sufficiency, we assume that Axioms A1-A5 hold. For the attention probability function $\pi$, let $\pi_{x} = \rho(x|\{x\}, \emptyset)$ for every $x\in X$. For the preference ordering $\succ$, let $x \succ y$ if $x \neq y$, $x \in c_{\scriptscriptstyle\! \rho}(\{x,y\}, \emptyset)$ and $\rho(y|\{x,y\},\emptyset) \neq \rho(y|\{y\},\emptyset)$.  For the associative network $\mathcal{N}$, let $(x,y) \in \mathcal{N}$ if and only if either (i) $x=y$, or (ii) $x\neq y$ and there exists $A \in \mathcal{M}$ such that $x \vdash A$ and $x \notin c_{\scriptscriptstyle\! \rho}(A\cup \{x,y\}, \emptyset)$. We proceed with a sequence of claims.\bigskip

\noindent \textbf{Claim 1.} \textit{For all $x \in X$ and $A \in \mathcal{M}$, if $x \vdash A$, then $x \in c_{\scriptscriptstyle\! \rho}(Ax, \emptyset)$.}

\begin{proof}
By $x \vdash A$, we have $\sum_{y \in A}\rho(y|A,\emptyset)=\sum_{y \in A}\rho(y|Ax,\emptyset)$. Since $\Phi_{\rho}(Ax,\emptyset)<\Phi_{\rho}(A, \emptyset)$, we have $\rho(x|Ax,\emptyset)\neq 0$, i.e., $x \in c_{\scriptscriptstyle\! \rho}(Ax,\emptyset)$. 
\end{proof}

\noindent \textbf{Claim 2.} \textit{The binary relation $\succ$ is a preference ordering and satisfies that for all $x,y \in X$ and $A \in \mathcal{M}$, if $x \succ y$ and $y \in c_{\scriptscriptstyle\! \rho}(A,\emptyset)$, then $\rho(x|A,\emptyset)=\rho(x|A\backslash y, \emptyset)$.}

\begin{proof}
The claim that $x  \succ y$ and $y \in c_{\scriptscriptstyle\! \rho}(A,\emptyset)$	imply $\rho(x|A,\emptyset)=\rho(x|A\backslash y, \emptyset)$ follows from the definition of $\succ$ and Axiom A5. Showing that $\succ$ is well-defined for each distinct pair of alternatives and asymmetric is trivial. To see that $\succ$ is transitive, suppose to the contrary that there are three mutually distinct alternatives $x, y$ and $z$ such that $x \succ y$, $y \succ z$, and $z \succ x$. By symmetry, we can focus on three representative cases, where in case 1, $c_{\scriptscriptstyle\! \rho}(\{x,y,z\}, \emptyset)=\{x,y,z\}$, in case 2, $c_{\scriptscriptstyle\! \rho}(\{x,y,z\}, \emptyset)=\{x,y\}$, and in case 3, $c_{\scriptscriptstyle\! \rho}(\{x,y,z\}, \emptyset)=\{x\}$. We want to show that all the three cases lead to contradiction.   
				
For case 1, we have $1-\mathring{\pi}_{\scriptscriptstyle\! \{x,y,z\}}=1-\Phi_{\rho}(\{x,y,z\}, \emptyset)=\sum_{w \in \{x,y,z\}}\rho(w|\{x,y,z\},\emptyset)$ $= \rho(x|\{x,z\},\emptyset)+\rho(y|\{x,y\},\emptyset) + \rho(z|\{y,z\},\emptyset) = 1-\Phi_{\rho}(\{x,z\},\emptyset)- \rho(z|\{x,z\},\emptyset) + 1-\Phi_{\rho}(\{x,y\},\emptyset)- \rho(x|\{x,y\},\emptyset) + 1-\Phi_{\rho}(\{y,z\},\emptyset)- \rho(y|\{y,z\},\emptyset)$ $= 3 - \Phi_{\rho}(\{x,z\},\emptyset)- \rho(z|\{z\},\emptyset)  -\Phi_{\rho}(\{x,y\},\emptyset)- \rho(x|\{x\},\emptyset)  -\Phi_{\rho}(\{y,z\},\emptyset)- \rho(y|\{y\},\emptyset)$ $= \pi_{x} + \pi_{y} + \pi_{z} - \pi_{\scriptscriptstyle\! \{x,y\}} - \pi_{\scriptscriptstyle\! \{y,z\}} - \pi_{\scriptscriptstyle\! \{x,z\}} < 1-\mathring{\pi}_{\scriptscriptstyle\! \{x,y,z\}},$ which is a contradiction.

For case 2, since $c_{\scriptscriptstyle\! \rho}(\{x,y,z\},\emptyset)=\{x,y\}$ and $z \succ x$, we have $c_{\scriptscriptstyle\! \rho}(\{y,z\},\emptyset)=\{y\}$. It follows that $y \trianglerighteq z$, and thus by Axiom A5 and $z \succ x$, we have for all $A \in \mathcal{M}$ with $x \in c_{\scriptscriptstyle\! \rho}(A,\emptyset)$, $\rho(y|A,\emptyset)=\rho(y|A\backslash x,\emptyset)$. By Axiom A2, $x \in c_{\scriptscriptstyle\! \rho}(\{x,y\},\emptyset)$, and thus   $\rho(y|\{x,y\},\emptyset)=\rho(y|\{y\},\emptyset)$, which contradicts to the fact that $x \succ y$.

For case 3, since $c_{\scriptscriptstyle\! \rho}(\{x,y,z\}, \emptyset)=\{x\}$, we have  $x \trianglerighteq z$. By $z \succ x$, we have $z \in c_{\scriptscriptstyle\! \rho}(\{x,z\},\emptyset)$ and $\rho(x|\{x,z\},\emptyset) \neq \rho(x|\{x\},\emptyset)$. By Axiom A5, this  contradicts to $x \trianglerighteq z$. 
\end{proof}

\noindent \textbf{Claim 3.} \textit{
For all $x \in X$ and $A \in \mathcal{M}$ with $x \in A$, if $x \notin c_{\scriptscriptstyle\! \rho}(A,\emptyset)$, then there exists $y \in \mathcal{N}^+_{\scriptscriptstyle\! A}(x)$ such that $y \succ x$. }

\begin{proof}
We prove by induction on $|A|$. First, if $|A|=2$, then $A=\{x,y\}$ and $x \notin c_{\scriptscriptstyle\!\rho}(\{x,y\},\emptyset)$. Then by the construction of $\mathcal{N}$ and $\succ$, we have $(x,y)\in \mathcal{N}$ and $y \succ x$. Therefore, the lemma holds when $|A|=2$.

Assume that the lemma holds whenever $|A|\le n$, where $n \ge 2$. Consider the case  where $|A|=n+1$. Since $x \notin c_{\scriptscriptstyle\!\rho}(A,\emptyset)$, by Axiom A3, there exists $y \in A$ such that $c_{\scriptscriptstyle\!\rho}(A,\emptyset)=c_{\scriptscriptstyle\!\rho}(A\backslash y, \emptyset)$. If $y \neq x$, then $x \notin c_{\scriptscriptstyle\!\rho}(A\backslash y, \emptyset)$, and we are done by our induction hypothesis. Hence, consider the case where $x$ is the only alternative in $A$ such that $c_{\scriptscriptstyle\!\rho}(A,\emptyset)=c_{\scriptscriptstyle\!\rho}(A\backslash x, \emptyset)$. By a similar argument, we can additionally assume that for all $z \in c_{\scriptscriptstyle\!\rho}(A, \emptyset)$, we have $x \in c_{\scriptscriptstyle\!\rho}(A\backslash z, \emptyset)$. Thus, by Claim 2, for all $z \in c_{\scriptscriptstyle\!\rho}(A,\emptyset)$, $z \succ x$.

To proceed, consider $(A\backslash x, \emptyset)$. We first show that if $c_{\scriptscriptstyle\!\rho}(A\backslash x, \emptyset)=A\backslash x$, then $|A\backslash x|=1$, and we are done. To see this, suppose to the contrary that  $|A\backslash x| \ge 2$, and let $y$ and $z$ be two distinct alternatives in $A\backslash x$. By Axiom A2 and the assumptions we impose on the case we consider, we have $x \in c_{\scriptscriptstyle\!\rho}(A\backslash y, \emptyset) = A\backslash y$ and $x \in c_{\scriptscriptstyle\!\rho}(A\backslash z, \emptyset)=A\backslash z$. Since for all $\hat{x} \in A\backslash x$, we have $\hat{x} \succ x$, by Claim 2, we have $x \vdash A\backslash \{x,y\}$ and $x \vdash A\backslash \{x,z\}$. It follows from Axiom A4 that $x \vdash A\backslash x$, which by Claim 1 is a contradiction since $x \notin c_{\scriptscriptstyle\!\rho}(A,\emptyset)$.

By the above argument, we can additionally assume that $c_{\scriptscriptstyle\!\rho}(A\backslash x,\emptyset) \neq A\backslash x$. By Axiom A3,  there exists $z \in A\backslash x$ such that $c_{\scriptscriptstyle\!\rho}(A\backslash \{x,z\}, \emptyset)=c_{\scriptscriptstyle\!\rho}(A\backslash x, \emptyset)=c_{\scriptscriptstyle\!\rho}(A,\emptyset)$. By Axiom A2, we have $c_{\scriptscriptstyle\!\rho}(A\backslash z)=\{x\} \cup c_{\scriptscriptstyle\!\rho}(A)$. Since $x \in c_{\scriptscriptstyle\!\rho}(A\backslash z, \emptyset)= \{x\} \cup c_{\scriptscriptstyle\!\rho}(A)$ and for all $w \in c_{\scriptscriptstyle\!\rho}(A)$, $w \succ x$, by Claim 2, we have $x\vdash A\backslash \{x,z\}$. Since $x \notin c_{\scriptscriptstyle\!\rho}(A,\emptyset)$, we have $(x,z) \in \mathcal{N}$. Since $z \notin c_{\scriptscriptstyle\!\rho}(A\backslash x, \emptyset)$, by the induction hypothesis, there exists $y \in A\backslash x$ such that $y \succ z$ and $y \in \mathcal{N}^+_{\scriptscriptstyle\!A\backslash x}(z)$. Repeatedly applying this argument, we can assume without loss of generality that $y \in c_{\scriptscriptstyle\!\rho}(A\backslash x, \emptyset)$.  It then follows that $y \in \mathcal{N}^+_{\scriptscriptstyle\!A}(x)$. Since $y \in c_{\scriptscriptstyle\!\rho}(A\backslash x, \emptyset)$, we have $y \in c_{\scriptscriptstyle\!\rho}(A, \emptyset)$, and thus $y \succ x$. 
\end{proof}

\noindent \textbf{Claim 4.} \textit{For all $x \in X$ and $A \in \mathcal{M}$, if $x \in c_{\scriptscriptstyle\!\rho}(A,\emptyset)$, then  $x=\max(\mathcal{N}^+_{\scriptscriptstyle\!A}(x); \succ)$.}

\begin{proof}
We show that in the menu $(A,\emptyset)$, if there is an alternative that is associated with $x$ and $\succ$-better than $x$, then $x$ is not chosen. By Axiom A2, it suffices to show that for any sequence of alternatives $(x_k)_{k=1}^n$, where $n \ge 2$, if for all $k \in \{1,...,n-1\}$, $x_n \succ x_{k}$ and $(x_k, x_{k+1}) \in \mathcal{N}$,  then $x_1 \notin c_{\scriptscriptstyle\!\rho}(\{x_1,...,x_n\},\emptyset)$. We show this by induction on $n$. First, let $n =2$. We have $x_1\mathcal{N} x_2$ and $x_2 \succ x_1$. Suppose to the contrary that $x_1 \in c_{\scriptscriptstyle\!\rho}(\{x_1, x_2\}, \emptyset)$, then by Claim 2 and the construction of $\succ$, we have $c_{\scriptscriptstyle\!\rho}(\{x_1, x_2\}, \emptyset)=\{x_1, x_2\}$ and $\rho(x_2|\{x_1, x_2\},\emptyset)=\rho(x_2|\{x_2\},\emptyset)$, i.e., $x_1 \vdash \{x_2\}$. However, since $x_1 \mathcal{N} x_2$, by the construction of $\mathcal{N}$, we can find $A \in \mathcal{M}$ such that $x_1 \vdash A$ and $x_1 \notin c_{\scriptscriptstyle\!\rho}(A\cup\{x_1, x_2\}, \emptyset)$. By Axiom A4, we have $x_1 \vdash Ax_2$, and by Claim 1, we have $x_1 \in c_{\scriptscriptstyle\!\rho}(A\cup \{x_1, x_2\}, \emptyset)$, which is a contradiction. Thus, we must have $x_1 \notin c_{\scriptscriptstyle\!\rho}(\{x_1, x_2\}, \emptyset)$.

Next, suppose that the induction hypothesis holds for all $n \le m$ ($m\ge 2$). Consider the case where $n=m+1$. Since for all $k \in \{1,...,n-1\}$, $x_{n} \succ x_k$, we have $c_{\scriptscriptstyle\!\rho}(\{x_2,...,x_n\},\emptyset)=\{x_n\}$ by our induction hypothesis. Suppose to the contrary that $x_1 \in c_{\scriptscriptstyle\!\rho}(\{x_1,...,x_n\}, \emptyset)$, by Axiom A2, we have $c_{\scriptscriptstyle\!\rho}(\{x_1,...,x_n\}, \emptyset)=\{x_1, x_n\}$. By Claim 2, we have $\rho(x_n|\{x_1,...,x_n\}, \emptyset)=\rho(x_n|\{x_2,...,x_n\}, \emptyset)$. Thus $x_1 \vdash \{x_2,...,x_n\}$. Since $x_1\mathcal{N}x_2$, we can find $A\in \mathcal{M}$ such that $x_1 \vdash A$ and $x_1 \notin c_{\scriptscriptstyle\!\rho}(A\cup \{x_1, x_2\},\emptyset)$. By Axiom A4, we have $x_1 \vdash A\cup \{x_2,...,x_n\}$, and by Claim 1, we have $x_1 \in c_{\scriptscriptstyle\!\rho}(A\cup\{x_1,...,x_n\}, \emptyset)$. It follows from Axiom A2 that $x_1 \in c_{\scriptscriptstyle\!\rho}(A\cup \{x_1, x_2\}, \emptyset)$, which is a contradiction. Therefore, we have $x_1 \notin c_{\scriptscriptstyle\!\rho}(\{x_1,...,x_n\}, \emptyset)$. \end{proof}

With Claims 3 and 4, we have for all $A \in \mathcal{M}$ and $x \in A$, $x \in c_{\scriptscriptstyle\!\rho}(A,\emptyset)$ if and only if $x=\max(\mathcal{N}^+_{\scriptscriptstyle\!A}(x); \succ)$. Let $c_{\scriptscriptstyle\!\rho}(A,\emptyset)=\{x_1,...,x_n\}$ such that for all $k\in \{1,...,n-1\}$, $x_k \succ x_{k+1}$. We can have a partition $\{B_k\}_{k=1}^n$ of $A$ such that  for every $k$, $B_k=\{y\in A:  \mathcal{N}^+_{\scriptscriptstyle\! A}(y) \cap \{x_1,...,x_k\}=\{x_k\}\}.$ Note that for each $k$, $x_k \in B_k$, and for all $y \in B_k \backslash x_k$, $y \notin c_{\scriptscriptstyle\! \rho}(A,\emptyset)$, $x_k \succ y$ and $x_k \in \mathcal{N}^+_{\scriptscriptstyle\! B_k}(y)$. To show that $\rho$ can be represented by $(\pi, \mathcal{N}, \succ)$ as a $\mathcal{E}^F$-ABC, it suffices to show that for all $k \in \{1,...,n\}$, \begin{equation}\label{eq:appendix}
				\sum_{t=1}^k\rho(x_t|A,\emptyset) = 1-\mathring{\pi}_{\scriptscriptstyle\!C_k},    
			\end{equation} where $C_k=\cup_{t=1}^k B_t.$ Note that equation (\ref{eq:appendix}) holds when $k=n$. Consider some $k < n$. Let $D_k = A\backslash C_k$. It follows that $\mathcal{N}^+_{\scriptscriptstyle\!A}(D_k) \cap C_k = \emptyset$. Thus, for all $D \subseteq D_k$,  $c_{\scriptscriptstyle\!\rho}(C_k \cup D, \emptyset) \cap D \neq \emptyset$. Therefore, we can enumerate $D_k=\{y_1,...,y_m\}$ such that for all $t \in \{1,...,m\}$, $y_t \in c_{\scriptscriptstyle\!\rho}(C_k \cup \{y_1,...,y_t\}, \emptyset)$. Note that for all $D \subseteq D_k$,  $\{x_1,...,x_k\} = c_{\scriptscriptstyle\!\rho}(C_k \cup D, \emptyset) \cap C_k$.  It  follows from Claim 2 that for all $t \in \{1,...,k\}$ and $s \in \{1,...,m\}$, $\rho(x_t|C_k, \emptyset) = \rho(x_t|C_k\cup \{y_1,...,y_s\}, \emptyset) = \rho(x_t|A, \emptyset)$. Therefore, $$\sum_{t=1}^k\rho(x_t|A,\emptyset) =  \sum_{t=1}^k\rho(x_t|C_k,\emptyset) = 1-\Phi_{\rho}(C_k, \emptyset) = 1-\mathring{\pi}_{\scriptscriptstyle\!C_k}.$$ The sufficiency is thus shown. 
		\end{proof}

\subsection*{OA-2. Imperfectly Recorded Set of Observable Alternatives}
In this section, we revisit the issue that unavailable but observable alternatives may be imperfectly recorded in the data. We provide formal analysis of two specific assumptions discussed in Section \ref{subsec:unobserve}.
\bigskip

\noindent\textbf{Independent observability.} 
When extra observable alternatives arise from exogenous sources, a natural specification is an \emph{independent observability} model. Let $\eta:X \rightarrow (0,1)$ be an observability probability function such that each alternative $x$ becomes observable with constant probability $\eta(x)$, independently across $x$. Fixing a recorded menu $(A,B)$, the actual menu faced by the DM takes the form $(A,C)$ where $C\in\mathcal{M}$ satisfies $B\subseteq C$, and the probability of $(A,C)$ is
$$
\eta(A,C|A,B)
\;:=\;
\left(\prod_{x\in C\backslash B}\eta(x)\right)
\left(\prod_{y\in X\backslash (AC)} (1-\eta(y))\right).
$$
Thus, if the DM's \emph{actual} choice behavior follows an ABC $\hat{\rho}$ represented by $(\pi,\mathcal{N},\succ)$, then the \emph{documented} choice rule $\rho$ satisfies, for all $(A,B)\in\mathcal{E}$ and $x\in A$,
$$
\rho(x|A,B)
\;=\;
\sum_{C\subseteq X\backslash A:\, B\subseteq C}
\eta(A,C|A,B)\,\hat{\rho}(x|A,C).
$$

As discussed in Section \ref{subsec:unobserve}, the DM's preference $\succ$ remains uniquely identified. For the associative network, consider any distinct $x,y\in X$. We have
$$
x\mathcal{N}y
\quad\Longleftrightarrow\quad
\Phi_{\rho}(\{y\},\{x\})=\Phi_{\rho}(\{x,y\},\emptyset).
$$
To see this, first suppose $x\mathcal{N}y$. Then whenever $x$ is incorporated into consideration, $y$ is also eventually considered, so the default option is blocked whenever $x$ is attended to, regardless of whether $x$ is available or merely observable. Hence, we have $\Phi_{\rho}(\{y\},\{x\})=\Phi_{\rho}(\{x,y\},\emptyset)$. Conversely, if $y \notin \mathcal{N}(x)$, then given menu $(\{y\},\{x\})$, there is an event with positive probability in which only $x$ and $y$ are observable and the DM attends only to $x$. In this case, the default option is chosen since $x$ is unavailable and does not trigger the attention to $y$. When the menu is $(\{x,y\},\emptyset)$,  the default option is blocked whenever $x$ is considered. Therefore,
$$
\Phi_{\rho}(\{y\},\{x\})>\Phi_{\rho}(\{x,y\},\emptyset),
$$
so equality can hold only if $x\mathcal{N}y$. This yields a unique  identification of $\mathcal{N}$.

Since $\pi$ and $\eta$ both affect initial attention, disentangling them can be difficult. In certain  cases where the associative network is sparse, one can partially identify $\pi$. For instance, if an alternative $x$ satisfies $\mathcal{N}(x)=\{x\}$, then for any menu $(A,\emptyset)$ with $x\notin A$,
$$
\frac{\Phi_{\rho}(Ax,\emptyset)}{\Phi_{\rho}(A,\emptyset)}
\;=\;
\mathring{\pi}(x).
$$
In another special case, if there is no $y\in X\backslash\{x\}$ such that $y\mathcal{N}x$, then
$$
\pi(x)=\rho(x|\{x\},\emptyset).
$$

When the associative network is dense, however, separating $\pi$ from $\eta$ becomes substantially harder. Consider the extreme case $\mathcal{N}=X\times X$. Then for any recorded menu $(A,B)$,
$$
\Phi_{\rho}(A,B)
=
\Phi_{\rho}(AB,\emptyset)
=
\left(\prod_{x\in AB}\mathring{\pi}(x)\right)
\left(\prod_{y\in X\backslash (AB)}\bigl(1-\eta(y)+\eta(y)\mathring{\pi}(y)\bigr)\right).
$$
In this case, for each $x\in X$ the data identify at most the ratio
\begin{equation}\label{eq:indepdnentobserve:ratio}
\frac{1-\eta(x)+\eta(x)\mathring{\pi}(x)}{\mathring{\pi}(x)}    
\end{equation}
together with the aggregate term $$\prod_{y\in X}\mathring{\pi}(y).$$

If we further assume that the observability probability is the same for all alternatives (i.e., $\eta(x)=\eta^*\in (0,1)$ for all $x\in X$), then $\eta^*$ and the function $\pi$ are uniquely identified when $\mathcal{N}=X\times X$.   To see this, observe that the ratio in equation (\ref{eq:indepdnentobserve:ratio}) implies that $\mathring{\pi}(x)$ can be expressed as a function of $\eta^*$ and is strictly decreasing in $\eta^*$. Hence, the aggregate term $\prod_{y\in X}\mathring{\pi}(y)$ pins down $\eta^*$ and therefore each $\mathring{\pi}(x)$.

The same logic applies when the alternatives can be partitioned into categories $\{X_k\}_{k=1}^n$ such that: (i) any two alternatives in the same category are mutually linked under $\mathcal{N}$; (ii) alternatives from different categories are not linked under $\mathcal{N}$; and (iii) $\eta$ is constant within each category. In this case, all model parameters can be uniquely identified. We leave for future work a systematic investigation of additional conditions under which sharper identification can be obtained in this model beyond the cases considered here. \bigskip

\noindent \textbf{Privately observable alternatives: identification.} There may be alternatives that are privately observable to the DM but not recorded by an outside analyst. Such extra observable alternatives can arise, for instance, from the DM's past experiences or memory. We capture this possibility by introducing a set $P\subseteq X$ of privately observable alternatives. Below, we formally define the generalized model.

\begin{definition}
A random choice rule $\rho$ is a \textbf{p}rivate \textbf{a}ssociation-\textbf{b}ased \textbf{c}onsideration rule (PABC) if there exists a tuple $(P, \pi, \mathcal{N}, \succ)$, where $P \in \mathcal{M}$, $\pi$ is an attention probability function, $\mathcal{N}$ is an associative network, and $\succ$ is a preference ordering, such that for all $(A,B) \in \mathcal{E}$ and $x \in A$,
$\rho(x|A,B)=\hat{\rho}(x|A, (BP)\backslash A),$ 
where $\hat{\rho}$ is the ABC represented by $(\pi, \mathcal{N}, \succ)$. The tuple $(P, \pi, \mathcal{N}, \succ)$ is said to represent $\rho$ as a PABC.
\end{definition}

The PABC captures the idea that the DM may have a set $P$ of privately observable alternatives which the analyst cannot observe.  Accordingly, when the documented menu is $(A,B)$, the actual menu faced by the DM is $(A, (BP)\backslash A)$. To study the model's identification properties, we introduce the following notation.

For a given choice rule $\rho$, define sets $P^{u}_{\rho}$ and $P^{l}_{\rho}$ such that
\begin{equation*}
    \begin{split}
        P^{u}_{\rho} &= \bigl\{x \in X : \forall y \in X\backslash x,\; \Phi_{\rho}(\{y\},\emptyset) = \Phi_{\rho}(\{y\},\{x\})\bigr\}, \\[2pt]
        P^{l}_{\rho} &= \bigl\{x \in P^{u}_{\rho} : \exists y \in X\backslash x,\; \Phi_{\rho}(\{y\},\{x\}) = \Phi_{\rho}(\{x,y\},\emptyset)\bigr\}.
    \end{split}
\end{equation*}
The set $P^{u}_{\rho}$ consists of all alternatives that do not affect the DM's choices when they are made unavailable but observable. Any alternative that is privately observable to the DM must belong to $P^{u}_{\rho}$, since such an alternative is always observable to the DM regardless of whether it is documented as observable or not. The set $P^{l}_{\rho}$ refines $P^{u}_{\rho}$ by collecting those alternatives that are \emph{revealed} to be privately observable. Specifically, an alternative $x$ belongs to $P^{l}_{\rho}$ if $x\in P^{u}_{\rho}$, and there exists some $y\in X$ such that attention to $x$ can trigger attention to $y$. For such an $x$, if it were \emph{not} privately observable, then making $x$ unavailable but observable in the documented menu would increase the likelihood that $y$ is considered, thereby raising the choice probability of $y$ and lowering that of the default option. This would contradict the fact that $x\in P^{u}_{\rho}$.

For any binary relation $\mathcal{R}$ on $X$ and $A\in \mathcal{M}$, we say that $\mathcal{R}$ is \emph{$A$-transitive} if for all $x,y,z \in X$ with $y \in A$, $x\mathcal{R}y$ and $y\mathcal{R}z$ imply $x\mathcal{R}z$. The \emph{$A$-transitive closure} $\mathcal{R}^*$ of $\mathcal{R}$ is defined such that $x\mathcal{R}^*y$ if and only if $x\mathcal{R}y$ or there is a sequence $(x_t)_{t=1}^n$ in $A$ such that $x\mathcal{R}x_1\mathcal{R}\cdots\mathcal{R}x_n\mathcal{R}y$.

\begin{proposition}\label{prop:oa:pabc:id}
Consider a PABC $\rho$ and a tuple $(P, \pi, \mathcal{N}, \succ)$ such that $P \in \mathcal{M}$, $\pi$ is an attention probability function, $\mathcal{N}$ is an associative network, and $\succ$ is a preference ordering. Let $\mathcal{W}$ be the $P^{l}_{\rho}$-transitive closure of $\mathcal{N}$ with its symmetric part denoted by $\mathcal{I}$. The tuple $(P, \pi, \mathcal{N}, \succ)$ represents $\rho$ as a PABC if and only if

\noindent \emph{(1)} $P^l_{\rho} \subseteq P \subseteq P^u_{\rho}$,

\noindent \emph{(2)}  for all distinct $x,y \in X$,  $x\mathcal{W}y$ if and only if $\Phi_{\rho}(\{y\}, \{x\})=\Phi_{\rho}(\{x,y\}, \emptyset)$,

\noindent \emph{(3)} for all distinct  $x, y \in X$, $x \succ y$ if and only if $\rho(x|\{x,y\},\emptyset)=\rho(x|\{x\},\{y\})$,

\noindent \emph{(4)} for all $x \in P^l_{\rho}$, $$\mathring{\pi}_{A}=\frac{\Phi_{\rho}(B,\emptyset)}{\Phi_{\rho}(C,\emptyset)},$$ where $A= \mathcal{I}(x) \cap P^l_{\rho}$, $B=\{y\in P^l_{\rho}: y \mathcal{W} x\}$, and $C=B\backslash A$,

\noindent \emph{(5)}  for all $x \in X\backslash P^l_{\rho}$, $$\mathring{\pi}_{x} = \frac{\Phi_{\rho}(D,\emptyset)}{\Phi_{\rho}(E,\emptyset)},$$ where $E=\{z\in P_{\rho}^l: z\mathcal{W}x\}$ and $D=Ex$. 
\end{proposition}

By Proposition \ref{prop:oa:pabc:id}, the exact set of privately observable alternatives $P$ is not precisely identified but is bounded between the revealed lower bound $P^l_\rho$ and the upper bound $P^u_\rho$. Notably, any such set works for the representation regardless of the other parameters. The associative network is identified  up to its $P^{u}_{\rho}$-transitive closure. As in the baseline ABC model, the preference ordering is uniquely identified through choices in binary menus. For attention probabilities, alternatives outside $P^l_\rho$ are point-identified, whereas those inside $P^l_\rho$ can only be partially identified: If two such alternatives are mutually associated, they are always considered jointly, so only their aggregate attention probability is recoverable. In summary, even when the data does not document all alternatives observable to the DM, key components of the decision process---preferences, the  associative structure, and attention probabilities---remain at least partially recoverable.

\begin{proof}[Proof of Proposition \ref{prop:oa:pabc:id}]
Consider a PABC $\rho$ represented by $(P, \pi, \mathcal{N}, \succ)$. It trivially holds that $P_{\rho}^l \subseteq P \subseteq P_{\rho}^u$. For every $x \in P_{\rho}^u\backslash P_{\rho}^l$, observe that for all $y \in X\backslash x$, $\Phi_{\rho}(\{y\}, \{x\}) \neq \Phi_{\rho}(\{x,y\}, \emptyset)$. It follows that $\mathcal{N}(x)=\{x\}$. Hence, for any $\hat{P}$ with $P_{\rho}^l \subseteq \hat{P} \subseteq P_{\rho}^u$, $(\hat{P}, \pi, \mathcal{N}, \succ)$ also represents $\rho$ as a PABC. In particular, the tuple $(P_{\rho}^l, \pi, \mathcal{N}, \succ)$ represents $\rho$ as a PABC. It can be then easily verified that the tuple $(P_{\rho}^l, \pi, \mathcal{N}, \succ)$ satisfies (1)-(5), and so does the tuple $(P, \pi, \mathcal{N}, \succ)$.

Conversely, consider a tuple $(P, \pi, \mathcal{N}, \succ)$ that satisfies (1)-(5). We want to show that it represents $\rho$ as a PABC. It suffices to show that $(P^l_{\rho}, \pi, \mathcal{N}, \succ)$ represents $\rho$ as a PABC.  Consider another tuple $(P_{\rho}^l, \hat{\pi}, \succ', \mathcal{N}')$ that represents $\rho$ as a PABC. By (2), $\mathcal{N}'$ and $\mathcal{N}$ have the same $P^l_{\rho}$-transitive closure. This ensures that in any menu $(A,B)$, if the DM starts with an initial consideration set $C$, the two associative networks will lead the DM's attention to the same final consideration set. By (3), $\succ =\succ'$. For each $x \in X$, if $x \in X\backslash P_{\rho}^l$, define $[x]=\{x\}$; otherwise, define $[x]=\{y \in P_{\rho}^l: x \mathcal{I} y\}$. 
By (4) and (5), $\pi$ and $\hat{\pi}$ agree on each $[x]$. Since alternatives in $[x]$ are always considered together or ignored together due to mental association, $\mathring{\pi}_{[x]}$ and $\mathring{\hat{\pi}}_{[x]}$ are sufficient statistics for the DM's distribution of final consideration sets. Therefore, $(P^l_{\rho}, \pi,  \mathcal{N}, \succ)$ also represents $\rho$ as a PABC.
\end{proof}

\medskip

\noindent \textbf{Privately observable alternatives: axioms.} We proceed to develop an axiomatic foundation for PABCs. Fix a choice rule $\rho$. We begin with the following definition.

\begin{definition}
A set $A\in \mathcal{M}$ is privately association-proof if for all $x \in X\backslash A$, $$\Phi_{\rho}(A,\emptyset) \neq \Phi_{\rho}(Ax,\emptyset).$$ Denote by $\mathcal{P}$ the collection of all privately association-proof sets. 
\end{definition}

With privately observable alternatives, it could happen that $\Phi_{\rho}(A,\emptyset)=\Phi_{\rho}(Ax,\emptyset)$ for some $x \in X\backslash A$, in which case $x$ is in the privately observable set and some alternative in $A$ is associated with $x$. If a set $A$ is privately association-proof, then every alternative in $A$ is not associated with any privately observable alternative. Hence, whether the DM chooses the default option in the menu $(A,\emptyset)$ depends solely on whether she pays initial attention to some alternative in $A$. This leads to the following axiom. 
\bigskip

\noindent \textbf{Axiom R1---(Private Independence):} For all $A,B,C,D \in \mathcal{P}$ with $B\subseteq A$, $D \subseteq C$ and $A\backslash B = C\backslash D$, $$\frac{\Phi_{\rho}(A,\emptyset)}{\Phi_{\rho}(B,\emptyset)} = \frac{\Phi_{\rho}(C,\emptyset)}{\Phi_{\rho}(D,\emptyset)}.$$

  \bigskip

The second axiom slightly modifies Axiom 2.
\bigskip

\noindent \textbf{Axiom R2---(Monotone Idempotence):} For all $(A,B) \in \mathcal{E}$ and $x \in A$,  $\Phi_{\rho}(A\backslash x,B) \ge \Phi_{\rho}(A\backslash x,Bx) \ge  \Phi_{\rho}(A,B)$ with at most one strict inequality. \bigskip 

The next axiom characterizes the key feature of private observability. 

\bigskip

\noindent \textbf{Axiom R3---(Revealed Private Observability):} For all $(A,B) \in \mathcal{E}$ and $x \in A$, if $\Phi_{\rho}(A,B) = \Phi_{\rho}(A\backslash x, B)$, then for all $(C,D) \in \mathcal{E}$,  $\Phi_{\rho}(C,D)=\Phi_{\rho}(C,D\backslash x).$ \bigskip

To understand Axiom R3, note that if $x$ is \emph{not} in the DM's privately observable set, then removing $x$ from both the available and observable sets must strictly increase the probability of choosing the default option. This is because the presence of $x$ attracts attention and deters the choice of the default option. Hence, the condition $\Phi_{\rho}(A,B) = \Phi_{\rho}(A\backslash x, B)$ reveals that $x$ must in fact be privately observable. In this case, its documented observability in any menu no longer affects choices as it is always observable to the DM.

Recall that we define $x \xrightarrow{B} A$ if either $x \in A$ or $\Phi_{\rho}(A,B) \neq \Phi_{\rho}(A,B\backslash x)$. This captures the idea that attention to $x$ can prompt attention to some alternative in $A$ via an association path in $B$. When $x$ is privately observable, however, making it  unavailable but observable in any menu does not affect the DM's choices. Consequently, we may have $x \in B$ and $\Phi_{\rho}(A,B) = \Phi_{\rho}(A,B\backslash x)$ even though some alternative in $A$ is associated with $x$ through $B$. To account for this possibility, we revise this definition as follows.

For any $x \in X$ and $(A,B) \in \mathcal{E}$, write $x \xrightarrow{B}_* A$ if either (1) $x \in A$ or (2) $x \in B$ and $\Phi_{\rho}(Ax,B\backslash x) =\Phi_{\rho}(A,B)$.  We note that $x \xrightarrow{B}_* A$ is equivalent to $x \xrightarrow{B} A$ when $x$ is not privately observable; the former notion, however, additionally captures the case where $x$ is privately observable yet its consideration leads to the consideration of some alternative in $A$ through $B$. The next two axioms parallel Axioms 3 and 4, with $\xrightarrow{B}_*$ replacing $\xrightarrow{B}$.

\bigskip

\noindent \textbf{Axiom R4---(Expansion$^*$):} For all $x \in X$ and $(A,B), (C,D) \in \mathcal{E}$, if $C\subseteq A$ and $CD \subseteq AB$, then $x \xrightarrow{D}_* C$ implies $x \xrightarrow{B}_* A$.  \bigskip

\noindent \textbf{Axiom R5---(Path Connectedness$^*$):} For all $x,y \in X$ and $(A,B) \in \mathcal{E}$ with $x \neq y$, if $x \xrightarrow{B}_* A$ and not $x \xrightarrow{B\backslash y}_* A\backslash y$, then $x \xrightarrow{B\backslash y}_* \{y\}$ and $y \xrightarrow{B\backslash x}_* A$. \bigskip

We have the following characterization result for PABCs.

\begin{theorem}\label{thm:oa:pabc:axiom}
A choice rule $\rho$ is a PABC if and only if it satisfies Axioms R1-R5 and Axiom 5.    
\end{theorem}

\begin{proof}[Proof of Theorem \ref{thm:oa:pabc:axiom}]
The necessity of the axioms can be shown similarly as the proof of Theorem \ref{thm:main} and is therefore omitted. We prove only sufficiency. Consider a choice rule $\rho$, and throughout the proof, we assume that Axioms R1–R5 and Axiom 5 hold for $\rho$.

Define $P$ such that $x \in P$ if there is $(A,B) \in \mathcal{E}$ with $x \in A$ and $\Phi_{\rho}(A,B)=\Phi_{\rho}(A\backslash x, B)$. Define $\mathcal{N}$ such that 
for all distinct $x,y \in X$, $x\mathcal{N}y$ if $\Phi_{\rho}(\{y\}, \{x\})=\Phi_{\rho}(\{x, y\}, \emptyset)$. Let $\mathcal{I}$ be the symmetric part of $\mathcal{N}$. Define $\succ$ such that for all distinct $x,y \in X$, $x\succ y$ if there is $A \in \mathcal{M}$ such that $x,y \in A$ and $$\rho(y|A,A^c) \neq \rho(y|A\backslash x, (A\backslash x)^c).$$ Define $\pi$ such that (1) for all $x \in X\backslash P$, $$\pi_x=1-\frac{\Phi_{\rho}(Ax,\emptyset)}{\Phi_{\rho}(A,\emptyset)},$$ where $A=\{y \in P: y \mathcal{N}x\}$, and (2) for all $x \in P$, $$\pi_x=1-\left(\frac{\Phi_{\rho}(B,\emptyset)}{\Phi_{\rho}(C,\emptyset)}\right)^{\frac{1}{|B\backslash C|}},$$ where $B=\{y \in P: y \mathcal{N}x\}$ and $C=B\backslash \{y \in P: x \mathcal{I} y\}$. We break down the proof into the following sequence of claims. \medskip

\noindent \textbf{Claim 1.}  For all  $(A,B) \in \mathcal{E}$ and $x \in B$,   $x \xrightarrow{B}_* A$ implies $\mathcal{N}_{\!\scriptscriptstyle AB}^+(x) \cap A \neq \emptyset$.

\begin{proof} 
Consider  $(A,B) \in \mathcal{E}$  and $x \in B$ such that $x  \xrightarrow{B}_* A$. It follows that $A \neq \emptyset$, since otherwise we have $x \notin A$ and $\Phi_{\rho}(A,B) = \Phi_{\rho}(\emptyset,B)  = \Phi_{\rho}(\{x\},B\backslash x) = 1$, which is a contradiction. By repeated applications of Axiom R5, there is $y \in A$ such that $x \xrightarrow{B}_* \{y\}$. It then suffices to show $y \in \mathcal{N}_{\!\scriptscriptstyle By}^+(x)$.

First, if $|B|=1$, then $B=\{x\}$. In this case, $x  \xrightarrow{B}_* \{y\}$ means $x  \xrightarrow{\{x\}}_* \{y\}$, which 
implies $\Phi_{\rho}(\{y\},\{x\}) = \Phi_{\rho}(\{x, y\},\emptyset)$. It then follows from the definition of $\mathcal{N}$ that $x\mathcal{N}y$. Thus, we have $y \in  \mathcal{N}_{\!\scriptscriptstyle \{x,y\}}^+(x) =\mathcal{N}_{\!\scriptscriptstyle By}^+(x)$.

Next, assume by induction that when $|B| \le n$, $x  \xrightarrow{B}_{*} \{y\}$ implies $y \in \mathcal{N}_{\!\scriptscriptstyle By}^+(x)$. We show that it remains true when $|B| = n+1$. To see this, note that if there is $z \in B\backslash x$ such that $x  \xrightarrow{B\backslash z}_* \{y\}$, then by the induction hypothesis, $y \in \mathcal{N}_{\!\scriptscriptstyle (By)\backslash z}^+(x) \subseteq \mathcal{N}_{\!\scriptscriptstyle By}^+(x)$. Otherwise, for every $z \in B\backslash x$, not $x  \xrightarrow{B\backslash z}_* \{y\}$. By Axiom R5, for every $z \in B\backslash x$, $x \xrightarrow{B\backslash z}_* \{z\}$ and $z  \xrightarrow{B\backslash x}_* \{y\}$. By the induction hypothesis, for every $z \in B\backslash x$, we have $z \in \mathcal{N}_{\!\scriptscriptstyle B}^+(x) \subseteq \mathcal{N}_{\!\scriptscriptstyle By}^+(x)$ and $y \in  \mathcal{N}_{\!\scriptscriptstyle (By)\backslash x}^+(z) \subseteq \mathcal{N}_{\!\scriptscriptstyle By}^+(z)$. By the transitivity of $\mathcal{N}_{\!\scriptscriptstyle By}^+$, we have  $y \in \mathcal{N}_{\!\scriptscriptstyle By}^+(x)$. \end{proof}

\noindent \textbf{Claim 2.}  For all $(A,B) \in \mathcal{E}$ and $x \in B$,   $\mathcal{N}(x) \cap A \neq \emptyset$ implies $x \xrightarrow{B}_* A$.

\begin{proof} 
Since $\mathcal{N}(x) \cap A \neq \emptyset$, there exists $y \in A$ such that $x\mathcal{N}y$. It then follows by definition of $\mathcal{N}$ that $x \xrightarrow{\{x\}}_* \{y\}$. By Axiom R4, we have $x \xrightarrow{B}_* A$. \end{proof}

\noindent \textbf{Claim 3.} For all $(A,B) \in \mathcal{E}$ and $x \in B$, $$\Phi_{\rho}(A,B)=\Phi_{\rho}(Ax,B\backslash x) \quad \Longleftrightarrow  \quad \mathcal{N}_{AB}^+(x) \cap A \neq \emptyset.$$

\begin{proof}
Note that  $\Phi_{\rho}(A,B)=\Phi_{\rho}(Ax,B\backslash x)$ is equivalent to $x \xrightarrow{B}_* A$. By Claim 1, we have $\mathcal{N}_{AB}^+(x) \cap A \neq \emptyset$. Conversely, if $\mathcal{N}_{AB}^+(x) \cap A \neq \emptyset$, then we can find a sequence $(x_k)_{k=1}^n$ in $B$ such that $x_1 = x$, $\mathcal{N}(x_n) \cap A \neq \emptyset$, and for all $k \in \{1,...,n-1\}$, $x_k \mathcal{N} x_{k+1}$. It follows from Claim 2 that $\Phi_{\rho}(Ax_n, B\backslash x_n)=\Phi_{\rho}(A,B)$. By repeated applications of Claim 2, we have $\Phi_{\rho}(A\cup \{x_k\}_{k=1}^n, B\backslash \{x_k\}_{k=1}^n) = \Phi_{\rho}(A, B)$. By Axiom R2 and $x=x_1$, we have $\Phi_{\rho}(A\cup \{x_k\}_{k=1}^n, B\backslash \{x_k\}_{k=1}^n) \le \Phi_{\rho}(Ax,B\backslash x) \le \Phi_{\rho}(A, B)$, which implies $\Phi_{\rho}(Ax,B\backslash x) =  \Phi_{\rho}(A, B)$. 
\end{proof}

\noindent \textbf{Claim 4.} For all $(A,B) \in \mathcal{E}$, $\Phi_{\rho}(A,B)=\Phi_{\rho}(H_{\mathcal{N}}(A, (BP)\backslash A), \emptyset)$.  

\begin{proof}
Recall that $H_{\mathcal{N}}(A, (BP)\backslash A)$ is defined to be the set $\{x \in ABP: \mathcal{N}^+_{ABP}(x) \cap A \neq \emptyset\}$. Let $\hat{H}=H_{\mathcal{N}}(A, (BP)\backslash A)$ for simplicity. By the definition of $P$ and Axiom R3, we have $\Phi_{\rho}(A,B)=\Phi_{\rho}(A,(BP)\backslash A)$. It then follows from Claim 3 that $\Phi_{\rho}(A,B)=\Phi_{\rho}(\hat{H},  (BP)\backslash \hat{H})$. Since for all $x \in (BP)\backslash \hat{H}$ and $y \in \hat{H}$, $y \notin \mathcal{N}(x)$, by Claim 3 and Axiom R2, we have $\Phi_{\rho}(A,B)=\Phi_{\rho}(\hat{H}, \emptyset)$.    
\end{proof}

\noindent \textbf{Claim 5.} $\mathcal{N}$ is $P$-transitive.   
\begin{proof}
Consider three pairwise distinct  $x,y,z \in X$ with $y \in P$, $x \mathcal{N} y$ and $y \mathcal{N} z$. We have $\Phi_{\rho} (\{z\}, \{x\}) = \Phi_{\rho} (\{z\}, \{x,y\}) = \Phi_{\rho} (\{z,y\}, \{x\}) = \Phi_{\rho} (\{x,y,z\}, \emptyset) = \Phi_{\rho} (\{x,z\}, \{y\}) = \Phi_{\rho} (\{x,z\}, \emptyset)$, where the first and last equalities are by Axiom R3, and the rest follow from Claim 3. Hence, we have $x \mathcal{N} z$.    
\end{proof}

\noindent \textbf{Claim 6.} For all $A \in \mathcal{M}$, $A$ is privately association-proof if and only if for all $x \in P\backslash A$, there is no $y \in A$ such that $x\mathcal{N}y$.  

\begin{proof}
First assume that $A$ is privately association-proof. If for some $x \in P\backslash A$ and $y \in A$, we have $x \mathcal{N} y$, then by  Axioms R3, Claim 3, and the definition of $P$,  $\Phi_{\rho}(A,\emptyset)= \Phi_{\rho}(A,\{x\}) = \Phi_{\rho}(Ax, \emptyset)$, which is a contradiction. Conversely, assume that for all $x \in P\backslash A$ and $y \in A$, not $x\mathcal{N}y$. If for some $x \in X\backslash A$, $\Phi_{\rho}(A,\emptyset) = \Phi_{\rho}(Ax, \emptyset)$, then $x \in P$ by definition. By Axiom R3, we have $\Phi_{\rho}(A,\emptyset) = \Phi_{\rho}(A, \{x\}) = \Phi_{\rho}(Ax, \emptyset)$. By Claim 3, we have $x \mathcal{N}y$  for some $y\in A$, which is a contradiction. 
\end{proof}

\noindent \textbf{Claim 7.} For all $x \in X$, $\pi_x \in (0,1)$.

\begin{proof}
If $x \in X\backslash P$, then $\mathring{\pi}_x= \Phi_{\rho}(Ax, \emptyset)/\Phi_{\rho}(A, \emptyset)$, where $A = \{y \in P: y \mathcal{N}x\}$. Since $x \notin P$, $\Phi_{\rho}(Ax, \emptyset) \neq \Phi_{\rho}(A, \emptyset)$. By Axiom R2, we have $\Phi_{\rho}(Ax, \emptyset) < \Phi_{\rho}(A, \emptyset)$, and thus $\mathring{\pi}_x \in (0,1)$, implying $\pi_x \in (0,1)$. If $x \in P$, then $\left( \mathring{\pi}_x \right)^{k} = \Phi_{\rho}(B, \emptyset)/\Phi_{\rho}(B\backslash C, \emptyset)$, where $B=\{y \in P: y \mathcal{N} x\}$, $C=\{y \in P: x \mathcal{I} y\}$, and $k=|C|$. If $\Phi_{\rho}(B,\emptyset)=\Phi_{\rho}(B\backslash C, \emptyset)$, then Axiom R2 implies $\Phi_{\rho}(B,\emptyset) = \Phi_{\rho}(B\backslash C, C)$. It follows from the $P$-transitivity of $\mathcal{N}$ and Claim 3 that for each $z \in C$, there is $z' \in B\backslash C$ such that $z \mathcal{N} z'$, which further implies $x \mathcal{N} z'$---a contradiction. By Axiom R2,  $\Phi_{\rho}(B,\emptyset) < \Phi_{\rho}(B\backslash C, \emptyset)$, and thus $\pi_x \in (0,1)$. 
\end{proof}

\noindent \textbf{Claim 8.} For all $A \in \mathcal{M}$ that is privately association-proof, $\Phi_{\rho}(A, \emptyset) = \mathring{\pi}_{A}$.

\begin{proof}
For any $x \in A\cap (X\backslash P)$,  by Claim 6, $A\backslash x$ is also privately association-proof. By the $P$-transitivity of $\mathcal{N}$, the two sets $B=\{y \in P: y \mathcal{N}x\}$ and $Bx$ are both privately association-proof. By Axiom R1, $\Phi_{\rho}(A, \emptyset)/\Phi_{\rho}(A\backslash x, \emptyset) = \Phi_{\rho}(Bx, \emptyset)/\Phi_{\rho}(B, \emptyset) = \mathring{\pi}_x$.

For any $x \in P \cap A$, define $[x]=\{y \in P: x \mathcal{I} y\}$ and $[x]^+=\{y \in P: y \mathcal{N} x\}$. Note that $[x] \subseteq [x]^+ \subseteq A$ since $A$ is privately association-proof. If $P \cap A \neq \emptyset$, then we can find $x \in P \cap A$ such that for all $y \in P \cap A$, $x \mathcal{N} y$ implies $y \in [x]$. Therefore, for such an alternative $x$,   $A\backslash [x]$, $[x]^+$, and $[x]^+\backslash [x]$ are all privately association-proof. We have $\Phi_{\rho}(A, \emptyset)/\Phi_{\rho}(A\backslash [x], \emptyset) = \Phi_{\rho}([x]^+, \emptyset)/\Phi_{\rho}([x]^+\backslash [x], \emptyset) = \left(\mathring{\pi}_x\right)^{|[x]|} = \mathring{\pi}_{|[x]|}.$ Hence, inductively, we conclude that $\Phi_{\rho}(A,\emptyset)=\mathring{\pi}_A$. 
\end{proof}

By Claim 6, for each $(A,B) \in \mathcal{E}$, $H_{\mathcal{N}}(A,(BP)\backslash A)$ is privately association-proof. Claim 4 and Claim 8 imply $\Phi_{\rho}(A,B) = \mathring{\pi}_{H_{\mathcal{N}}(A,(BP)\backslash A)}$. By Axiom 5 and Lemma \ref{lm:general:preference}, we can show that $\succ$ is a preference ordering and satisfies that for all $x,y \in X$,  $x \succ y$ implies that for all $(A,B) \in \mathcal{E}$ with $y \in A$, $\rho(x|A,B)=\rho(x|A\backslash y, By)$. Following a similar proof procedure as the one for Theorem \ref{thm:main}, we can obtain the desired choice probability $\rho(x|A,B)$ for each $x \in A$.  
\end{proof}
\bigskip

\noindent \textbf{Privately observable alternatives: limited data.} We end this section by discussing identification of PABCs when the choice domain is restricted to
\[
\mathcal{E}^F=\{(A,\emptyset): A\in\mathcal{M}\}
\]
as in Section OA-1. That is, in each menu, the set of observable alternatives coincides with the set of available alternatives. Formally, we define $\mathcal{E}^F$-PABCs as follows.

\begin{definition}
A $\mathcal{E}^{F}$-choice rule $\rho$ is $\mathcal{E}^{F}$-PABC if there exists tuple $(P, \pi,  \mathcal{N}, \succ)$ such that for all $A \in \mathcal{M}$ and all $x \in A$, $$\rho(x|A,\emptyset)=\hat{\rho}(x|A, P\backslash A),$$ where $\hat{\rho}$ is the ABC represented by $(\pi,  \mathcal{N}, \succ)$. The $\mathcal{E}^{F}$-choice rule $\rho$ is  said to be represented by $(P, \pi, \mathcal{N}, \succ)$ as a $\mathcal{E}^{F}$-PABC.     
\end{definition}

First, following the argument in Section OA-1, we can uniquely identify the preference ordering $\succ$ via 
\[
x\succ y \quad \Longleftrightarrow \quad \rho(y|\{x,y\},\emptyset) < \rho(y|\{y\},\emptyset).
\]


For the remaining parameters, we focus on identifying minimal components that suffice for the representation. For the privately observable set, note that if there exist $A \in \mathcal{M}$ and $x\in A$ such that \begin{equation}\label{eq:unobserve:partialdata:id:private}
    \Phi_{\rho}(A,\emptyset)=\Phi_{\rho}(A\backslash x,\emptyset),
\end{equation} then $x$ must belong to the DM's privately observable set. The intuition is that, absent private observability, removing $x$ from the menu would also remove it from what is observable, in which case $x$ could no longer deter the choice of the default option. Let $P_{\rho}$ denote the set of all alternatives $x$ for which there exists some $A \in \mathcal{M}$ with $x\in A$ satisfying equation (\ref{eq:unobserve:partialdata:id:private}). Clearly, $P_{\rho}$ is contained in the DM's true privately observable set. That is, if $\rho$ is represented by $(P, \pi, \mathcal{N}, \succ)$ as an $\mathcal{E}^F$-PABC, then $P_{\rho}\subseteq P$. Moreover, the representation remains valid if we replace $P$ by $P_{\rho}$, as stated by the following proposition.

\begin{proposition}\label{prop:oa:pabc:EF:id:P}
If a $\mathcal{E}^{F}$-choice rule $\rho$ is represented by $(P, \pi, \mathcal{N}, \succ)$ as a $\mathcal{E}^{F}$-PABC, then $P_{\rho} \subseteq P$, and $\rho$ is also represented by $(P_{\rho}, \pi, \mathcal{N}, \succ)$ as a $\mathcal{E}^{F}$-PABC.
\end{proposition}

\begin{proof}[Proof of Proposition \ref{prop:oa:pabc:EF:id:P}]
It remains to show that $\rho$ is  represented by $(P_{\rho}, \pi, \mathcal{N}, \succ)$ as a $\mathcal{E}^{F}$-PABC. For any $x \in P \backslash P_{\rho}$, by the definition of $P_{\rho}$, we have for all $A\in \mathcal{M}$ with $x \in A$, $\Phi_{\rho}(A,\emptyset) \neq \Phi_{\rho}(A\backslash x,\emptyset)$. This implies  $\mathcal{N}(x)=\{x\}$. For such an alternative $x$, removing it from $P$ does not affect the DM's association process, and thus the representation remains valid. 
\end{proof}

Next, we turn to identification of the attention probability function $\pi$. Fix a $\mathcal{E}^F$-choice rule  $\rho$, and assume that it is represented by $(P, \pi, \mathcal{N}, \succ)$ as a $\mathcal{E}^{F}$-PABC. Define an associative network $\mathcal{U}[\rho]$ such that $$ \mathcal{U}[\rho]= \mathcal{X} \cup \{(x,y)\in X\times X: x \neq y \text{~and~} \Phi_{\rho}(\{x,y\}, \emptyset) = \Phi_{\rho}(\{y\}, \emptyset)\}.$$ When there is no confusion about $\rho$, we write $\mathcal{U}$ instead of $\mathcal{U}[\rho]$. By the definition of $P_{\rho}$, for any distinct $x,y \in X$,  $x\mathcal{U}y$ if and only if there is a sequence of alternatives $(x_k)_{k=1}^n$ in $P_{\rho}$ such that $x_1=x$ and $x_1 \mathcal{N} \cdots \mathcal{N} x_n \mathcal{N} y$. As a result, the network $\mathcal{U}$ is $P_{\rho}$-transitive.

Using $\mathcal{U}$, we can partially identify $\pi$. For each $x \in X$, define $$[x]^+=\{y \in P_{\rho}: y \mathcal{U}x\} \cup \{x\}.$$ We have $\Phi_{\rho}(\{x\}, \emptyset) = \Phi_{\rho}([x]^+, \emptyset) = \mathring{\pi}_{[x]^+}$. If $x \in X\backslash P_{\rho}$, then we have $\Phi_{\rho}([x]^+\backslash x, \emptyset) = \mathring{\pi}_{[x]^+\backslash x}$. It follows that \begin{equation}\label{eq:unique:PABC:limitdata:1}
\mathring{\pi}_x= \frac{\Phi_{\rho}([x]^+, \emptyset)}{\Phi_{\rho}([x]^+\backslash x, \emptyset)}.    
\end{equation}
If $x \in P_{\rho}$, define $[x]=\{y \in [x]^+: x \mathcal{U} y\}$. Since $\Phi_{\rho}(\{[x]^{+}\backslash [x]\}, \emptyset) = \mathring{\pi}_{[x]^{+}\backslash [x]}$, we have  
\begin{equation}\label{eq:unique:PABC:limitdata:2}
\mathring{\pi}_{[x]}=\frac{\Phi_{\rho}([x]^+, \emptyset)}{\Phi_{\rho}([x]^+\backslash [x], \emptyset)}.  
\end{equation}
Note that equations (\ref{eq:unique:PABC:limitdata:1}) and (\ref{eq:unique:PABC:limitdata:2}) (for $x\in X\backslash P_{\rho}$ and $x\in P_{\rho}$, respectively) characterize the necessary and sufficient conditions that $\pi$ must satisfy. In particular, the sufficiency stems from the fact that when $x\in P_{\rho}$, any two alternatives in $[x]$ must be either both considered or both ignored. Consequently, there is no point identification of $\pi_y$ for each $y\in [x]$, and the aggregate term $\mathring{\pi}_{[x]}$ suffices to rationalize the DM's choices. We summarize this discussion in the proposition below.

\begin{proposition}\label{prop:oa:pabc:EF:id:pi}
Suppose that $(P, \pi, \mathcal{N}, \succ)$ represents a $\mathcal{E}^{F}$-choice rule $\rho$ as a $\mathcal{E}^{F}$-PABC, then $\pi$ satisfies equation (\ref{eq:unique:PABC:limitdata:1}) for all $x \in X\backslash P_{\rho}$ and equation (\ref{eq:unique:PABC:limitdata:2}) for all $x \in P_{\rho}$. Furthermore, for any $\pi'$ that satisfies equation (\ref{eq:unique:PABC:limitdata:1}) for all $x \in X\backslash P_{\rho}$ and equation (\ref{eq:unique:PABC:limitdata:2}) for all $x \in P_{\rho}$, the tuple $(P, \pi', \mathcal{N}, \succ)$ also represents $\rho$ as a $\mathcal{E}^{F}$-PABC.
\end{proposition}

Finally, we discuss identification of the associative network. We note that it is without loss of generality to assume that the DM's associative network $\mathcal{N}$ is $P_{\rho}$-transitive. Since all alternatives in $P_{\rho}$ are privately observable, if $x \mathcal{N} y$ and $y \mathcal{N} z$ for some $y\in P_{\rho}$, then consideration of $x$ can always lead the DM to further consider $z$ via the privately observable alternative $y$. This is observationally equivalent to the case in which $z$ is directly associated with $x$. Accordingly, if $\mathcal{N}$ is not $P_{\rho}$-transitive, then we can simply replace it with its $P_{\rho}$-transitive closure, which leaves observed choice behavior unchanged.

Recall that we define the associative network $\mathcal{N}[\rho]$ in equation (\ref{eq:def:Nrho}). It turns out the combination of $\mathcal{N}[\rho]$ and  the $P_{\rho}$-transitive network $\mathcal{U}[\rho]$ is the minimal associative network that works for representing a $\mathcal{E}^F$-PABC. We have the following proposition.

\begin{proposition}\label{prop:oa:pabc:EF:id:network}
If a $\mathcal{E}^{F}$-choice rule $\rho$ is represented by $(P, \pi, \mathcal{N}, \succ)$ as a $\mathcal{E}^{F}$-PABC in which $\mathcal{N}$ is $P_{\rho}$-transitive, then $\mathcal{U}[\rho] \cup \mathcal{N}[\rho] \subseteq \mathcal{N}$, and $(P, \pi, \mathcal{U}[\rho] \cup \mathcal{N}[\rho], \succ)$ also represents $\rho$ as a $\mathcal{E}^{F}$-PABC. 
\end{proposition}

\begin{proof}[Proof of Proposition \ref{prop:oa:pabc:EF:id:network}]
Without loss of generality, we can let $P=P_{\rho}$. It can be shown similar to the proof of Proposition \ref{prop:oa:EF:id} that $\mathcal{U}[\rho] \cup \mathcal{N}[\rho] \subseteq \mathcal{N}$. We focus on showing that  $(P, \pi,  \mathcal{U}[\rho] \cup \mathcal{N}[\rho], \succ)$ also represents $\rho$ as a $\mathcal{E}^{F}$-PABC. It suffices to show that in any given menu $(A,\emptyset)$, any initial consideration set $B\subseteq A\cup P_{\rho}$ leads to the same final choice through either $\mathcal{N}$ or $\mathcal{U}[\rho] \cup \mathcal{N}[\rho]$.   Let $\hat{B}$ and $\bar{B}$ be the final consideration set of $B\subseteq A\cup P_{\rho}$ through the two networks $\mathcal{U}[\rho] \cup \mathcal{N}[\rho]$ and $\mathcal{N}$, respectively. Since $\hat{B} \subseteq \bar{B}$, it suffices to show that if $x$ is the $\succ$-optimal choice in $\bar{B} \cap A$, then $x \in \hat{B}$. Suppose to the contrary that $x \notin \hat{B}$. There exists a shortest path $(x_k)_{k=1}^{n+1}$ in $A\cup P_{\rho}$ such that $x_1 \in \hat{B}$, $x_{n+1}=x$, and for all $k\in \{1,...,n\}$, $x_k \mathcal{N} x_{k+1}$. Furthermore, we may assume that $x_{k+1}\notin \hat{B}$ for all $k\in\{1,...,n\}$. Otherwise, we can shorten the path by deleting any intermediate nodes that belong to $\hat{B}$, while keeping as $x_1$ the last alternative in $\hat{B}$ along the original path. Since for all $x \in P_{\rho}$, $x\mathcal{N}y$ implies $x\,\mathcal{U}[\rho]\,y$, we must have $x_1 \in A\backslash P_{\rho}$. Since $\mathcal{N}$ is $P_\rho$-transitive, without loss of generality, we can assume that for all $k \in \{2,...,n\}$, $x_k \in A\backslash P_{\rho}$.  Hence, for all $k \in \{1,...,n\}$, we have $x_{n+1} \succ x_k$. It follows from the proof of Proposition \ref{prop:oa:EF:id} that for all $k\in \{1,...,n\}$, $(x_k, x_{k+1}) \in \mathcal{N}[\rho]$, which contradicts the fact that $x_{n+1}=x\notin \hat{B}$.
\end{proof}

\subsection*{OA-3. Connection to Other Random Attention Models}
This section compares our model with other random attention models in the case where all observable alternatives are available. We first characterize our attention rule, that is, the distribution over the DM's final consideration sets induced by the attention and association procedure. We then show that this rule satisfies monotonicity \citepAppendix{OA-jpe2020random} but can violate attention overload \citepAppendix{OA-wp2021attention}.

A function $\sigma: \mathcal{M}\times \mathcal{M} \rightarrow [0,1]$ is an \textit{attention rule} if for all $S \in \mathcal{M}$, (i) $\sum_{S'\subseteq S}\sigma(S'|S)=1$,  (ii) for all $T \not\subseteq S$, $\sigma(T|S)=0$, and (iii) $\sigma(\{x\}|\{x\}) \in (0,1)$ for all $x\in X$. For each $S\in\mathcal{M}$, interpret $\sigma(T|S)$ as the probability that the DM's final consideration set is $T$ when $S$ is the set of observable alternatives. We assume that when $S=\{x\}$, the DM pays attention to $x$ with probability between $0$ and $1$.
Recall that in our ABC model, availability does not affect the formation of the final consideration set.

Consider an ABC $\rho$ represented by $(\pi,\mathcal{N},\succ)$. Then the tuple $(\pi,\mathcal{N})$ induces an attention rule $\sigma$. Specifically, for each $S\in\mathcal{M}$ and each $T\subseteq S$, define \begin{equation}\label{eq:oa:attentionrule:def}
    \sigma(T|S)= \sum_{C\subseteq S:\:T=\mathcal{N}_{\!\scriptscriptstyle S}^+(C)} \pi_{\scriptscriptstyle\! C} \mathring{\pi}_{\scriptscriptstyle\! S\backslash C}.
\end{equation} That is, the probability that $T$ is ultimately considered in $S$ equals the total probability of all initial consideration sets $C$ that, through the associative network $\mathcal{N}$, generate the final consideration set $T$. Denote by $\sigma_{\pi, \mathcal{N}}$ the attention rule induced by $(\pi, \mathcal{N})$.  Below, we characterize all such attention rules.  \bigskip

\noindent \textbf{Property 1.} For all $S,\hat{S} \in \mathcal{M}$, $T \subseteq S$, $\hat{T} \subseteq \hat{S}$, and $x \in (S\backslash T) \cap (\hat{S}\backslash \hat{T})$,  $$\sigma(T|S)> 0 \text{~and~} \sigma(\hat{T}|\hat{S})> 0 \quad \Longrightarrow \quad  \frac{\sigma(T|S\backslash x)}{\sigma(T|S)} = \frac{\sigma(\hat{T}|\hat{S}\backslash x)}{\sigma(\hat{T}|\hat{S})}.$$


For the DM to finally consider a subset $T$ when the observable set of alternatives is $S$, she must ignore all alternatives in $S\backslash T$. Thus, the alternatives in $S\backslash T$ affect the probability that the DM ultimately considers $T$ through their ability to attract the DM's attention. Property 1 states that, for a given alternative $x$, this effect---i.e., $x$'s ability to attract attention---is constant and does not depend on the set of observable alternatives. \bigskip

\noindent \textbf{Property 2.} For all mutually disjoint $S_1,S_2,T_1,T_2 \in \mathcal{M}$
  $$\sigma(T|T\cup S)\neq 0 \quad \Longleftrightarrow  \quad \text{for all~}  i,j \in \{1,2\},\, \sigma(T_i|T_i \cup S_j) \neq 0,$$ where $S=S_1\cup S_2$ and $T=T_1 \cup T_2$. \bigskip

Property 2 captures the network structure of the DM's attention. Note that $\sigma(T|T\cup S) \neq 0$ if and only if no alternative in $T$ can direct the DM's attention to any alternative in $S$. This holds if and only if neither the alternatives in $T_1$ nor those in $T_2$ direct the DM's attention to alternatives in $S_1$ or $S_2$.

\begin{proposition}\label{prop:oa:attention:characterization}
An attention rule $\sigma$ can be induced by some tuple $(\pi, \mathcal{N})$ if and only if it satisfies Properties 1 and 2.  
\end{proposition}
\begin{proof}[Proof of Proposition \ref{prop:oa:attention:characterization}]
The necessity part is trivial and omitted. To prove the sufficiency, define $\pi: X \rightarrow (0,1)$ such that for all $x \in X$, $\pi_x=\sigma(\{x\}|\{x\})$. Define $\mathcal{N}$ such that $x \mathcal{N} y$ if $x = y$ or $x \neq y$ and $\sigma(\{x\}|\{x,y\})=0$. By Property 2, a simple induction shows that for all $S\in \mathcal{M}$ and $T \subseteq S$, $\sigma(T|S) \neq 0$ if and only if for all $x \in T$ and $y \in S\backslash T$,  $(x,y) \notin \mathcal{N}$.

To proceed, we show by induction that $\sigma$ satisfies equation (\ref{eq:oa:attentionrule:def}) for all $S \in \mathcal{M}$ and $T \subseteq S$. Suppose that for all $S \in \mathcal{M}$ with $|S| \le k$ and for all $T \subseteq S$, equation (\ref{eq:oa:attentionrule:def}) holds for $\sigma(T|S)$. Consider some $S\in \mathcal{M}$ with $|S|=k+1$ and $T \subseteq S$. If $\sigma(T|S)=0$, then we are done. If $\sigma(T|S) > 0$, then there are two cases to consider. In case 1, $T \neq S$. It follows that for all $y \in S\backslash T$ and $x \in T$, $(x,y) \notin \mathcal{N}$. Pick any $y \in S\backslash T$. Since $\sigma(\emptyset|\{y\})>0$, Property 1 implies that $\frac{\sigma(T|S)}{\sigma(T|S\backslash y)} = \frac{\sigma(\emptyset|\{y\})}{\sigma(\emptyset|\emptyset)} = \sigma(\emptyset|\{y\}) = \mathring{\pi}_y$. By the induction hypothesis, $\sigma(T|S\backslash y)$ satisfies equation (\ref{eq:oa:attentionrule:def}), and thus $\sigma(T|S)$ satisfies it as well. In case 2, we have $T=S$. It follows that $\sigma(S|S)=1-\sum\limits_{\hat{T} \subsetneq S} \sigma(\hat{T}|S)$. Since for all $\hat{T} \subsetneq S$, $\sigma(\hat{T}|S)$ satisfies equation (\ref{eq:oa:attentionrule:def}), we conclude that $\sigma(S|S)$ also satisfies equation (\ref{eq:oa:attentionrule:def}). 
\end{proof}

A key feature of our attention rule $\sigma_{\pi, \mathcal{N}}$ is that it permits the possibility that $\sigma_{\pi, \mathcal{N}}(T|S)=0$ for some $S \in \mathcal{M}$ and $\emptyset \neq T \subseteq S$. This occurs when  some alternative in $S\backslash T$ is associated with some alternative in $T$, making it impossible for the DM to stop expanding her consideration set at $T$. This feature is important because it generates \emph{testable zero restrictions} on which final consideration sets can arise, and thus sharply distinguishes our attention rule from others in which any subset of $S$ may, in principle, be attainable as a final consideration set. In particular, these zeros reflect an endogenous ``attention spillover'' mechanism driven by associations, rather than by exogenous feasibility constraints.

To see how our attention rule differs from others, consider the case where, for mutually distinct $x,y,z\in X$, we have
$$
\mathcal{N}\cap\left(\{x,y,z\}\times\{x,y,z\} \right)=\{(x,x),(y,y),(z,z),(x,y)\}.
$$
That is, the only effective associative link among the three alternatives is from $x$ to $y$. It follows that for all $S\subseteq\{x,y,z\}$ with $x\in S$, $\sigma_{\pi,\mathcal{N}}(\{x\}|S)=0$ if and only if $y\in S$, while for $y$ and $z$, they can always be considered solely with positive probability in any $S\subseteq\{x,y,z\}$ that includes them. What is hard to replicate without an association mechanism is precisely the implied \emph{menu-dependent} and \emph{directional} support restriction: The mere presence of $y$ makes the singleton $\{x\}$ impossible, yet the presence of $x$ does not make $\{y\}$ impossible (and $z$ never plays such a role). In attention models driven only by alternative-specific salience or fixed weights, adding $y$ to a menu typically changes the {attention probabilities} assigned to subsets of alternatives but does not naturally force $\sigma(\{x\}|S)$ to drop \emph{exactly} to zero.  Achieving this would require  an ad hoc, pair-specific exclusion rule that depends on which other options are present. Our associative network provides an intuitive and transparent microfoundation for these zeros: A directed link $x\mathcal{N}y$ endogenously eliminates $\{x\}$ from the support of the attention rule exactly in those menus  containing $y$.


In what follows, we examine whether our attention rule satisfies general properties: monotone attention and attention overload.
Following \citeAppendix{OA-jpe2020random}, we say an attention rule $\sigma$ satisfies monotone attention if for all $S \in \mathcal{M}$, $T \subseteq S$, and  $x\in S\backslash T$, $$\sigma(T|S)\leq \sigma(T|S\backslash x).$$ This notion captures the idea that each \emph{consideration set} competes for the DM's attention.  An alternative property on attention rules  proposed by \citeAppendix{OA-wp2021attention} highlights the competition among each  alternative. Given an attention rule $\sigma$, define the attention probability map $\phi_\sigma: X\times (\mathcal{M}\backslash\{\emptyset\}) \rightarrow[0,1]$ such that for all $S \in \mathcal{M}\backslash \{\emptyset\}$ and $x \in S$,
$$\phi_\sigma(x|S) = \sum_{T\subseteq S:x\in T} \sigma (T|S).$$  The attention rule $\sigma$ is said to satisfy \textit{attention overload} if for any $x\in T\subseteq S$, $$\phi_\sigma(x|S)\leq \phi_\sigma(x|T).$$

Fix a tuple $(\pi,\mathcal{N})$. We first show that $\sigma_{\pi,\mathcal{N}}$ satisfies monotone attention. To see this, note that for each $T\subseteq S\neq \emptyset$, either $\sigma_{\pi,\mathcal{N}}(T|S)=0$, in which case removing some $x\in S\backslash T$ from $S$ can only weakly boost the probability that $T$ is finally considered, or $\sigma_{\pi,\mathcal{N}}(T|S)>0$, in which case Property 1 implies that removing some $x\in S\backslash T$ from $S$ strictly boosts the probability that $T$ is finally considered by a factor of $\frac{1}{\mathring{\pi}_x}$. Hence attention monotonicity holds for $\sigma_{\pi,\mathcal{N}}$.

By contrast, the following simple example shows that mental association generically leads to violations of attention overload. Consider $T=\{x\}$, $S=\{x,y\}$, and $(y,x)\in \mathcal{N}$. Then $\phi_\sigma(x|S) = \pi_x +\pi_y(1-\pi_x) > \phi_\sigma(x|T) = \pi_x$. The intuition is that larger menus may contain alternatives whose consideration further prompts attention toward $x$, leading to a higher attention probability of $x$.

\subsection*{OA-4. Random Associative Network} 
In this section, we study the extension in which the DM's associative network can be random. Allowing for randomness in the associative network captures the idea that associative links need not be stable over time: The links that are activated can depend on transient cues (e.g., context, framing, or recent exposure) or idiosyncratic shocks that affect which associations are retrieved at the moment of choice. This extension also accommodates environments in which consumers have heterogeneous association procedures, or firms can influence associations only probabilistically---for instance, because advertising exposure and recall are noisy, or because competing firms exert opposing forces on the same link (e.g., one firm seeks to strengthen it while another seeks to weaken it).

Let $\mathscr{N}$ be the set of all possible associative networks over $X$. A \textit{random associative network} is a probability distribution $\mu$ over $\mathscr{N}$. We say that a choice rule $\rho$ is a \textbf{R}andom \textbf{A}ssociation \textbf{B}ased  \textbf{C}onsideration rule (RABC) if there exists a tuple $(\pi, \mu, \succ)$, where $\pi$ is an attention probability function, $\mu$ is a random associative network, and $\succ$ is a preference ordering, such that for all  $(A,B) \in \mathcal{E}$ and $x \in A$, \begin{equation}\label{eq:def:RABC}
\rho(x|A,B) =  \sum_{\mathcal{N} \in \mathscr{N}} \mu(\mathcal{N}) \left( \sum_{C \subseteq AB:\, x=\max (\mathcal{N}_{\!\scriptscriptstyle AB}^+(C) \cap A; \,\succ) }  \pi_{C} \mathring{\pi}_{(AB)\backslash C} \right). 
\end{equation} The choice rule $\rho$ is also said to be represented by $(\pi, \mu, \succ)$ as a RABC.


For a given RABC $\rho$, the unique identification of the attention probability function $\pi$ and the preference $\succ$ proceeds exactly as in our baseline ABC model. However, the random associative network $\mu$ need not be uniquely identified. The following example illustrates that, holding $\pi$ and $\succ$ fixed, two distinct random associative networks can generate the same random choice rule.

\begin{example} \label{eg:random_network}
\emph{Let $X=\{x,y\}$.  Consider the following choice rule $\rho$.   
\begin{equation*}
\begin{split}
	\rho(y|\{x,y\}, \emptyset)=\frac{1}{8},~ & \rho(x|\{x,y\}, \emptyset)=\rho(x|\{x\}, \{y\})= \rho(y|\{y\}, \{x\})=\frac{5}{8}, \\ & \rho(x|\{x\},\emptyset)=\rho(y|\{y\},\emptyset)=\frac{1}{2}.
\end{split}
\end{equation*} It can be revealed that DM's preference  $\succ$ satisfies $x \succ y$, and her attention probability function is given by $\pi_x=\pi_y=\frac{1}{2}$. However, the random choice rule is consistent with more than one random associative networks: Any random associative network $\mu$ can represent it (together with $\pi$ and $\succ$) as a RABC if the probability of $y$ being associated with $x$ and that of $x$ being associated with $y$ are both equal to $\frac{1}{2}$ under $\mu$. Thus, with the four  associative networks $\mathcal{N}=\mathcal{X},$ $\mathcal{W}=\mathcal{X} \cup \{(x,y)\}$,  $\mathcal{V}=\mathcal{X} \cup \{(y,x)\}$, and $\mathcal{U}=\mathcal{X} \cup \{(x,y), (y,x)\}$, we can construct two random associative networks $\mu_1$ and $\mu_2$ such that $\mu_1(\mathcal{N})=\mu_1(\mathcal{U})=\frac{1}{2}$ and $\mu_2(\mathcal{W})=\mu_2(\mathcal{V})=\frac{1}{2}$, and both represent $\rho$ as a RABC.  \qed}
\end{example}

Although it is an intriguing open question whether there exists a canonical class of random associative networks such that each RABC $\rho$ admits a \emph{unique} representation with a random associative network from that class, addressing this question is beyond the scope of the present paper and is left for future research. We conclude this section by focusing on a special class of random associative networks that we find particularly relevant. We show that within this class the representation is uniquely identified.

A random associative network $\mu$ is said to be \textit{link-independent} if there is a link-formation probability function (LPF) $\theta: X^2 \rightarrow [0,1]$ such that (i) for all $z \in X$, $\theta(z,z)=1$,  and (ii) for all $\mathcal{N} \in \mathscr{N}$, $$\mu(\mathcal{N})=\left(\prod_{(x,y) \in \mathcal{N}} \theta(x,y) \right) \left(\prod_{(x,y) \not\in \mathcal{N}} (1-\theta(x,y)) \right)$$  The value of $\theta(x,y)$ is the probability that the DM can associate $y$ with $x$. Condition (ii) indicates that the DM forms each associative link independently.

Link-independent random associative networks provide a parsimonious and tractable way to model stochastic association formation. They capture the idea that, while the DM’s ability to retrieve a particular association $(x,y)$ may be noisy---because of limited memory, fluctuating attention, or recent exposure---these retrieval events only depend on the intrinsic attributes (such as similarity, thematic connection, or prompted exposure) of $x$ and $y$ and   each associative link is not predetermined by the presence of other connections. 
The LPF $\theta$ then summarizes, in reduced form, the ``strength'' or ``salience'' of each potential association: Higher $\theta(x,y)$ means that $y$ is more likely to come to mind when $x$ is encountered. This specification is also empirically convenient, as it replaces an unrestricted distribution over networks with a low-dimensional object $\theta$, and it provides a natural benchmark against which richer models with correlated link formation (e.g., clustering or common shocks) can be compared.
Therefore, we view this special case as a natural starting point for capturing heterogeneity in individuals' associative networks and
studying the joint effects of random association and random attention. 

This framework parallels the well-known Erd\H{o}s--R\'enyi random network model, where each undirected edge between any two nodes is formed independently with a fixed probability $p$. Our setting generalizes this by allowing directed links as well as heterogeneous link probabilities: Depending on relevant contextual or behavioral factors, each link has a specific probability to be activated. 

Below, we show that the LPF $\theta$ is uniquely identified. Consider a choice rule $\rho$ that can be represented by $(\pi,\mu,\succ)$ as a RABC, where $\mu$ is link-independent. To identify the associated LPF $\theta$, fix two distinct alternatives $x$ and $y$. We have $$\Phi_{\rho}(\{y\}, \{x\})= (1-\pi_y)(1-\pi_x \theta(x,y)),$$ i.e., the probability for $y$  being unselected in the menu $(\{y\}, \{x\})$ is equal to the probability that $y$ is not initially paid attention to and not considered through the initial consideration of $x$. Thus, the probability for $y$ being associated with $x$ is $$\theta(x,y)=\frac{1}{\pi_x}-\frac{\Phi_{\rho}(\{y\}, \{x\})}{\pi_x-\pi_x\pi_y}.$$

\subsection*{OA-5. General Models of Initial Attention} 
In this section, we study two relaxations of our assumption regarding how the DM's initial attention set is formed. First, we consider more general initial attention distributions by relaxing the assumption of independent attention. Second, we maintain the assumption of independent attention but allow the DM's attention probability for an alternative to depend on its availability. \bigskip

\noindent \textbf{General attention distributions.} Consider a general attention distribution function $\sigma: \mathcal{M} \times \mathcal{M} \rightarrow [0,1]$ such that for all $A \in \mathcal{M}$,  $\sum_{B \subseteq A} \sigma(B,A)=1$, and $\sigma(B,A) > 0$ if and only if $B \subseteq A$. To interpret, $\sigma(B,A)$ is the probability that the DM's initial consideration set is $B$ when  $A$ is the set of all observable alternatives. With $\sigma$, if the DM's associative network and preference ordering are given by $\mathcal{N}$ and $\succ$ respectively, then for all menu $(A,B)$ and $x \in A$, we have $$ \rho(x|A,B) = \sum_{C\subseteq AB:\, x = \max(\mathcal{N}^+_{\scriptscriptstyle\! AB}(C) \cap A;\, \succ)} \sigma(C,AB).$$

In what follows, we demonstrate that both $\mathcal{N}$ and $\succ$ can be uniquely identified. 		
To see this, consider two distinct alternatives $x$ and $y$. The identification of $\mathcal{N}$ is exactly the same as that in our baseline model.  If $x\mathcal{N}y$, then   $$\Phi_{\rho}(\{y\}, \{x\}) = \sigma(\emptyset, \{x,y\}) = \Phi_{\rho}(\{x,y\}, \emptyset),$$ and if not $x\mathcal{N}y$, then $$\Phi_{\rho}(\{y\}, \{x\}) = \sigma(\emptyset, \{x,y\}) + \sigma(\{x\}, \{x,y\}) > \Phi_{\rho}(\{x,y\}, \emptyset).$$ Therefore, $x\mathcal{N}y$ if and only if $\Phi_{\rho}(\{y\}, \{x\}) =   \Phi_{\rho}(\{x,y\}, \emptyset)$.  For the identification of the preference ordering, note that if $x \succ y$, then we have \begin{equation*}
	\begin{split}
		& \rho(x|\{x,y\}, \emptyset) = \rho(x|\{x\}, \{y\}), \text{and}\\
		& \rho(y|\{x,y\},\emptyset) \le \sigma(\{y\}, \{x,y\}) \\ < \,& \sigma(\{y\}, \{x,y\}) + \sigma(\{x, y\}, \{x,y\}) \le \rho(y|\{y\},\{x\}).
	\end{split}
\end{equation*} Therefore, $x \succ y$ if and only if $\rho(x|\{x,y\}, \emptyset) = \rho(x|\{x\}, \{y\})$. \bigskip

While the general attention distribution may be hard to identify, certain parametric assumptions on $\sigma$ can lead to a unique identification. For instance, consider the attention distribution $\sigma$ introduced by \citeAppendix{OA-ecta2016feasibility}: There is a function $\zeta: \mathcal{M} \rightarrow (0,1)$ such that for all $B \subseteq A$, $\sigma(B,A)=\frac{\zeta(B)}{\sum_{C\subseteq A} \zeta(C)}$. \citeAppendix{OA-ecta2016feasibility} show that such an attention distribution generalizes the independent attention distribution in MM14. In fact, for all $A\subseteq X$, $\frac{\zeta(A)}{\zeta(\emptyset)}$ can be pinned down inductively via $\Phi_{\rho}(A,\emptyset)$.\footnote{Specifically, suppose that for all $C \subsetneq A$, the ratio $\frac{\zeta(C)}{\zeta(\emptyset)}$ is already pinned down. Then we have $\frac{\zeta(A)}{\zeta(\emptyset)} = \frac{\sum_{B \subseteq A: B \neq \emptyset} \zeta(B)}{\zeta(\emptyset)}-\frac{\sum_{C \subsetneq A: C \neq \emptyset} \zeta(C)}{\zeta(\emptyset)} = \frac{1-\Phi_{\rho}(A,\emptyset)}{\Phi_{\rho}(A,\emptyset)}-\frac{\sum_{C \subsetneq A: C \neq \emptyset} \zeta(C)}{\zeta(\emptyset)}.$} Therefore, $\zeta$ is unique up to  rescaling. \bigskip

\noindent \textbf{Availability-dependent attention.} Our baseline model assumes that the attention probability of an observable alternative is the same regardless of its availability. However, this assumption may not hold in certain contexts. For example, in some online shopping platforms, products that are sold out are explicitly labeled as ``out of stock'' on the display page. It is plausible that such labeling may result in excessive or less attention from the consumer. Therefore, a natural extension of our baseline model is to incorporate the availability of alternatives as a factor that influences the attention probability assigned to them.

Formally, let $\pi^{f}: X \rightarrow (0,1)$ be the attention probability function for available alternatives, and $\pi^{n}: X \rightarrow (0,1)$ be the attention probability function for unavailable but observable alternatives.\footnote{The definitions of $\pi_A^f$, $\pi_A^n$, $\mathring{\pi}_A^f$ and $\mathring{\pi}_A^n$ are similar to those of $\pi_A$ and $\mathring{\pi}_A$ in Section \ref{sec:model}.} Let $\mathcal{N}$ be the associative network and $\succ$ be the DM's preference ordering. For each  menu $(A,B)$ and $x \in A$, we have $$\rho(x|A,B)=\sum_{C\subseteq AB: \, x=\max(\mathcal{N}^+_{AB}(C) \cap A; \,\succ)} \pi^f_{C\cap A} \pi^n_{C\cap B} \mathring{\pi}^f_{A\backslash C}   \mathring{\pi}^n_{B\backslash C}.$$

We argue that all relevant parameters of the model above can be uniquely identified. First, the preference ordering $\succ$ can be identified similarly as the case in Section OA-1, and the associative network $\mathcal{N}$ can be identified similarly as in our baseline model. Second, the attention probability for each alternative $x$ when it is available is given by $\pi^f_x=\rho(x|\{x\},\emptyset)$. Finally, the attention probabilities for unavailable but observable alternatives can be identified through the extent to which they boost the choice probabilities of other alternatives. To see this, consider alternative $x$ and assume that there exists a distinct alternative $y$ such that $x\mathcal{N}y$. We have $\rho(y|\{y\}, \{x\})= 1-(1-\pi^n_x)(1-\pi_y^f)$, which implies $$\pi^n_x=1-\frac{1-\rho(y|\{y\}, \{x\})}{1-\pi_y^f}.$$ We note that $\pi^n$ cannot be fully identified: For a given alternative $x$, if every other alternative is not associated with it, then we are unable to identify $\pi^n_x$. Nevertheless, in this case, since the consideration of $x$ does not prompt the consideration of any other alternative, the value of $\pi^n_x$ is irrelevant.

\subsection*{OA-6. Identification without the Default Option}
In this section, we discuss the identification of the parameters of our model when there is no default option. Following \citeAppendix{OA-ecta2014consideration}, we assume that the DM will continue to pay attention to extra alternatives if her final consideration set does not contain any available alternative. Specifically, we say that a choice rule $\rho$ is a \textbf{A}ssociation \textbf{B}ased \textbf{C}onsideration with \textbf{N}o Default Option rule (ABCN) if there is tuple $(\pi, \mathcal{N}, \succ)$ such that for all menu $(A,B)$  with $A\neq \emptyset$ and $x \in A$, $$ \rho(x|A,B) = \frac{\hat{\rho}(x|A,B)}{1-\Phi_{\hat{\rho}}(A,B)},$$ where $\hat{\rho}$ is the ABC  that is represented by $(\pi, \mathcal{N}, \succ)$. Such a tuple $(\pi, \mathcal{N}, \succ)$ is also said to represent $\rho$ as an ABCN. Note that for an ABCN  $\rho$, we have for every menu $(A,B)$ with $A\neq \emptyset$, $\sum_{x \in A} \rho(x|A,B)=1$.

Similar to \citeAppendix{OA-ecta2014consideration}, the attention probability function $\pi$ \textit{cannot} be uniquely identified. To see this, let the space of alternatives be $X=\{x,y,z\}$. Consider two tuples $(\pi, \mathcal{N}, \succ)$ and $(\pi', \mathcal{W}, \succ')$ where $\mathcal{N}=\mathcal{W}=X\times X$, $x \succ y \succ z$, $x \succ' y \succ' z$,  $$ \pi(x)=\pi(y)=\pi(z)=\frac{1}{2}, ~\pi'(x)= \frac{1}{3}, \pi'(y)= \frac{1}{4},  \text{~and~}
\pi'(z)=\frac{1}{6}. $$ It can be easily verified that the two tuples represent the same ABCN.

To ensure unique identification of the DM's  associative network and preference, we impose a mild condition on the choice rule. Formally, we say that a choice rule $\rho$ is \textit{associatively non-dense} if for all distinct $x,y \in X$, there exists $z \in X\backslash \{x,y\}$ such that: \begin{equation}\label{eq:positive}
	\rho(z|\{x,z\}, \emptyset) > 0,   \rho(x|\{x,z\}, \emptyset) > 0,  \rho(z|\{y,z\}, \emptyset) > 0,  \rho(y|\{y,z\}) > 0,
\end{equation} 
\begin{equation}\label{eq:nonboost}
	\rho(z|\{x,z\}, \emptyset) \geq \rho(z|\{x,z\}, \{y\}) \text{, and } \rho(z|\{y,z\}, \emptyset) \geq \rho(z|\{y,z\}, \{x\}).
\end{equation} To understand this condition, consider an ABCN $\rho$ represented by $(\pi, \mathcal{N}, \succ)$. A sufficient condition for $\rho$ to be associatively non-dense is that for all distinct $x,y \in X$, there exists $z \in X\backslash\{x,y\}$  such that:

$$\{(x,z), (z,x), (y,z), (z,y)\} \cap \mathcal{N} = \emptyset.$$ A natural setting where this condition holds is when alternatives can be categorized into three or more distinct groups, and the DM forms associative links only among alternatives within the same category. Note that when there are sufficiently many alternatives and the DM faces cognitive constraints in forming associative links, the non-denseness condition is typically satisfied.

We show that for a given ABCN $\rho$ that is associatively non-dense, the associative network and preference can be uniquely identified. To see this, consider two distinct alternatives $x, y \in X$. We first demonstrate the identification of the preference $\succ$. If $\rho(x|\{x,y\}, \emptyset) = 0$, we must have $y \succ x$. If both  $\rho(x|\{x,y\}, \emptyset)$ and $\rho(y|\{x,y\}, \emptyset)$ are positive, then by the non-denseness condition, we can find a third alternative $z$ such that condition (\ref{eq:positive}) holds. Since $\rho$ is an ABCN, it can be shown that in the menu $(\{x,y,z\}, \emptyset)$, all the three alternatives are chosen with positive probabilities. Let $a,b,c \in \{x,y,z\}$ be such that $a \succ b \succ c$, then we have $$ \frac{\rho(b|\{b,c\}, \emptyset)}{\rho(b|\{a, b,c\}, \emptyset)} = \frac{\rho(c|\{b,c\}, \emptyset)}{\rho(c|\{a, b,c\}, \emptyset)},$$ 
$$ \frac{\rho(a|\{a,b\}, \emptyset)}{\rho(a|\{a, b,c\}, \emptyset)} = \frac{\rho(b|\{a,b\}, \emptyset)}{\rho(b|\{a, b,c\}, \emptyset)},$$ 
$$\frac{\rho(a|\{a,c\}, \emptyset)}{\rho(a|\{a,b,c\}, \emptyset)} < \frac{\rho(c|\{a,c\}, \emptyset)}{\rho(c|\{a, b,c\}, \emptyset)}.$$ This leads to a unique identification of the preference relation between $x$ and $y$.

For the associative network $\mathcal{N}$, by the non-denseness condition, consider some $z \in X\backslash \{x,y\}$ such that conditions (\ref{eq:positive}) and (\ref{eq:nonboost}) hold. It then follows that $y\mathcal{N}x$ if and only if $\rho(x|\{x,z\},\emptyset) < \rho(x|\{x,z\}, \{y\})$, and $x\mathcal{N}y$ if and only if $\rho(y|\{y,z\},\emptyset) < \rho(y|\{y,z\}, \{x\})$. By this, $\mathcal{N}$ is uniquely identified.

\subsection*{OA-7. Platform Network Design}
In this section, we concretize the discussion in Section \ref{subsec:moreapp} on platform network design.
Consider a multi-product firm or a platform that sells a set of products $X$. We assume that all products are always observable, but some products are sometimes out of stock. With this assumption, we consider a probability distribution $\kappa$ over menus of the form $(A, X\backslash A)$, and let $\kappa_A$ denote the probability of menu $(A,X\backslash A)$. The firm desires to maximize the sales volume of a particular alternative $x^*\in X$. This analysis is particularly pertinent to modern digital platforms---such as Amazon---that serve the dual role of matching third-party buyers and sellers while simultaneously selling their own private-label offerings. We assume that the demand faced by the firm is captured by the RUMABC $\rho$ represented by $(\pi, \mathcal{N}, \tau)$. That is, when the menu is $(A,B)$, the sales volume of good $x$ is given by $\rho(x|A,B)$. 

We analyze a scenario in which the firm seeks to strategically augment the existing network, $\mathcal{N}$, by adding a single associative link to maximize the choice probability of the target product. We briefly discuss the feasibility of adding such links and the rationale behind limiting it to a single additional link. This added link could be literal or metaphorical. For instance, platforms like Amazon might create an association from $x$ to $y$ by recommending $y$ on the product page of $x$. Alternatively, firms could run advertisements that strengthen consumers' mental association from $x$ to $y$. However, the firm faces constraints, such as limited space on product pages or the high costs of advertising, which prevent it from adding multiple links indiscriminately.

While our analysis below focuses on the optimal design of an associative network when only one additional link is to be added,  our framework  allows us to explore more general problems beyond this baseline application, such as the optimal associative network when multiple links are allowed to be added or when the firm is concerned with maximizing the total profits across all products. While all these applications are of potential interest, to formally deal with them is beyond the scope of the current paper, and we leave them for future research.

We define the notation used in this section. Let $\overleftarrow{\mathcal{N}}$ be the inverse of $\mathcal{N}$ such that $x\mathcal{N}y$ if and only if $y\overleftarrow{\mathcal{N}}x$. Let $\overleftarrow{\mathcal{N}}^+$ be the transitive closure of $\overleftarrow{\mathcal{N}}$. For a given menu $(A,X\backslash A)$ and preference ordering $\succ$, let $\bar{A}_{\succ}=\{y \in X: y\mathcal{N}^+x^*  \text{~or~} y\mathcal{N}^+z \text{~for some~} z \in A \text{~with~} z \succ x^*\}$. In words, $\bar{A}_{\succ}$ contains alternatives such that the attention to any of them prompts the consideration of an available alternative that is weakly better than $x^*$ in the menu $(A, X\backslash A)$ under the preference $\succ$. For each preference ordering $\succ$, denote by $\rho_{\!\scriptscriptstyle\succ}$ the ABC represented by $(\pi, \mathcal{N} ,\succ)$. The following proposition characterizes the optimal link to be added.

\begin{proposition} \label{prop:application_ABC}
The associative link $(y,x^*)$ that satisfies the following condition is the link to be added that maximally increases the sales volume of $x^*$:
$$y \in \arg\max_{z \in X}   \left( \sum_{\succ \in \mathscr{P}} \tau(\succ) \left(\sum_{A \subseteq X:\, \rho_{\scriptscriptstyle\succ}(x^*|A,X\backslash A) > 0} \kappa_{\!\scriptscriptstyle A} \mathring{\pi}_{\!\scriptscriptstyle\bar{A}_{\succ}}  \left(1-\mathring{\pi}_{\!\scriptscriptstyle\overleftarrow{\mathcal{N}}^+(z)\backslash \bar{A}_{\succ}}\right) \right) \right) \!.$$ 
\end{proposition} 

We defer the proof of the proposition to the end of this section. To simplify the illustration of Proposition \ref{prop:application_ABC}, consider the case in which $\tau$ is degenerate and assigns probability $1$ to a single preference ordering $\succ$.  For a given menu $(A,X\backslash A)$ if $\rho_{\scriptscriptstyle\succ}(x^*|A,X\backslash A)=0$, then $x^*$ is either unavailable or prompts the attention to some better alternative in $A$. In this case, $x^*$ remains unchosen no matter what link we add. Hence, we restrict attention to menus $(A,X\backslash A)$ such that $\rho_{\scriptscriptstyle\succ}(x^*|A,X\backslash A)> 0$.

For such a menu, the consideration of any alternative in $\bar{A}_{\succ}$ prompts the consideration of $x^*$ or some better  available alternative. In either case, adding one more link does not affect the choice probability of $x^*$. By contrast, if the set $\overleftarrow{\mathcal{N}}^+(z)\backslash \bar{A}_{\succ}$ is enlarged, the link is more likely to boost the consideration and the choice of $x^*$ through the intermediate alternative $z$. Thus, the value of $\kappa_{\!\scriptscriptstyle A} \mathring{\pi}_{\!\scriptscriptstyle\bar{A}_{\succ}}$ can be regarded as the marginal benefit of boosting the attention of $x^*$ in the menu $(A, X\backslash A)$. The size of the menu $\overleftarrow{\mathcal{N}}^+(y)\backslash \bar{A}_{\succ}$ can be interpreted as the additional connectedness of $x^*$ in the menu $(A, X\backslash A)$ brought by the new link, and the value of $1-\mathring{\pi}_{\!\scriptscriptstyle\overleftarrow{\mathcal{N}}^+(z)\backslash \bar{A}_{\succ}}$ captures how the attention of $x^*$ can be boosted by the new link.  Accordingly, Proposition \ref{prop:application_ABC} states that the optimal link is the one that maximizes the  $\tau$-weighted average of the product of this marginal benefit and the resulting attention boost.

Notably, if the distribution of menus is not exogenous  but can be determined by the firm, then the firm would choose the deterministic menu $(\{x^*\}, X\backslash x^*)$ to maximize the sales of product $x^*$. In this case, by Proposition \ref{prop:application_ABC}, the optimal link to be added is given by the following corollary.

\begin{corollary}\label{coro:specialcase}
Suppose that the firm can choose the distribution of the menus. To maximize the sales volume of $x^*$, the firm can optimally choose menu $(\{x^*\},X\backslash x^*)$ and add a link $(y,x^*)$ such that $y \in \argmin_{z \in X} \mathring{\pi}_{\scriptscriptstyle \overleftarrow{\mathcal{N}}^{+}(z)\backslash \overleftarrow{\mathcal{N}}^{+}(x^*)}.$ In particular, if $\pi(\cdot) \equiv \alpha \in (0,1)$, then the link    $(y,x^*)$ satisfies  $y \in \argmax_{z\in X} |\overleftarrow{\mathcal{N}}^{+}(z)\backslash \overleftarrow{\mathcal{N}}^{+}(x^*)|.$
\end{corollary}

By Corollary \ref{coro:specialcase}, when each product has the same chance to attract the attention of the consumers,  the connectedness of alternative $x^*$ is given by $\overleftarrow{\mathcal{N}}^{+}(x^*)$. The link we add is simply the one that increases the connectedness of $x^*$ the most. 

\bigskip

\begin{proof}[Proof of Proposition \ref{prop:application_ABC}]
It suffices to prove the degenerate case where $\tau(\succ) = 1$ for some preference $\succ$, as the objective is simply a $\tau$-weighted average over each preference in the support of $\tau$. In this case, we have $\rho=\rho_{\succ}$. Write $\bar{A}$ for $\bar{A}_{\succ}$ for simplicity.

Let $\hat{\rho}(x^*|A, X\backslash A)$  be the choice probability of $x^*$ in the menu $(A,X\backslash A)$ after the link $(y,x^*)$ is added. 	It remains to show that when $\rho(x^*|A,X\backslash A)>0$, we have $\hat{\rho}(x^*|A, X\backslash A)-\rho(x^*|A, X\backslash A)=\mathring{\pi}_{\!\scriptscriptstyle\bar{A}}  \left(1-\mathring{\pi}_{\!\scriptscriptstyle\overleftarrow{\mathcal{N}}^+(y)\backslash \bar{A}}\right)$.  To see this, consider a partition $\{B, C\}$ of $\bar{A}$ with $B=\{w \in X: \text{for some~} z \in A, z \succ x^* \text{~and~} w\mathcal{N}^+z\}$ and $C = \bar{A} \backslash B$. Let events 1 and 2 denote that no alternative in $B$ and $C$  is initially considered, respectively. If event 1 does not occur, then some alternative better than $x^*$ in $A$ will be considered and always blocks the choice $x^*$. If event 2 does not occur, then $x^*$ already appears in the final consideration set, and adding one more link cannot further boost the choice of $x^*$. Hence, the new link only affects the choice of $x^*$ when both events occur, of which the probability is $\mathring{\pi}_{\!\scriptscriptstyle \bar{A}}$. Conditional on the two events, 
the extra choice probability of $x^*$ by the new link equals the chance that some alternative prompts the attention to $x^*$ through the new link, i.e., $1-\mathring{\pi}_{\!\scriptscriptstyle\overleftarrow{\mathcal{N}}^+(y)\backslash \bar{A}}$. 	
\end{proof}




 \subsection*{OA-8. Proofs Omitted in Section \ref{sec:app}}

In this section, we provide detailed proofs for Propositions \ref{prop_branding} and \ref{prop:app:imitation}   and statements that appear in Section \ref{subsec:imitation}.  Throughout this section we focus on the pure-strategy equilibrium and call it an equilibrium for short.

\begin{proof}[Proof of Proposition \ref{prop_branding}]
If the firm chooses not to introduce product $H$, then it sells product $L$ at price $p_L$, generating profit $\Pi^0 = \pi_Lp_L$.

If the firm introduces product $H$ by brand extension, then it is optimal to sell product $L$ to type-2 consumers at price $p_L$ and product $H$ to type-1 consumers at price $p_H= v_H-(v_L-p_L)$. This leads to profit 
\[\Pi^E = \left(\pi_H+\pi_L-\pi_H\pi_L\right)(\alpha (v_H-v_L+p_L) + (1-\alpha)p_L) -C.\]

If the firm introduces product $H$ by sub-branding, it suffices to focus on the case where $p_H=v_H$ with profit 
\[\Pi^S = \alpha \pi_H(1-\pi_L)v_H + \pi_Lp_L-C.\]

The   statements in the proposition directly follow from comparing $\Pi^0, \Pi^E$, and $\Pi^S$.
\end{proof}

\medskip

\begin{proof}[Proof of Proposition \ref{prop:app:imitation}]
\textbf{Equilibrium  under normal competition.} 
When firm 2 chooses to enter with normal competition, the two firms choose prices $p_1$ and $p_2$ simultaneously. Equivalently, they choose the cutoff consumer types $\theta_1, \theta_2 \in [0,1]$ such that each type-$\theta_i$ ($i \in \{1,2\}$) consumer obtains is indifferent between product $i$ and the outside option. The corresponding prices are pinned down by $p_1=\theta_1$ and $p_2=\theta_2 q_2$. 

First, consider the case where $\theta_1 < \theta_2$. It follows that all consumers prefer product 1 to product 2. Firm 1 operates as a monopoly for all consumers who consider product 1, and firm 2 operates as a monopolist for all consumers who only considers product 2. It follows that each firm $i$ can strictly benefit from making their cutoff type $\theta_i$ closer to $\frac{1}{2}$, and thus there must be one firm desiring to deviate. This cannot be an equilibrium.

Next, consider the case where $\theta_1=\theta_2$. By the analysis above, a necessary condition for this to be an equilibrium is that $\theta_1=\theta_2=\frac{1}{2}$. We show that firm 2 can slightly lower its price   to be strictly better off. Its profit by choosing some $\theta_2$ lower but close to $\theta_1$ is $$(1-\pi_1) \pi_2 q_2 \theta_2(1-\theta_2)+\pi_1\pi_2q_2\theta_2\left( \frac{\theta_1-\theta_2}{1-q_2} \right).$$  
Firm 2's first order condition at $\theta_1=\theta_2=\frac{1}{2}$ is given by  \begin{equation*}
    \begin{split}
        & \pi_2 q_2 \left( (1-\pi_1)  (1-2\theta_2) + \pi_1  \left(  \frac{\theta_1-2\theta_2}{1-q_2} \right) \right)  =  - \frac{\pi_1\pi_2 q_2}{2-2q_2} < 0.
    \end{split}
\end{equation*}  Thus, this case cannot be an equilibrium.

We proceed to consider the case where $\theta_2 < \theta_1$. If $(1-\theta_2)q_2 > 1-\theta_1$, then all consumers strictly prefer product 2 to product 1. By an argument analogous to the case where $\theta_1<\theta_2$, no equilibrium exists. Similar logic also rules out the case where $(1-\theta_2)q_2 = 1-\theta_1$ with $\theta_2 > \frac{1}{2}$ or $\theta_1 < \frac{1}{2}$. If instead $(1-\theta_2)q_2 = 1-\theta_1$ with $\theta_2 \le \frac{1}{2} \le \theta_1$, then the existence of an equilibrium requires the following local first order condition of firm 1
 $$\pi_1 \left( (1-\pi_2) (1-2\theta_1) + \pi_2 \left( 1-\frac{2\theta_1-q_2\theta_2}{1-q_2} \right) \right) \ge 0,$$ which contradicts $\theta_1 \ge \frac{1}{2}$ and $\theta_2 < 1$.

The only remaining case is when $\theta_2 < \theta_1$ and $(1-\theta_2)q_2 < 1-\theta_1$. In this case, the profits of the two firms are given respectively by: 
$$\pi_1 (1-\pi_2) \theta_1 (1-\theta_1)+\pi_1\pi_2 \theta_1\left(1- \frac{\theta_1-q_2\theta_2}{1-q_2} \right).$$ 
$$(1-\pi_1) \pi_2 q_2 \theta_2(1-\theta_2)+\pi_1\pi_2q_2\theta_2\left( \frac{\theta_1-\theta_2}{1-q_2} \right).$$
A pair $(\theta^*_1, \theta^*_2)$  that constitutes an equilibrium must satisfy the first-order conditions:  \begin{equation}\label{eq:app:foc1}
    \pi_1 \left( (1-\pi_2) (1-2\theta_1^*) + \pi_2 \left( 1-\frac{2\theta_1^*-q_2\theta_2^*}{1-q_2} \right) \right) = 0,
\end{equation}
\begin{equation}\label{eq:app:foc2}
\pi_2 q_2 \left( (1-\pi_1)  (1-2\theta_2^*) + \pi_1  \left(  \frac{\theta_1^*-2\theta_2^*}{1-q_2} \right) \right) = 0.    
\end{equation} Let $r=1-q_2$, and the two equations above yield \begin{equation*} 
    \theta_1^*  = \frac{2r^2+2\pi_1  q_2r +\pi_2 q_2 r-\pi_1\pi_2 q_2r}{4r^2+4\pi_1 q_2r+4\pi_2 q_2r + 4\pi_1\pi_2 q_2^2 -\pi_1\pi_2 q_2},
\end{equation*}
\begin{equation*} 
\theta_2^*  = \frac{2r^2-\pi_1 r+2\pi_1 q_2r +2\pi_2 q_2 r - 2\pi_1\pi_2 q_2 r}{4r^2+4\pi_1 q_2 r+4\pi_2 q_2 r + 4\pi_1\pi_2 q_2^2 -\pi_1\pi_2 q_2}.
\end{equation*}
It can be easily verified that $\theta_1^*, \theta_2^* \in (0, \frac{1}{2})$ in this case. The two conditions $\theta_2^* < \theta_1^*$ and $(1-\theta_2^*)q_2 < 1-\theta_1^*$ are equivalent to \begin{equation}\label{eq:app:eqexist1}
\frac{\pi_1}{1-\pi_1} > \pi_2 q_2,
\end{equation}
\begin{equation}\label{eq:app:eqexist2}
2 - q_2(4 - \pi_1 - 3\pi_2) + 2q_2^2(1-\pi_1)(1-\pi_2) > 0.
\end{equation} Note that the  profits of the two firms at $(\theta_1^*, \theta_2^*)$ are given by $$ \Pi^N_1 (\theta_1^*, \theta_2^*) =\left( \pi_1(1-\pi_2)+\frac{\pi_1\pi_2}{1-q_2} \right) {\theta^*_1}^2,$$
$$ \Pi^N_2 (\theta_1^*, \theta_2^*) =\frac{\pi_2 q_2 (1 - q_2(1-\pi_1))}{1-q_2} {\theta_2^*}^2.$$
To ensure $(\theta_1^*, \theta_2^*)$ to be an equilibrium, we need further ensure that no firm $i$ wants to deviate to  $\theta_i =\frac{1}{2}$ (i.e., the monopoly price). This leads to  two additional conditions: 
\begin{equation}\label{eq:app:eqexist3}
\left( \pi_1(1-\pi_2)+\frac{\pi_1\pi_2}{1-q_2} \right) {\theta^*_1}^2 \ge \Pi_1^n(\frac{1}{2}, \theta^*_2),
\end{equation}
\begin{equation}\label{eq:app:eqexist4}
\frac{\pi_2 q_2 (1 - q_2(1-\pi_1))}{1-q_2} {\theta^*_2}^2 \ge  \frac{(1-\pi_1)\pi_2 q_2}{4}, 
\end{equation} where $\Pi_1^n(\frac{1}{2}, \theta^*_2)$ denotes the profit of firm 1 under  normal competition when it sets the cutoff $\theta_1$ as $\frac{1}{2}$ and firm 2 sets the cutoff $\theta_2$ as $\theta_2^*$. It follows that conditions (\ref{eq:app:eqexist1})-(\ref{eq:app:eqexist4}) are sufficient and necessary to ensure the existence (and thus uniqueness) of the pure-strategy equilibrium under normal competition.

Note that given $\pi_1$ and $\pi_2$, when $q_2$ is small enough, conditions (\ref{eq:app:eqexist1}) and (\ref{eq:app:eqexist2}) are satisfied. When $q_2 \le \frac{1}{2}$, we have $(1-\theta^*_2)q_2 < \frac{1}{2}$, meaning that even if firm 1 deviates to set the cutoff $\theta_1 = \frac{1}{2}$, we still have $(1-\theta_2^*)q_2<1-\theta_1$, and condition (\ref{eq:app:foc1}) guarantees   condition (\ref{eq:app:eqexist3}). Finally, when $q_2$ is sufficiently small, the left-hand side of condition (\ref{eq:app:eqexist4}) divided by its right-hand side converges to $\frac{(2-\pi_1)^2}{4-4\pi_1} > 1$. Hence, it holds when $q_2$ is small.

\medskip

\noindent \textbf{Equilibrium under imitation.} 
Under imitation, a similar analysis  implies that the unique equilibrium $(\theta_1^{**}, \theta_2^{**})$ can be pinned down by the following first-order conditions: 
\begin{equation*}
(\pi_1+\pi_2-\pi_1\pi_2) \left( 1-\frac{2\theta_1^{**}-\bar q_2\theta_2^{**}}{1-\bar q_2} \right)  = 0, 
\end{equation*}
\begin{equation*}
(\pi_1+\pi_2-\pi_1\pi_2) \bar q_2   \left(  \frac{\theta_1^{**}-2\theta_2^{**}}{1-\bar q_2} \right)   = 0.
\end{equation*} It follows that 
\begin{equation*}
\theta_1^{**}= \frac{2(1-\bar q_2)}{4-\bar q_2}  \text{\,~and~\,} \theta_2^{**}=\frac{1-\bar q_2}{4-\bar q_2}.
\end{equation*} The equilibrium profits of the two firms are given respectively by \begin{equation}
   \Pi_1^{I}(\theta_1^{**}, \theta_2^{**}) = (\pi_1+\pi_2-\pi_1\pi_2)\frac{4(1-\bar q_2)}{(4-\bar q_2)^2},
\end{equation} 
\begin{equation} 
\Pi_2^{I}(\theta_1^{**}, \theta_2^{**}) =  (\pi_1+\pi_2-\pi_1\pi_2)\frac{\bar q_2(1-\bar q_2)}{(4-\bar q_2)^2}. 
\end{equation}

Now we can finish the proof of Proposition \ref{prop:app:imitation}.  First,  the monopoly profit of firm 1 is $\pi_1/4$. Thus, it is lower than firm 1's equilibrium payoff under imitation if and only if condition (\ref{eq:app:xiaomi1}) holds.

Second, when $q_2=\bar q_2$, the profits of firm 2 under normal competition and imitation are approximated by $(2-\pi_1)^2\pi_2q_2/16$ and $(\pi_1+\pi_2-\pi_1\pi_2)q_2/16$, respectively.  The comparison reduces to the comparison between $(2-\pi_1)^2\pi_2$ and  $\pi_1+\pi_2-\pi_1\pi_2$, and we have   $(2-\pi_1)^2\pi_2 >  \pi_1+\pi_2-\pi_1\pi_2$ if and only if $J(\pi_1, \pi_2)>0$. Thus, when $q_2 = \bar{q}_2$ is small, firm 2 prefers normal competition   if $J(\pi_1, \pi_2)>0$ and prefers imitation otherwise.   The conclusion of Proposition \ref{prop:app:imitation} follows from the fact that $\Pi_2^{I}(\theta_1^{**}, \theta_2^{**})$ is first increasing then decreasing to zero with respect to $\bar{q}_2$. Notably, when $J(\pi_1, \pi_2)>0$ and $q_2$ is sufficiently small enough, the profit of normal competition is close to $0$ and hence there must exist some intermediate value of $\bar q_2$ such that firm 2 strictly prefers imitation.
\end{proof}

\medskip

\noindent \textbf{Numerical example with intermediate values of $q_2$.} Let $q_2=\bar{q}_2=0.29$, $\pi_1=0.48$, and $\pi_2=0.2$. Under normal competition, we have $$\frac{\pi_1}{1-\pi_1} \approx 0.92 > 0.058 = \pi_2 q_2,$$  $$2 - q_2(4 - \pi_1 - 3\pi_2) + 2q_2^2(1-\pi_1)(1-\pi_2) \approx 1.22 > 0,$$ $$ \theta_1^*\approx 0.48,~  \theta_2^*\approx 0.35, \text{~and}$$  $$\Pi_{2}^{N}(\theta_1^*, \theta_2^*) \approx 0.0086 > 0.0075 \approx \frac{(1-\pi_1)\pi_2 q_2}{4}.$$  Thus, conditions (\ref{eq:app:eqexist1})-(\ref{eq:app:eqexist4}) hold, where condition (\ref{eq:app:eqexist3}) holds because $q_2 \le \frac{1}{2}$.  Under imitation, we have 
$$ \theta_1^{**}\approx 0.3827, ~ \theta_2^{**}\approx 0.1914,$$  $$\Pi_1^{I}(\theta_1^{**},\theta_2^{**}) \approx 0.121 > 0.12 = \frac{\pi_1}{4}, \,\text{and}~ \Pi_2^{I}(\theta_1^{**},\theta_2^{**})\approx 0.0087 > \Pi_{2}^{N}(\theta_1^*, \theta_2^*).$$ This verifies that firm 2 prefers imitation and firm 1 prefers to accommodate.

\bibliographystyleAppendix{ecta}
\setstretch{0.85}
\bibliographyAppendix{association-online}	
\end{document}